%% file: ncrank.tex
\documentclass[11pt]{article}

\input{mymath}

\usepackage[margin=0.75in]{geometry}

\usepackage{verbatim}
\usepackage{graphicx}
\usepackage{epsfig}
\usepackage{float}
\usepackage[all]{xy}
\usepackage{color}

\usepackage{fullpage}
\usepackage{appendix}
\usepackage{mathtools}
\usepackage{cases}

\floatstyle{boxed}
\newfloat{Protocol}{h}{pro}

\floatstyle{boxed}
\newfloat{Algorithm}{h}{pro}

\usepackage{hyperref}

\hypersetup{
    colorlinks,
    citecolor=blue,
    filecolor=blue,
    linkcolor=blue,
    urlcolor=blue
}

\begin{document}

\title{Operator scaling: theory and applications}

\author{
Ankit Garg \thanks{Microsoft Research New England, email: garga@microsoft.com. This research was done when the author was a student at Princeton University and his research was partially supported by Mark Braverman's NSF grant CCF-1149888, Simons Collaboration on Algorithms and Geometry, Simons Fellowship in Theoretical Computer Science and a Siebel Scholarship.}
\and
Leonid Gurvits\thanks{Department of Computer Science, The City College of New York, email: l.n.gurvits@gmail.com}
\and
Rafael Oliveira \thanks{Department of Computer Science, Princeton University, email: rmo@cs.princeton.edu. Research supported by NSF Career award (1451191) and by CCF-1523816 award.}
\and
Avi Wigderson \thanks{Institute for Advanced Study, Princeton, email: avi@math.ias.edu. This research was partially supported by NSF grant
CCF-1412958.}
}

\maketitle

\begin{abstract}

\input{abstract.tex}

\end{abstract}

\tableofcontents

\newpage 

\section{Introduction}

\input{Intro.tex}

\section{Quantum Operators and Analysis of Algorithm G}\label{mainsection}

\input{AlgG}

\section{Properties of Capacity}\label{sec:cap_properties}

\input{cap_properties}

\section{Bit Complexity Analysis of Algorithm $G$ and Continuity of Capacity}\label{sec:bit_complexity_continuity}

\input{bit_and_continuity}

\section{Computing the Capacity of a Quantum Operator}\label{compute-capacity}

\input{comp-capacity}



\section{Conclusion and Open Problems}\label{open}

\input{conclusion.tex}

\subsection*{Acknowledgments} 
We would like to thank Harm Derksen, Pavel Hrubes, Louis Rowen and K. V. Subrahmanyam for helpful discussions. We would also like to thank Oded Regev for suggesting us that operator scaling could be used for approximating capacity. Finally, we thank the anonymous reviewers for a comprehensive reading of the paper and pointing several typographical errors and minor bugs.

\bibliographystyle{alpha}
\bibliography{refs}

\appendix

\section{Symbolic matrices with polynomial entries and non-commutative rank}\label{computeNC}

\input{comp-ncrank}

\end{document}

%% file: mymath.tex
\usepackage{amsthm, amssymb, amsmath, amsfonts, url, enumerate}

\newcommand{\dst}{\displaystyle}

\newtheorem{theorem}{Theorem}[section]

\newtheorem{claim}[theorem]{Claim}
\newtheorem{proposition}[theorem]{Proposition}

\newtheorem{corollary}[theorem]{Corollary}
\newtheorem{lemma}[theorem]{Lemma}

\newtheorem{example}[theorem]{Example}
\newtheorem{fact}[theorem]{Fact}
\newtheorem{remark}[theorem]{Remark}
\theoremstyle{definition}
\newtheorem{definition}[theorem]{Definition}

\newenvironment{prf}{\noindent{\bf Proof.~}}{\(\qed\)}
\newcommand{\BPF}{\begin{prf}} 
\newcommand {\EPF}{\end{prf}}

\def\mult{\textsf{Mult}}

\newcommand{\capac}{\textnormal{cap}}
\newcommand{\trn}{\textnormal{Trn}}
\newcommand{\Det}{\textnormal{Det}}
\newcommand{\Per}{\textnormal{Per}}
\newcommand{\ds}{\textnormal{ds}}
\newcommand{\tensor}{\otimes}
\newcommand{\expon}{\textnormal{exp}}
\newcommand{\QuantPer}{\textnormal{QPer}}
\newcommand{\sgn}{\textnormal{sgn}}

\newcommand{\tr}{\textnormal{tr}}
\newcommand{\rk}{\textnormal{rank}}

\def\poly{\textnormal{poly}}
\newcommand{\ncrk}{\textnormal{nc-rank}}
\def\ncrank{\textnormal{nc-rank}}

\def\F{{\mathbb{F}}}
\def\Q{{\mathbb{Q}}}
\def\Z{{\mathbb{Z}}}

\def\R{{\mathbb{R}}}

\def\C{{\mathbb{C}}}

\def\GL{\mathbb{GL}}
\def\FX{ \mathbb{F}\langle \bx \rangle }
\def\QX{ \mathbb{Q}\langle \bx \rangle }
\def\Fx{ \mathbb{F}\langle \bx \rangle }
\def\Fxx{\mathbb{F}\langC \bx \rangC}
\def\Qxx{\mathbb{Q}\langC \bx \rangC}

\def\cM{\mathcal M}

\def\cPS{\mathcal{PSPACE}}
\def\cP{\mathcal{P}}

\def\bf{{\mathbf f}}

\def\bx{{\mathbf x}}
\def\by{{\mathbf y}}

\def\bx{{\mathbf x}}
\def\by{{\mathbf y}}

\newcommand{\wt}[1]{\widetilde{#1}}

\def\then{\Rightarrow}

\def\to{\rightarrow}

\newcommand{\eps}{\epsilon}

\newcommand{\langC}{\begin{picture}(9,7)
\put(3.0,0){$($}
\put(1.1,2.5){\rotatebox{30}{\line(1,0){8.0}}}
\put(1.1,2.5){\rotatebox{330}{\line(1,0){8.0}}}
\end{picture}}
\newcommand{\rangC}{\begin{picture}(9,7)
\put(2.0,0){$)$}
\put(.8,2.5){\rotatebox{150}{\line(1,0){8.0}}}
\put(.8,2.5){\rotatebox{210}{\line(1,0){8.0}}}
\end{picture}}


%% file: abstract.tex
In this paper we present a deterministic polynomial time algorithm for testing if a symbolic matrix in {\em non-commuting} variables over $\Q$ 
is invertible or not. The analogous question for commuting variables is the celebrated polynomial identity testing (PIT) for symbolic determinants. 
In contrast to the commutative case, which has an efficient probabilistic algorithm, the best previous algorithm for the non-commutative 
setting required exponential time~\cite{IQS2015} (whether or not randomization is allowed). The algorithm efficiently solves the ``word problem'' 
for the free skew field, and the identity testing problem for arithmetic formulae with division over non-commuting variables, 
two problems which had only exponential-time algorithms prior to this work.

The main contribution of this paper is a complexity analysis of an existing
algorithm due to Gurvits~\cite{gurvits2004}, who proved it was polynomial time for 
certain classes of inputs.
We prove it always runs in polynomial time. The main component of our 
analysis is a simple (given the necessary known tools) lower bound on 
central notion of {\em capacity} of operators (introduced by Gurvits~\cite{gurvits2004}).
We extend the algorithm to actually approximate capacity to any accuracy 
in polynomial time, and use this analysis to give quantitative bounds on 
the continuity of capacity (the latter is used in a subsequent paper on 
Brascamp-Lieb inequalities).
We also extend the algorithm to compute not only singularity, but 
actually the (non-commutative) rank of a symbolic matrix, yielding a 
factor 2 approximation of the commutative rank. This naturally raises a 
relaxation of the
commutative PIT problem to achieving better deterministic approximation 
of the commutative
rank.

Symbolic matrices in non-commuting variables, and the related structural and algorithmic questions, have a remarkable number of diverse 
origins and motivations. They arise independently in (commutative) invariant theory and representation theory, linear algebra, optimization, 
linear system theory, quantum information theory, approximation of the permanent and naturally in non-commutative algebra. We provide a 
detailed account of some of these sources and their interconnections. In particular we explain how some of these sources 
played an important role in the development of Gurvits' algorithm and in our analysis of it here. 

%% file: Intro.tex
This introduction will be unusually long, due to the unusual number of settings in which the problems we study appear, and their 
interconnections and history. 

The main object of study in this paper are {\em symbolic} matrices (also called linear pencils) whose entries are linear forms in variables $\bx = \{x_1, x_2, \dots,  x_m\}$ 
over a field\footnote{Our main results will be for the rationals $\Q$ (and will hold for $\R$ and $\C$ as well) but not for finite fields. However many of the questions are interesting for any field.}  $\F$. Any such matrix can be expressed 
as a linear combination of the variables with matrix coefficients 
$$L = x_1 A_1 + x_2 A_2 + \dots + x_m A_m$$ 
where $A_1, A_2 \ldots, A_m$ are  $n \times n$ matrices\footnote{For all purposes we may assume that the matrices $A_i$ are linearly 
independent, namely span a space of matrices of dimension exactly $m$.} over $\F$. 

The main computational problem we will be concerned with in this paper (which we call SINGULAR) is determining whether 
such a symbolic matrix is invertible or not (over the field of fractions in the given variables). This problem has a dual life, depending on whether the variables commute or don't commute. In the {\em commutative} case this problem has an illustrious history and significance. It was first explicitly stated by Edmonds~\cite{Edm67}, and shown to have a randomized polynomial time algorithm by Lovasz~\cite{Lov79}.  The completeness of determinant for arithmetic formulas by Valiant~\cite{Val} means that singularity captures the celebrated Polynomial Identity Testing (PIT) problem, and so in the commutative setting we will refer to it as PIT. The derandomization of the latter probabilistic algorithm for PIT (namely, proving PIT $\in\cP$) became all-important overnight when Kabanets and Impagliazzo~\cite{KabImp} showed it would imply nontrivial arithmetic or Boolean lower bounds well beyond current reach. 

On the other hand, in the {\em non-commutative} case even the meaning of this problem SINGULAR is unclear. It took decades to fully define and understand the related notion of a ``field of fractions'' for non-commutative polynomials, namely the {\em free skew field} over which we (attempt to) invert the matrix\footnote{For now, the reader may think of the elements of this ``free skew  field"  simply as containing all expressions (formulas) built from the variables and constants using the arithmetic operations of addition, multiplication and division (we define it more formally a bit later). We note that while this is syntactically the same definition one can use for commuting variables, the skew field is vastly more complex, and in particular its elements cannot be described canonically as ratios of polynomials.}. But as we will explain below, this non-commutative SINGULAR problem has many intuitive and clean equivalent formulations (some entirely commutative!). It captures a non-commutative version of identity testing for polynomials and rational functions, provides a possible avenue to attack the notorious commutative PIT version, and quite surprisingly, its different formulations arise naturally in diverse areas of mathematics, revealing surprising connections between them. We only note that unlike the commutative PIT, it is not even clear from its definition
that the problem SINGULAR  is decidable. It requires some of the very nontrivial characterizations above, and other important results in commutative algebra, to prove a deterministic  {\em exponential time} upper bound on its complexity (see two very different  proofs in~\cite{CR99,IQS2015}), the best known before this work. No better bound was known even allowing randomness. 

The main result of this paper is a {\em deterministic} polynomial time algorithm for this problem!

\begin{theorem}\label{main}
For non-commutative variables over $\Q$, {\em SINGULAR} $\in \cP$. \\
More specifically, there is a deterministic algorithm which, given $m$ $n\times n$ integer matrices 
$A_1,\dots A_m$ with entries of bit-size $b$, decides in time $\emph{\poly}(n,m,b)$ if the matrix 
$L = \sum_{i=1}^m x_i A_i$ is invertible over the free skew field.
\end{theorem}

We now formulate the problem SINGULAR for non-commuting variables in its various contexts and incarnations. We will keep throughout 
an analogy with the commutative case. To avoid confusion, we will switch to calling commuting variables $y$, and keep $x$ to denote 
non-commuting variables. As is common let  $\F[\by]$ denote the algebra of polynomials in commuting variables $\by= \{ y_1, y_2, \dots, y_m \}$  over a field $\F$,  and $\F(\by)$ the field of rational functions over these variables (every element of this field is a ratio 
of two polynomials in $\F[\by]$).  For commuting variables there are many simple equivalent formulations of the problem.

\begin{fact}[PIT]\label{CommutEquiv}
The following are equivalent for a matrix $L = \sum_{i=1}^m y_i A_i$ in commuting variables $\by$. Some of these equivalences only hold for large enough fields $\mathbb{F}$.
\begin{enumerate}[(1)]
\item $L$ is singular (namely not invertible) over $\F(\by)$.
\item $L$ has nontrivial factors in $\F(\by)$. Namely, there exist $r <n$, matrices $K,M$ of dimensions $n\times r , r\times n$ (respectively) with 
entries in $\F(\by)$ such that $L=KM$. The smallest $r$ with this property is called the (commutative) {\em rank} of $L$, and is denoted 
$\emph{\rk} (L)$.
\item The linear space defined by $L$ is {\em singular}, namely contains no non-singular matrix over $\F$. In other words, for every choice of constants $\beta_i \in \F$ the matrix  $\sum_{i=1}^m \beta_i A_i$  is singular over $\F$.
\item $\Det\left(\sum_{i=1}^m y_i A_i \right) \equiv 0$ as a polynomial over $\F[\by]$, where $\Det$ is the determinant polynomial.
\end{enumerate}
\end{fact}

Now consider the non-commutative situation. The algebra of polynomials over non-commutative variables $\bx = \{x_1, x_2, \dots, x_m\}$ 
is familiar enough, and denoted $\Fx$. Its (skew) ``field of fractions''\footnote{Actually there are many, but only one ``universal field of fractions''}, 
which we have yet to define, is denoted $\F\langC \bx \rangC$. As in the commutative case, its elements include all inverses of 
nonzero polynomials in $\Fx$(but as mentioned above, these elements may be far more complex than ratios of polynomials in $\Fx$). One can formally define, as we will, the non-commutative SINGULAR problem, exactly as (1) in Fact~\ref{CommutEquiv} for the commutative case. 
But then most analogs of these equivalences seem to break down if we take them literally. E.g. (3) clearly doesn't make sense (as it will not distinguish commuting from non-commuting elements), and (4) is ill-defined due to the plethora of non-commutative variants of the determinant polynomial (see e.g. \cite{quasidet}). However, the beauty of the non-commutative setting is revealed in many more equivalences, which we collect in the following theorem (which in particular has interesting analogs of (2),(3),(4)). These equivalences illustrate some of the connections to optimization, (commutative) invariant theory, non-commutative algebra and quantum information theory. All of these equivalences will be discussed in this introduction, and elaborated on more formally later.  We note that some of them  are highly nontrivial theorems we shall later reference individually. To compare between the commutative and non-commutative worlds, and appreciate better the equivalences, it is instructive to have in mind the following example.

\begin{example}{\label{skew-symm}} The following $3 \times 3$ skew symmetric matrix is singular over $\F(z,w)$ but not singular over $\F \langC z,w\rangC$.
\[
  \begin{bmatrix}
    0 & z & w \\
    -z & 0 & 1 \\
    -w & -1 & 0
  \end{bmatrix}
\] 
\end{example}

\begin{theorem}[SINGULAR]\label{Equivalences}
The following are equivalent for a matrix $L = \sum_{i=1}^m x_i A_i$ in non-commuting variables $\bx$. Again some of these equivalences only hold for large enough fields $\mathbb{F}$ and some only over reals or complex numbers. But for concreteness, we will fix the field to be $\C$.
\begin{enumerate}
\item $L$ is singular (namely not invertible) over $\Fxx$.
\item $L$ has nontrivial {\em polynomial} factors in $\Fx$. Namely, there exist $r <n$, matrices $K,M$ of dimensions $n\times r , r\times n$ (respectively) with entries in $\Fx$\footnote{Moreover, the polynomial entries of $K,M$ in such a minimal decomposition can actually be taken to be polynomials of {\em degree at most 1}, namely {\em affine} combinations of the variables.}
 such that $L=KM$. The smallest $r$ with this property is called the {\em non-commutative rank} 
(and sometimes {\em inner rank}) of $L$, and is denoted $\emph{\ncrk} (L)$. 
\item The linear space defined by $L$ {\em with matrix coefficients} is {\em singular}.  Namely, for every integer $k$ and every tuple of  $k\times k$ matrices $(B_1, B_2, \ldots, B_m)$ over $\F$,  the matrix  $\sum_{i=1}^m B_i \otimes A_i$  is singular over $\F$.
\item For every integer $k$ we have  $\Det\left(\sum_{i=1}^m X_i \otimes A_i\right) \equiv 0$, where the $X_i$ are 
$k\times k$ matrices of {\em commutative} variables,
 $\Det$ is still the commutative determinant polynomial, acting here on $kn \times kn$ matrices.
\item $L$ has a {\em shrunk subspace}. Namely, there are subspaces $U,W$ of $\F^n$, with $\dim W < \dim U$, such that for all $i$ $A_iU \subseteq W$. 
\item $L$ has a {\em Hall-blocker} (in analogy with the Hall theorem for perfect matchings). Namely, there exist nonsingular matrices $B,C$ over $\F$ such that the symbolic matrix $BLC$ has an all-zeros minor of size  $i \times j$, with $i+j >n$.
\item The {\em completely positive} quantum operator (or map) associated with $L$ is rank-decreasing. Namely, there exist a positive semi-definite matrix $P$ such that $\rk (\sum_{i=1}^m A_iPA_i^\dagger) < \emph{\rk}(P)$.\footnote{Here $A^{\dagger}$ denotes the conjugate transpose of a complex matrix $A$.}
\item The tuple of matrices $(A_1, A_2, \ldots, A_m)$ is in the {\em null-cone} of the Left-Right action of $(SL_n(\F))^2$ on $m$ $n\times n$ matrices, namely they evaluate to zero on every non-constant homogeneous \textbf{invariant} polynomial under this action\footnote{The Left-Right action and its invariant polynomials are defined as follows. Consider  $mn^2$ commuting variables which are arranged in $m$ matrices $(Y_1, Y_2, \ldots Y_m)$, and consider polynomials in these variables. Every pair $B,C$ of determinant 1 matrices over $\F$ defines a linear map of these variables by sending this tuple to $(BY_1C, BY_2C, \ldots BY_mC)$.  A polynomial in these variables is {\em invariant} if it remains unchanged by this action for every such pair $B,C$.}. 
\end{enumerate}
\end{theorem}

In the rest of this introduction we will try to explain how and why the problem SINGULAR arises in these different contexts, how are they related. We will discuss the algorithm that solves SINGULAR (which already appears in~\cite{gurvits2004}), its complexity analysis (which is the main result of this paper), and how they arise from these (and other!) sources. We will also discuss the  extension of this algorithm to non-commutative rank computation, and the implications of these algorithms to the various settings.
We also highlight the recurring analogs and connections between the commutative and non-commutative worlds. There are probably many ways to weave this meandering story due to the multiple connections between the topics; we chose one (some accounts of subsets of these connections appear e.g. in~\cite{gurvits2004,HW14,IQS2015}). Note that most descriptions will be informal or semi-formal; formal definitions for everything that is actually needed will be given in the technical sections, and for the rest we provide references.

\subsection{The Free Skew Field}


Polynomials, both in commuting and non-commuting variables, can always be uniquely defined by their coefficients. And so, 
while their computational complexity is of course interesting, it is not part of their description. 
For rational functions this persists in the commutative setting, as they are simply ratios of polynomials. It is perhaps amusing that in the 
non-commutative setting, the first definition of the {\em free skew field} of rational functions was computational, namely one whose 
objects are {\em described} by their computation via a sequence of arithmetic operations. Both this description, and the subsequent discovery of a  syntactic representation play an important role in our story. 


Books on history and construction of the skew field include~\cite{Cohn-SF, Rowen}, 
and  a  thorough discussion from the computational perspective is given in~\cite{HW14}. Here we will be relatively brief, highlighting the key 
roles played by matrices and symbolic matrices in these constructions of the free skew field. 

Fix a field $\F$ and a set of non-commuting variables $\bx= \{ x_1, x_2,\dots, x_m \}$. As before, let $\Fx$ denote the the (free) algebra of 
non-commutative polynomials. A major objective is to construct a (minimal)  ``skew field of fractions'', which contains this algebra and 
is a division ring, namely every non-zero element is invertible. As it happens, there can be many non-isomorphic ways of doing that, but 
only one that is \emph{universal}\footnote{This is a technical term which we will not define here.}, the {\em free skew field} $\Fxx$, first 
defined by Amitsur~\cite{amitsur66} as follows (see a more detailed and precise description in~\cite{vinnikov10}).

A {\em rational expression} $r(\bx)$ is simply a formula whose inputs are variables and constants from $\F$, over the standard 
arithmetic operations of addition, multiplication and division, e.g. $x_1^{-1}- (x_1 + x_2)^{-1}$ and $\left(x_1+ \left(x_1x_2^{-1} \right)x_1\right)^{-1}$ (note that inversions can be arbitrarily nested\footnote{In general {\em inversion height}, the minimum amount of nesting needed, can be arbitrarily high. This is an important theorem of Reutenauer \cite{Reut96}. However, in the example above the nested inversion can be eliminated, and in fact the two expressions are equal (a simple fact which the reader might try to prove)! This equality is called Hua's identity \cite{Hua}, underlying the fundamental theorem of projective geometry}).  Essentially, the elements of the field will consist of these expressions, modulo an equivalence relation that we now define. Note that substituting 
$d \times d$ matrices for the variables of an expression $r$ it becomes  a (partial) function\footnote{when the expression attempts to invert a singular 
matrix it is undefined on that input. The domain of an expression is simply all input tuples on which it is defined}. Two expressions $r,s$ are {\em equivalent} if they agree (output the same matrix) for {\em all} inputs for which 
they are both defined (where we go over  all tuples of matrices from $M_d(\F)$ for {\em all} finite $d$).

\begin{theorem}[\cite{amitsur66}]\label{free-skew-field1}
The rational expressions under the above equivalence relation comprise the (universal) free skew field $\Fxx$.
\end{theorem}

\noindent One important complication (or perhaps a fountain of interesting problems) in the non-commutative setting is that a non-zero rational 
expression (or even a polynomial expression, namely one with no inverses) can be an identically zero function over the entire $M_d(\F)$ for some finite $d$. 
The simple example $x_1 x_2 - x_2 x_1$ is a {\em polynomial identity} for $d=1$, and more generally the {\em standard polynomial} 
$\sum_{\pi \in S_{2d}} (-1)^{\sgn(\pi)} x_{\pi(1)}x_{\pi(2)}\cdots x_{\pi(2d)}$ is a polynomial identity for $M_d(\F)$. By the celebrated 
Amitsur-Levitsky theorem, this example is tight in terms of the relation between the degree of the polynomial identity and dimension 
of matrices.

\begin{theorem}[\cite{AmiLev}]\label{AL}
For a fixed $d > \lfloor k/2 \rfloor$, a degree $k$ polynomial cannot vanish on all inputs in $M_d(\F)$.
\end{theorem}

This theorem immediately implies a polynomial identity test in probabilistic polynomial time (just like in the commutative case) mentioned in the next subsection - simply plug random matrices of appropriate size for the variables. However, for rational expressions, which is our 
main interest here, such upper bounds on the sufficient dimension to exclude {\em rational identities} are much harder to prove, and what is known is much weaker. The best bounds (and only for some fields) are exponential, which seem to only provide probabilistic exponential time rational identity tests. The last subsection of this introduction shows how these exponential dimension bounds arise and obtained from invariant theory, and how we use them to nonetheless derive a {\em deterministic, polynomial} time identity test promised in our main theorem.  

A second construction of the free skew field which is somehow more concrete was developed by Cohn. In the notation we have already established of symbolic matrices Cohn proved:

\begin{theorem}[\cite{cohn1971-sf}]\label{cohn71}
Every element of the free skew field $\Fxx$ is an entry of the inverse of some (invertible) symbolic matrix whose entries are polynomials in $\Fx$.
\end{theorem}

For this definition not to be self-referential (and other reasons) Cohn proved some of the important characterizations of invertible matrices, 
including items (2) and (5) in Theorem~\ref{Equivalences}. It is clear that (each entry of) the inverse of such a matrix can be given by 
a rational expression, and the question of whether a matrix is invertible again arises naturally, and it turns out to capture rational identity testing.

\subsection{Word Problems and Identity Tests}

Word problems and their complexity are central throughout mathematics, arising whenever mathematical objects in a certain class have several representations. In such cases, a basic problem is whether two such representations describe the same object. Often, indeed, some problems of this form have served as early examples of decidable and undecidable problems\footnote{E.g. deciding if two knots diagrams describe the same knot was proved decidable by Haken~\cite{Hak}, and deciding if two presentations with generators and relations describe the same group was proved undecidable by Rabin~\cite{Rab}}. 

For us, the objects in question are polynomials and rational functions in commuting and non-commuting variables. Their standard representations will be arithmetic formulas and circuits, which take the variables (and constants in $\F$) as inputs, and use plus, times and {\em division} gates. An excellent exposition of arithmetic complexity models and the state-of-art in the subject is \cite{Shpilka-Yehudayoff} (which discusses at length the polynomial identity testing problem\footnote{while this paper focuses on identity testing, we note that our interest is partly (and indirectly) motivated by the more basic problem of proving lower bounds for non-commutative circuits. We refer the reader to the papers~\cite{nisan1991,HWY10,HWY11,limaye2015} and their references for existing lower bounds on weaker models, some completeness results, and possible approaches to proving stronger lower bounds}).
We stress that we will specifically allow division gates in our computational models, both because we want to consider rational functions, and because unlike for computing polynomials, in the non-commutative case there is no known analog of Strassen's theorem~\cite{Str} that divisions can be efficiently eliminated. This issue and partial results, many relevant to this paper, are discussed in~\cite{HW14}. The word problem in our arithmetic context is equivalent to an Identity Test, checking if a given representation of a polynomial or rational function describes the trivial element, the zero polynomial. 

As is well known, the main reason why PIT is so important in the commutative setting is that it captures the polynomial and rational function identity test for formulas, as proved by Valiant~\cite{Val}. Combining it with the equivalences of Fact~\ref{CommutEquiv}, we have

\begin{theorem}[\cite{Val}]\label{PIT-complete}
There is an efficient algorithm which converts every arithmetic formula $\phi (\by) $ in {\em commuting} variables $\by$ of size $s$ to a 
symbolic matrix $L_\phi$ of size $\emph{\poly}(s)$, such that $\phi$ is identically zero if and only 
if $L_\phi \in$ {\em PIT} (i.e., $L_\phi$ is singular).
\end{theorem}

Theorem~\ref{cohn71}  of the previous subsection, showing that in the non-commutative setting as well SINGULAR captures polynomial and rational function identity test, is the analog Valiant's completeness theorem. It was proved even earlier by Cohn~
(see also Malcolmson's version~\cite{malcolmson}), using similar linear algebraic ideas. Here is a more precise statement which gives the analogy.

\begin{theorem}[\cite{cohn1971-sf} ]\label{RIT-compete}
There is an efficient algorithm which converts every arithmetic formula $\phi (\bx)$ in {\em non-commuting} variables $\bx$ of size $s$ to a symbolic matrix $L_\phi$ of size $\emph{\poly}(s)$, such that the rational expression computed by $\phi$ is identically zero if and only if $L_\phi \in$ {\em SINGULAR}.
\end{theorem}

As mentioned, the structure of the free skew field is so complex that unlike the commutative case, even decidability of SINGULAR (and rational identity testing) is far from obvious. The first to prove decidability was Cohn in~\cite{Cohn-Word, Cohn-Word2}. The first explicit bound on time was given by Cohn and Reutenauer in~\cite{CR99}, reducing it to a system of commutative polynomial equations using characterization (5) of SINGULAR, proved earlier by Cohn, which puts it in $\cPS$. The best upper bound before this work was singly exponential time, obtained by~\cite{IQS2015}, and of course yields the same bound for rational identity testing.

From our Theorem~\ref{main} we conclude a {\em deterministic, polynomial} identity test for non-commuting rational expressions.
\begin{corollary}\label{Formula-RIT}
The non-commutative rational identity testing problem is in $\cP$. Namely, there is an algorithm which for any non-commutative formula 
over $\Q$ of size $s$  and bit complexity $b$ determines in $\emph{\poly} (s, b)$ steps if it is identically zero.
\end{corollary}

It is important to stress that unlike the commutative case, where symbolic matrix inversion can be simulated by quasi-polynomial sized 
formulas, in the non-commutative case symbolic matrix inversion is exponentially more powerful than formulas (and so Theorem~\ref{main} is far more powerful than its corollary above). We state these two results formally below for contrast\footnote{Replacing formulas by circuits there is no contrast - in both the commutative and non-commutative setting matrix inverse has a polynomial size circuit (with division of course)~\cite{HW14}}. Note that when saying ``computing the inverse'' we mean computing any entry in the inverse matrix.

\begin{theorem}[\cite{Hya,Ber}]\label{comm-Hardness}
The inverse of a commutative $n\times n$ symbolic matrix  in $\F(\by )$ can be computed by formulas of size $n^{O(\log n)}$.
\end{theorem}
\begin{theorem}[\cite{HW14}]\label{comm-Hardness}
The inverse of a non-commutative $n\times n$ symbolic matrix  in $\Fxx$ requires formulas of size $2^{\Omega(n)}$.
\end{theorem}

Our algorithm of Theorem~\ref{main} and Corollary~\ref{Formula-RIT}  is a ``white-box'' identity test (namely one which uses the description of the given matrix or formula), in contrast to a ``black-box'' algorithm which only has input-output access to the function computed by them. The best-known ``black-box'' identity test for these formulas, even if randomization is allowed(!), requires exponential time. It is a very interesting question to find a faster black-box algorithm.

When division is {\em not} allowed, efficient deterministic identity tests of non-commutative {\em polynomials} were known in several models. The strongest is for arithmetic branching programs (ABPs). As this model is easily simulable by matrix inversion (see Theorem 6.5 in~\cite{HW14}, our algorithm provides an alternative (and very different) proof of the following theorem of Raz and Shpilka~\cite{RazShp}.

\begin{theorem}[\cite{RazShp}]\label{ABPs}
There is a deterministic  polynomial time white-box  identity testing algorithm for non-commutative ABPs.
\end{theorem}

For certain classes computing only polynomials even efficient black-box algorithms are known, and we mention two below; 
the first deterministic, for non-commutative ABPs, and the second probabilistic (but for circuits). 
\begin{theorem}[\cite{ForShp}]\label{roABPs}
There is a deterministic quasi-polynomial time black-box  identity testing algorithm for non-commutative ABPs.
\end{theorem}

\begin{theorem}[\cite{BogVee,AmiLev}]\label{Poly-degree}
There is a probabilistic polynomial time black-box identity testing algorithm for non-commutative circuits.
\end{theorem}

\subsection{Commutative and Non-commutative Rank of Symbolic Matrices}

There are many sources, motivations and results regarding invertibility, and more generally rank, of {\em commutative} symbolic matrices, which are much older than the complexity theory interest in it.  These mathematical papers, like the ones in the 
computational complexity literature, regard it as a difficult problem, and attempt to understand a variety of special cases (of a different nature than restrictions of the computational power of the model). Some of the many references to this body of work can be found in the papers~\cite{FR04,Mes}. Often, the {\em non-commutative} rank (explicitly or implicitly) is used to give upper bound on the commutative rank, and  their relationship becomes of interest. We focus on this connection here, and explain how our main result implies a deterministic {\em approximation algorithm} to the commutative rank.

In this section we assume that the underlying field $\F$ is infinite.
We will use the same notation $L$ for both a symbolic matrix, as well as the subspace of matrices spanned by it (when fixing the variables to constants in the field). We repeat the definitions  and elaborate on what is known regarding the commutative and non-commutative ranks, from now on denoted by $\rk(L)$ and $\ncrk(L)$.  Note that the characterizations allow thinking about both without reference to the respective fields of fractions (over which they are most 
naturally defined), but only to polynomials.

\begin{fact}
The following are equivalent for a matrix of linear forms over commutative variables, $L(\by)$.
\begin{itemize}
\item $\emph{\rk} (L(\by)) = r$ over $\F(\by)$
\item $r$ is the {\em maximal} rank of any matrix in the subspace $L$ over $\F$.
\end{itemize}
\end{fact}

While the characterization above is simple, the one below is very substantial, mostly developed by Cohn for his construction of the free skew field. The first characterization is due to him~\cite{Cohn-SF} and the second  is due to Fortin and Reutenauer~\cite{FR04} who heavily use Cohn's techniques.
\begin{theorem}\label{cohn-nc=inner}
The following are equivalent for a matrix of linear forms over non-commutative variables, $L(\bx)$.
\begin{itemize}
\item $\emph{\ncrk} (L(\bx)) = r$ over $\Fxx$
\item The {\em inner rank} of $L$ over $\Fx$ is $r$. Namely, $r$ is the {\em minimal} number such there exists matrices 
 $K,M$ of dimensions $n\times r , r\times n$ (respectively) with {\em polynomial} entries in $\Fx$ such that $L=KM$.
 Moreover, in this decomposition the factors $K,M$ can be assumed to have {\em affine linear} entries.
\item The space $L$ is $r$-{\em decomposable}. Namely, $r$ is the {\em minimal} number such that there exists invertible matrices $B,C$ over
$\F$ such that $BLC$ has a minor of zeros of size $i \times j$ with $i+j = 2n-r$.
\end{itemize}
\end{theorem}

We can extend our main Theorem~\ref{main} from testing singularity to  efficiently computing the 
non-commutative rank over $\Q$, this is done in Theorem~\ref{thm:comp-ncrank}.

\begin{theorem}\label{computingNCrank}
There is a deterministic algorithm which, given $m$ $n\times n$ integer matrices $A_1,\dots A_m$ with entries of bit-size $b$, 
computes $\emph{\ncrk}(L)$ in time $\emph{\poly}(n,m,b)$ (where $L = \sum_{i=1}^m x_i A_i$).
\end{theorem}

\begin{remark}
We note here that in some formulations of the problem, $L$ is represented in {\em affine form}, namely where an additional constant 
matrix $A_0$ is added to $L$ (this may be viewed as a non-commutative analog of the rank-completion problem). However, the 
algorithm above works in this case as well , since $\ncrk(L)$  remains unchanged if an additional variable $x_0$ was added as well. 
So, we will stick with the linear formulation. 
\end{remark}

It is not hard to see from the definitions that for every $L$ we have $\rk(L) \leq\ncrk(L)$. We have already seen that they can be different in Example~\ref{skew-symm}. Taking many copies of that $3\times 3$ matrix we see that there can be a factor $3/2$ gap between the two: for any $r$ there are matrices $L$ with $\rk(L)=2r$ and $\ncrk(L)=3r$. However, Fortin and Reutenauer~\cite{FR04} proved that this gap is never more than a factor of 2.

\begin{theorem}[\cite{FR04}]
For every $L$ we have $\emph{\ncrk}(L) \leq 2\emph{\rk}(L)$. 
\end{theorem}

An immediate corollary of our main result is thus a factor 2 approximation of commutative rank.

\begin{corollary}
There is a polynomial time algorithm which for every symbolic matrix in {\em commuting} variables over $\Q$ approximates 
$\emph{\rk}(L)$ to within a factor of 2.
\end{corollary}

We find the question of obtaining a better approximation ratio efficiently a very interesting problem, a different relaxation of the commutative PIT problem that as far as we are aware was not studied till now.

\subsection{Compression Spaces, Optimization and Gurvits' Algorithm $G$}

An important class of spaces of matrices $L$, which is studied in many of the papers mentioned above, is the class of 
{\em compression spaces} (this notation was apparently introduced by Eisenbud and Harris~\cite{EisHar}). They are defined 
as those spaces $L$ for which  $L$ is $\rk (L)$-decomposable. A simpler definition follows the characterization of \cite{FR04} 
above, namely a space $L$ is a compression space iff $\rk(L) = \ncrk(L)$. The importance of compression spaces and their many origins 
will be surveyed below.

A deterministic polynomial time {\em commutative} PIT for compression spaces (over $\mathbb{Q}$) was discovered by Gurvits~\cite{gurvits2004}, 
a paper which serves as the starting point for this work.\footnote{It is interesting to note that most recent progress on deterministic 
PIT algorithms (e.g.~\cite{KS01, DS06, KS07, SV11, ForShp} among many others) are for polynomials computed by a variety 
of restricted classes of 
arithmetic circuits. Algorithm $G$ seems to differ from all of them in solving PIT for a very different class of polynomials, which we do 
not know how to classify in arithmetic complexity terms.} Indeed Gurvits' algorithm, which we will denote by Algorithm $G$, is 
{\em the same} one we use here for our main Theorem~\ref{main}; 
our main contribution here is the {\em analysis} of its performance for any input $L$, not necessarily a compression space. We will 
return to the algorithm and its analysis, but first let us discuss its extension from testing singularity to rank computation. Before that, we  state  Gurvits' theorem, which actually proves more about Algorithm $G$.

\begin{theorem}[\cite{gurvits2004}]\label{Gur}
There is a deterministic polynomial time algorithm, Algorithm $G$, which for every $n$ and every $n\times n$ matrix $L$ given 
by a set of integer matrices $(A_1, A_2, \dots , A_m)$, outputs ``singular'' or ``invertible'', and its output is guaranteed to be 
correct\footnote{For both the commutative and non-commutative definitions.} when either $\emph{\rk}(L)=n$ or $\emph{\ncrk}(L)<n$.
\end{theorem}

In particular, the algorithm always gives the correct answer for compression spaces. Our result on efficient computation of the non-commutative rank implies that for compression spaces we can determine the {\em commutative rank}. 

Many natural spaces are compression spaces, and these arise from a variety of motivations and sources in linear algebra, geometry 
and optimization. One such source is the attempt to characterize singular spaces (namely when $\rk(L)<n$) by relating the rank to the 
dimension of the space matrix $L$ defines, denoted $\dim(L)$. Perhaps the earliest result of this type is due to Dieudonn\'{e}~\cite{Die} who 
proved that if $L$ is singular and $\dim(L)=n^2-n$, then for some nonsingular matrices $B,C$ over $\F$ the matrix $BLC$ has an all-zero 
row or all-zero column. In other words, it has a Hall-blocker as in condition (6) in Theorem~\ref{Equivalences}, and so is in particular 
decomposable. Many other similar results appear in e.g. \cite{Atkinson, AtkLyd,Beasley,EisHar,Mes}, which prove decomposability 
under more general conditions on the dimension, and study a variety of other properties of singular and rank-bounded spaces. 

Yet another source, elaborated on in~\cite{gurvits2004}, is geometric duality theorems, such as the ones for matroid intersection and 
matroid parity problems. These sometimes give rise to compression spaces, with one prototypical example being spaces where the 
generating matrices $A_i$ all have rank-1; this follows from the Edmonds-Rado~\cite{Edm_matroid, rado} duality theorem for matroid 
intersection\footnote{It is an interesting question whether a compression space naturally arises from matroid parity duality of 
Lovasz~\cite{Lov80, Lov1989}.}. Another, simpler example are spaces generated by upper-triangular matrices. Other examples of 
compression spaces are given in ~\cite{gurvits2004}. Following his paper, we note that our rank algorithm above can solve such 
optimization problems {\em in the dark}, namely when the given subspace has such {\em structured} spanning generators, but the 
generators $A_i$ actually given to the algorithm may be {\em arbitrary}. This is quite surprising, and in general cannot work by 
uncovering the original structured generators from the given ones: an efficient algorithm is not known for the problem of deciding if a given space
of matrices is spanned by rank-1 matrices and we believe it is hard. It would be interesting to explore which other optimization problems can be
encoded as the rank of a compression space, and whether the fact that it can be computed ``in the dark" has any applications.

In recent related work, \cite{IKQS} give a different ``in the dark'' algorithm for computing the rank of certain compression spaces (including the ones spanned by rank-1 matrices) using so-called second Wong sequences, a linear algebraic analog of ``augmenting paths for matchings". Its main advantage is that it works not only over subfields of $\C$ as our algorithm, but also over large enough finite fields. This resolves an open problem in~\cite{gurvits2004}.

\subsection{Permanents, Quantum Operators, Origins and Nature of Algorithm $G$}\label{LSWanalysis}

We will describe the algorithm formally in Section~\ref{AlgG}. Here we give an informal description, which makes it easy to explain its origins and nature. We find these aspects particularly interesting. One  (which has very few parallels) is that while the problem SINGULAR solved by Algorithm $G$ in Theorem~\ref{main} is purely {\em algebraic}, the algorithm itself is purely {\em analytic}; it generates from the input a sequence of complex-valued matrices and attempts to discover if it is convergent. Another is that the algorithm arises as a quantum analog of another algorithm with very different motivation that we now discuss.

To give more  insight to the working of algorithm $G$, let us describe another algorithm (for a different problem) which inspired it, which we  call Algorithm $S$. This {\em matrix scaling} algorithm was developed by Sinkhorn~\cite{Sink} for applications in numerical analysis, and has  since found many other applications (see survey and references in~\cite{LSW}, who used it as a basis for their deterministic algorithm for approximating the permanent of non-negative matrices).
Two different analyses of Sinkhorn's algorithm $S$, one of~\cite{LSW} and the other in the unpublished~\cite{GurYianilos} inspire the analysis of Algorithm $G$ in~\cite{gurvits2004}. 

We describe the \cite{LSW} analysis for Algorithm $S$.
We need a few definitions. For a non-negative matrix $A$, let $R(A)$ denote the diagonal matrix whose $(i,i)$-entry is the inverse of the $L_1$ norm of row $i$ (which here is simply the sum of its entries as $A$ is non-negative). Similarly $C(A)$ is defined for the columns\footnote{A ``non-triviality'' assumption is that no row or column in $A$ is all zero.}. 

Algorithm $S$ gets as input a non-negative integer matrix $A$. For a fixed polynomial (in the input size) number of iterations it repeats the following two steps
\begin{itemize} 
\item Normalize rows: $A \leftarrow R(A)\cdot A$
\item Normalize columns: $A \leftarrow A\cdot C(A)$
\end{itemize}

What does this algorithm do? It is clear that in alternate steps either $R(A)=I$ or $C(A)=I$, where $I$ is the identity matrix. Thus $A$ itself alternates being row-stochastic and column-stochastic. The question is whether both converge to $I$ together, namely, if this process converts $A$ to a {\em doubly stochastic} matrix. In \cite{LSW} it is proved that this happens if and only if $\Per(A) >0$, where $\Per$ is the permanent polynomial. Moreover, convergence is easy to detect after a few iterations! If we define $\ds(A) = || R(A)-I ||^2 + || C(A)-I ||^2$ as  a notion of distance between $A$ and the doubly stochastic matrices, then the convergence test is simply whether $\ds(A) < 1/n$. If it is that small at the end of the algorithm then $\Per(A) >0$, otherwise $\Per(A) =0$. 

The analysis of convergence of Algorithm $S$ in~\cite{LSW} is extremely simple, using the permanent itself as a progress measure (or potential function). It has the usual three parts which makes a potential function useful: 
\begin{enumerate}
\item The input size provides an exponential lower bound on the starting value of $\Per(A)$, 
\item The arithmetic-geometric mean inequality guarantees that it grows by a factor of $1+1/n$ at every iteration, and 
\item The permanent of any stochastic matrix is upper bounded by 1. 
\end{enumerate}
We shall return to this analysis of Algorithm $S$ soon.

As it happens, Algorithm $G$ is a {\em quantum} analog of Algorithm $S$! In quantum analogs of classical situations two things typically happen, diagonal matrices (which commute) become general matrices (which do not), and the $L_1$ norm is replaced by $L_2$. This happens here as well, and we do so almost syntactically, referring the reader to~\cite{gurvits2004} for their quantum information theoretic intuition and meaning of all notions we mention. 

The input  to algorithm $G$ is a symbolic matrix $L=\sum_i x_i A_i$, given by the $n\times n$ integer matrices $(A_1, A_2, \dots ,A_m)$. 
Briefly, $L$ is viewed as a {\em completely positive (quantum) operator}, or map,  on psd matrices, mapping such a (complex valued) matrix $P$ to 
$L(P)= \sum_i A_i PA_i^\dagger$ ($P$ is typically a ``density matrix'' describing a quantum 
state, namely a psd matrix with unit trace, and the operator $L$ will typically preserve trace or at least not increase it). The dual operator 
$L^*$ acts (as you'd expect) by $L^*(P)= \sum_i A_i^\dagger PA_i$. The analog ``normalizing factors'' for $L$, named $R(L)$ and $C(L)$  are defined\footnote{Again using a ``non-triviality'' assumption these matrices are invertible.} 
 by $R(L) = ( \sum_i A_i A_i^\dagger)^{-\frac12}$, and $C(L) = ( \sum_i A_i^\dagger A_i)^{-\frac12}$. Note that 
$R(L)=L(I)^{-\frac12}$ and $C(L)=L^*(I)^{-\frac12}$.

On input $(A_1, A_2, \dots ,A_m)$ Algorithm $G$ repeats, for a fixed polynomial (in the input size) number of iterations, the following analogous two steps
\begin{itemize} 
\item Normalize rows: $L \leftarrow R(L)\cdot L$
\item Normalize columns: $L \leftarrow L\cdot C(L)$
\end{itemize}

So again, row and column operations are performed alternately, simultaneously on all matrices $A_i$. It is clear, as above, that after each step either $R(L)=I$ or $C(L)=I$. It is natural to define the case when both occur as ``doubly stochastic'', and wonder under what conditions does this sequence converge, namely both $R(L)$ and $C(L)$ simultaneously approach $I$, and alternatively the limiting $L$ simultaneously fixes the identity matrix $I$. A natural guess would be that it has to do with a ``quantum permanent''. Indeed, Gurvits~\cite{gurvits2004} defines such a polynomial (in the entries of the $A_i$) $\QuantPer(L)$, and proves several properties and characterizations (for example, it is always non-negative like the permanent, and moreover specializes to the permanent when the operator $L$ is actually a ``classical'' operator described by a single non-negative matrix $A$). 

One can similarly define in an analogous way to the classical setting a ``distance from double stochastic" by 
$\ds(L) = || R(L)-I ||^2 + || C(L)-I ||^2$, and test (after polynomially many iterations) if $\ds < 1/n$. This is precisely what Algorithm $G$ 
does. It is not hard to see that if $L$ (with non-commuting variables) is singular, then there is no convergence, the test above fails
(and indeed $\QuantPer(L)=0$). However, in contrast to Algorithm $S$, the analysis in~\cite{gurvits2004} falls short of proving that, 
otherwise we have convergence (and $\QuantPer(L)>0$). It does so  {\em only} under the strong (commutative) assumption  
$\Det(L) \neq 0$, namely when the symbolic matrix, viewed with {\em commuting} variables, is nonsingular. This result was stated above as Theorem~\ref{Gur}. The 3-step complexity analysis proceeds exactly as in \cite{LSW} for the classical Algorithm $S$ described above, using the quantum permanent as a progress measure. However, if $\Det(L) = 0$ then $\QuantPer(L) =0$, and lower bounding the initial quantum permanent from below in part (1) of 
that  analysis fails.

To prove our main theorem, showing that convergence happens {\em if and only if} $L$ is non-singular over the non-commutative skew  field, we use another ingredient from Gurvits' paper. He proposes another progress measure, called {\em capacity} and denoted  $\capac(L)$. Capacity (defined in the next section) is a quantum analog of another classical progress measure, based on 
KL-divergence, which was used in~\cite{GurYianilos} for the analysis of Algorithm $S$. The advantage of capacity, as pointed out in~\cite{gurvits2004}, is that it does not necessarily vanish when $\Det(L) = 0$. Indeed, it is positive whenever $L$ is non-singular in the non-commutative sense!
However Gurvits could only bound it sufficiently well from below as a function of the input size if $\Det(L) \neq 0$.  Thus, for a general input $L$, positive capacity only guarantees that Algorithm $G$ converges in a  finite number of steps, without providing an upper bound on that number.

Our main contribution, which leads to a polynomial upper bound, is proving an explicit capacity lower bound for {\em every} $L$ that is nonsingular over the skew field. The source of the proof turns out to be commutative algebra and invariant  theory, as we discuss next.

\subsection{Invariant Theory, Left-Right Action, Capacity and the Analysis of Algorithm $G$} 

Invariant theory is a vast field; we will exposit here only the minimal background that is necessary for this paper. The books~\cite{CLO07, Derksen-Kemper, Kraft-Procesi} provide  expositions on the general theory. 
More focused discussions towards our applications appear in the appendix of~\cite{HW14} and Section 1.2 
of~\cite{Forbes-Shpilka}. 

Invariant  theory deals with understanding the symmetries of mathematical objects, namely transformations of the underlying space 
which leave an object unchanged or {\em invariant}. Such a set of transformations always form a group (and every group arises this way). 
One major question in this field is, given a group acting on a space, characterize the objects left invariant under all elements of the group. 
Here we will only discuss very specific space and actions:  polynomials (with commuting variables!) that are left invariant under certain linear transformations of the variables. The invariant polynomials under such action clearly form a ring (called the {\em invariant ring}). Important for us is the  {\em null-cone} of the action, which is simply all assignments to variables which vanish on {\em all} non-constant homogeneous invariant polynomials. Namely the null cone is the affine algebraic variety {defined} by the ideal generated by homogeneous invariant polynomials of positive degree.

Here are two motivating examples surely familiar to the reader. The first example is given by polynomials in $n$ variables $y_1,\dots ,y_n$ 
invariant under the full group of permutations $S_n$; here the invariant polynomials are (naturally) all {\em symmetric} polynomials. 
While this is an infinite set, containing polynomials of arbitrarily high degrees, note that this set of polynomials is generated (as an algebra) 
by the $n+1$ {\em elementary symmetric polynomials}, with the largest degree being $n$. This finite generation is no coincidence! A general, important 
theorem of Hilbert~\cite{Hil} assures us that  for ``such" linear actions there is always a finite generating set (and hence a finite upper 
bound on their maximum degree). Obtaining upper bounds on the degree of generating sets, finding descriptions of minimal generating 
sets for natural actions are the classical goals of this area, and a more modern one\footnote{Arising in particular in the GCT program of Mulmuley and Sohoni} is obtaining succinct descriptions and efficient 
computation of these invariants (e.g. see ~\cite{Mulmuley,Forbes-Shpilka}). The case of the action of the symmetric group is an 
excellent example of a perfect understanding of the invariant polynomial ring. Note that for this action the null-cone is simply the all-zero vector.

The second example, much closer to our interest here, is the ``Left-Right'' action of the group $(SL_n (\F))^2$ acting on $n^2$ variables 
$y_{ij}$ arranged in an $n\times n$ matrix $Y$. Specifically, two $n\times n$ matrices (of determinant 1) $(B,C)$ take $Y$ to $BYC$, changing the basis on the left and right \footnote{Note that for it to be a group action in the strict sense, one should study the action which takes $Y$ to $BYC^{-1}$ or $BYC^{T}$ but for simplicity, we will avoid this distinction.}. It is not hard to see that all invariant polynomials under this action are generated by one - the 
determinant $\Det(Y)$ (which again has degree $n$). Consequently the null-cone of this action is all singular matrices over the field.

What we consider here is the left-right action on $m$ $n\times n$ matrices, where a pair $(B,C)$ as above takes $(Y_1,Y_2, \dots ,Y_m)$ to $(BY_1C,BY_2C, \dots ,BY_mC)$. The study of the invariant ring of polynomials (in the $mn^2$ variables sitting in the entries of these matrices) for this action was achieved\footnote{We note that this is part of the larger project of understanding {\em quiver representations}, started by the works of Procesi, Razmysolov, and Formanek \cite{procesi, razmyslov, formanek}.} by \cite{DW2000, DZ2001, SVdB2001, ANS10}, and will look very familiar from condition (4) in Theorem \ref{Equivalences}, as well as from Amitsur's definition of the free skew field\footnote{Note though that the roles of which matrices in the tensor product are variable, and which are constant, has switched!}

\begin{theorem}[\cite{DW2000, DZ2001, SVdB2001, ANS10}]\label{invariant-left-right}
Over algebraically closed fields, the invariant ring of polynomials of the left-right action above is generated by all polynomials of the form $\Det \left(\sum_i D_i \otimes Y_i \right)$, for all $d$ and all $d \times d$ matrices $D_i$.
\end{theorem}

It is  worthwhile to stress the connection forged between the commutative and non-commutative worlds by this theorem when combined with Amitsur's and Cohn's constructions of the skew field. A set of matrices $(A_1,A_2,\dots ,A_m)$ is in the null-cone of the left-right action if and only if the symbolic matrix $L=\sum_i x_i A_i$ is not invertible in the free skew field! In other words, the 
non-commutative SINGULAR problem (and thus rational identity testing, and the word problem in the skew field) arises completely 
naturally in commutative algebra. Of course, invertibility itself is invariant under the left-right action (indeed, even by any invertible 
matrices $B,C$, not necessarily of determinant 1), so one expects a connection in at least one direction. At any rate, one may hope now that commutative algebraic geometric tools will aid in solving these non-commutative problems, and they do!

The generating set of the invariant ring given in the theorem is still an infinite set, and as mentioned above, Hilbert proved that some finite subset of it is already generating. 

\begin{theorem}[\cite{Hil}]\label{hilbert_finite} The invariant ring of polynomials of the action of a reductive group on a finite dimensional vector space is finitely generated.
\end{theorem}

Hilbert proved his theorem for the actions of $SL_n (\C)$ but his proof can be easily seen to be extendable to arbitrary reductive groups. Since $(SL_n (\C))^2$ is reductive, 
Hilbert's theorem implies that the invariant ring of polynomials of the left-right action is finitely generated. 
A natural definition to make now is $\beta(n)$, the smallest integer such that taking only matrices with 
$d\leq \beta(n)$ generates the invariant ring\footnote{This bound may a-priori depend on $m$, the number of matrices, but we already noted that $m\leq n^2$.}. Hilbert did not give an explicit bound. The first to do so was Popov~\cite{Popov}, and this 
bound was significantly improved by Derksen \cite{derksen}. A related quantity $\sigma(n)$ is the smallest integer such that taking only 
matrices with $d\leq \sigma(n)$ suffice for testing membership in the null-cone. Derksen  \cite{derksen} proved that $\sigma(n)$ and 
$\beta(n)$ are polynomially related (over algebraically closed fields of characteristic zero) and that (for algebraically closed fields of characteristic zero) 
$\sigma(n) \leq n^2 4^{n^2}$. All bounds above hold for general reductive groups. For the specific left-right action further progress was made! First, Derksen's bound was further improved to a smaller exponential by~\cite{IQS2015}, and extended to finite fields.

\begin{theorem}[\cite{IQS2015}]
$\sigma(n) \leq (n+1)!$
\end{theorem}

This exponential upper bound was the state of art when we wrote our paper. Subsequent to it, this bound was greatly improved to linear! 

\begin{theorem}[\cite{derksen2015, IQS15b}]\label{degree-bound}
$$\sigma(n) \leq n-1$$
\end{theorem}

This much better bound gave rise to  further important developments discussed below under ``Subsequent Work'', that highlight even more how invariant theory informs algorithmic efficiency in certain problems. We conclude this section describing the idea behind our proof, which uncovers this connection.

We plan to used the exact same 3-step complexity analysis of Algorithm $G$ as the one for Algorithm $S$ described in subsection \ref{LSWanalysis}, and which Gurvits~\cite{gurvits2004} used to analyze Algorithm $G$ for a special subset of inputs. As mentioned, the only missing ingredient is step (1), namely an explicit lower bound on the {\em capacity} of an input operator $L$. 
Just like in the classical case of Algorithm $S$, an exponential lower bound (in terms of the size of $L$) is required to establish a polynomial convergence. This is a proper time to finally define this important parameter of positive operators\footnote{We note that this notion of capacity seems to have nothing to do with the usual capacity of a quantum channel}. 

Given $L=(A_1, A_2,\ldots,A_m)$, recall that $L(P) = \sum_i A_i PA_i^\dagger$ and  define $\capac(L)$ by
$$
\capac(L) = \text{inf} \: \Det(L(P))
$$
where the infimum is taken over all psd matrices $P$ with $\Det(P) \ge 1$.

First, observe that capacity is the solution to a non-convex optimization problem (over a convex domain \footnote{Recall that log determinant is a concave function over the domain of positive definite matrices}). We will see in the next section that we can not only bound it, but actually compute it efficiently!
Next, observe that capacity is an invariant (though not a polynomial invariant) of the left-right action.
What we know is that $\capac(L)\neq 0$. From this we infer that some polynomial invariant of the form above, $\Det \left(\sum_i A_i \otimes D_i \right)$ is nonzero, where the $D_i$ are $d\times d$ matrices for some $d$.

This expression naturally suggests considering several new positive operators, and relating their capacities.
The first, in dimension $d$, is the operator, $M=(D_1, D_2,\ldots,D_m)$. Two others, in $nd$ dimensions are
$L_M=(A_1\otimes D_1, A_2\otimes D_2, \ldots,A_m\otimes D_m)$, and  $L\otimes M = (A_1\otimes D_1, \ldots \, A_i\otimes D_j, \ldots, \,A_m\otimes D_m)$. The non-vanishing of the determinant above proves that both $\capac(L_M)>0$ and $\capac(L\otimes M) >0$. And it is easy to obtain exponential lower bounds on both quantities, which however depend not only on the description size of $L$ (which is given), but also of that of $M$ which could be much larger. 

To facilitate our original capacity lower bound, which only used $d\leq \sigma(n) \leq \exp(\poly (n))$,  we proved that the normalized\footnote{Taking the dimension's root of capacity} capacity is multiplicative, a fact of independent interest (see details in Section \ref{subsec:mult_cap}):
$$\capac(L\otimes M)^{1/nd} = \capac(L)^{1/n} \cdot \capac(M)^{1/d}$$

Using only the easier direction, namely that the left hand side is at most the right hand side, it suggests that our lower bound on $\capac(L)$ will follow from a lower bound on $\capac(L\otimes M)$ and an upper bound on $\capac(M)$. Both of these in turn follow from an upper bound on the entry sizes in $M$, namely in the $D_i$, which easily follow from the given exponential upper bound on $d$. A key thing to note is that even when $d$ is exponential in $n$, and so these size bounds are doubly exponential, it is the $d$'th root of this bound that gives the required lower bound $\capac(L)$.

For the new lower bound, presented  in Section \ref{sec:capacity_lb}, we bound $\capac(L)$ using directly a bound on $\capac(L\otimes M )$. 
Using much more careful size bounds on the entries of $M$, which follow from Alon's Nullstellensatz theorem, and tightening the reduction between the two capacities, we are able to derive the desired capacity lower bound using any finite bound on $d$ !

This new proof eliminates the need for good degree upper bounds for the purpose of analyzing this algorithm (if one does not care about efficiency beyond having a polynomial bound). However good degree upper bounds (that are anyway an important goal in invariant theory) were essential for the original proof, and for general actions were useful for other algorithmic problems which arose in other contexts in computational complexity, for example in geometric complexity theory (GCT). Further, we believe they will be important for studying other problems in algebraic complexity and optimization. Indeed, in a work in progress, we need degree bounds for invariant rings whose generating sets are not understood as well as for the left-right action, but Derksen's general exponential degree upper bound still holds and is useful.

The capacity lower bound above implies, as mentioned, a polynomial upper bound on the {\em number of iterations} of Algorithm $G$, assuming we can perform exact arithmetic! To actually bound the running time the algorithm has to be refined, truncating numbers so that they can be represented by polynomially many bits, and so a careful analysis of the bit complexity is required to actually prove its running time. This issue naturally leads to the next section, in which we discuss how the algorithm can be extended to actually {\em compute} capacity of the input operator, as well as provide bounds on its continuity. These results in turn are essential in a follow-up work~\cite{GGOW2} we have done on the Brascamp-Lieb inequalities, where capacity arises naturally via a simple reduction to operator scaling.

\subsection{Computation and Continuity of Capacity}

We also show that Algorithm G can be modified to compute the capacity of a completely positive operator $L$. Recall that the capacity of an operator $L$ is defined to be $\capac(L) = \text{inf} \: \Det(L(P))$, where the infimum is taken over all psd matrices $P$ with $\Det(P)=1$. The following theorem is proven as Theorem \ref{thm:computing-capacity} in Section~\ref{compute-capacity}. 

\begin{theorem} There is a $\poly(n,b,1/\eps)$ time algorithm for computing a $(1+\eps)$-multiplicative approximation to $\capac(L)$. Here $n$ is the dimension on which $L$ acts and $b$ denotes the bit-sizes involved in description of $L$. 
\end{theorem}

This is quite fascinating since we don't know of a convex formulation for $\capac(L)$. The idea for the computation of capacity is quite simple. We know that the capacity of a doubly stochastic operator is exactly $1$. We also know how operator scaling changes the capacity. If $L_{B,C}$ is a scaling of $L$ i.e.
$$
L_{B,C}(X) = B \cdot L(C \cdot X \cdot C^{\dagger}) \cdot B^{\dagger}
$$
then 
$$
\capac(L) = |\Det(B)|^{-2} \cdot |\Det(C)|^{-2} \cdot \capac(L_{B,C}) = |\Det(B)|^{-2} \cdot |\Det(C)|^{-2}
$$
if $L_{B,C}$ is doubly stochastic. Thus if we can find a doubly-stochastic-scaling of $L$, then we can 
compute $\capac(L)$ exactly. Algorithm G helps in finding an approximately-doubly-stochastic-scaling of $L$, 
which results in an approximation scheme for computing $\capac(L)$. It is a very intriguing open problem if a 
$(1+\eps)$ approximation to $\capac(L)$ can be computed in time $\poly(n,b,\log(1/\eps))$. In the 
aforementioned upcoming paper, we use this result (via a black-box reduction!) to approximating optimal 
constants in Brascamp-Lieb inequalities. 

Although, we don't know of a convex formulation for capacity, it does have some nice properties. Let us look at the Lagrangian for $\log(\capac(L))$:
$$
f(X,\lambda) = \log(\Det(L(X))) + \lambda \cdot \log(\Det(X))
$$
A critical point $C$ i.e. where the gradient $\nabla f(C,\lambda) = 0$ should satisfy $L^*\left(L(C)^{-1} \right) = C^{-1}$. This follows from 
$$
\nabla f(C,\lambda) = \Bigg(L^*\left(L(C)^{-1} \right) + \lambda C^{-1}, \log\left(\Det(C) \right)\Bigg)
$$
and the normalization condition
$$
\tr\left[ C \cdot  L^*\left(L(C) ^{-1}\right)\right] = \tr\left[ L(C) \cdot L(C)^{-1}\right] = n
$$
This implies by Lemma \ref{approximate_fixedpt} that 
$$
\capac(L) \ge \frac{\Det(L(C))}{\Det(C)}
$$
and hence any critical point is in fact a global minimum for $\capac(L)$! It would be interesting to see if our techniques can help solve other non-convex optimization problems which share the same property. 

Algorithm $G$ can also be thought of as an alternating minimization algorithm for computing capacity. First let us write capacity in an alternative way.
$$
n \cdot \capac(L)^{1/n} = \text{inf} \: \tr [L(X) \cdot Y] = \text{inf} \: \tr[X \cdot L^*(Y)] \: \: \text{s.t.} \: \: \Det(X) \ge 1, \Det(Y) \ge 1, X,Y \succ 0
$$
 While the constraints are convex (by log-concavity of determinant), the objective is quadratic. However if, we fix $X$ or $Y$, this program is convex. In fact, if we fix $X$, the optimum $Y$ is $\Det(L(X))^{1/n} \cdot L(X)^{-1}$. Similarly, if we fix $Y$, the optimum $X$ is $\Det(L^*(Y))^{1/n} \cdot L^*(Y)^{-1}$. Starting from $X_0 = I$, if we do alternate minimization, we get exactly Algorithm $G$ and our results imply that this process converges in polynomial number of steps!
 
From the description of capacity, it is not clear if it is continuous. The reason is that a priori, the optimizing matrices could be radically different for $\capac(L_1)$ and $\capac(L_2)$ even if $L_1$ and $L_2$ are quite close. We show that this not the case, essentially because $\capac(L)$ (as well as the optimizer) can be approximated by a simple iterative process namely Algorithm $G$! And Algorithm $G$ is clearly continuous in $L$. We prove this formally in Theorem \ref{capacity_continuity} of Section~\ref{sec:capacity_continuity}. In the aforementioned upcoming paper, we use this result to prove continuity of Brascamp-Lieb constants. 

\begin{remark} The continuity of capacity can also be proven via other methods and is already mentioned in \cite{gurvits2004}. But here we manage to get explicit bounds on the continuity parameter which we don't know how to get by methods other than analyzing Algorithm $G$. 
\end{remark}

\subsection*{Subsequent Work}
After our work, Derksen and Makam \cite{derksen2015} obtained polynomial degree bounds for the left-right action over any field! Using this bound, Ivanyos, Qiao and Subrahmanyam \cite{IQS15b} designed a completely different deterministic polynomial time algorithm for SINGULAR that works over all fields. This algorithm has a combinatorial/linear algebraic flavor, and has several important advantages over our algorithm, but also some limitations in comparison.
One clear advantage is of course that it works over every field. Another is that this algorithm can {\em certify} the non-commutative singularity of an input $L$ over $\Fxx$ by outputting a shrunk subspace when one exists! On the other hand, over the rationals, our algorithm computes capacity, and finds an operator scaling which is doubly stochastic. 

Using some of the techniques developed in \cite{IQS2015, IQS15b}, Bl\"{a}ser, Jindal and Pandey \cite{BlaserJP16} designed a deterministic PTAS for the commutative rank i.e. a $(1+\eps)$-approximation algorithm which runs in deterministic $n^{O(1/\eps)}$ time. This improves the factor-$2$ approximation algorithm given in this paper.

\subsection{Organization}

In Section~\ref{mainsection}, we describe Algorithm G for testing if quantum operators 
are rank-decreasing. We explain the notion of capacity of quantum 
operators, prove an explicit lower bound on the capacity of rank non-decreasing operators, and explain how it is used in the analysis to prove 
that Algorithm $G$ runs in polynomial time. In Section \ref{sec:cap_properties}, we study various properties of capacity, some of which are used in other places in the paper. In Section \ref{sec:bit_complexity_continuity}, we analyze the bit complexity of Algorithm $G$ and also prove an explicit bound on the continuity of capacity. In Section~\ref{compute-capacity}, we will show how a 
modification of Algorithm $G$ can be used to approximate capacity. 
In the appendix Section~\ref{computeNC}, we show how to compute the non-commutative rank.
We conclude in Section~\ref{open} with a short discussion and open problems.

%% file: AlgG.tex
This section is devoted to Algorithm $G$ and its analysis. We start with preliminaries about completely positive operators, their properties, and basic quantities associated with them. We then formally describe Algorithm $G$, and proceed to give a full analysis of its running time, proving the main theorem of this paper. Section \ref{sec:definitions_operators} contains definitions and properties of completely positive operators and their capacity. Section \ref{AlgG} describes the Algorithm $G$ and its analysis assuming an explicit lower bound on the capacity. Section \ref{sec:capacity_lb} contains the main theorem concerning the lower bound on capacity of rank non-decreasing operators. 

\subsection{Completely Positive Operators and Capacity}\label{sec:definitions_operators}

Given a complex matrix $A$, we will use $A^{\dagger}$ to denote the conjugate-transpose of $A$. For matrices 
with real entries this will just be $A^{T}$. For a matrix $A$, $\text{Im}(A)$ will denote the image of $A$ i.e. 
$\{v \in \mathbb{C}^n | v = Au \:\:\text{for some $u \in \mathbb{C}^n$}\}$.

\begin{definition}[Completely positive operators] An operator (or map) $T: M_n(\mathbb{C}) \to M_n(\mathbb{C})$ is called 
completely positive if there are $n \times n$ complex matrices 
$A_1, \ldots, A_m$ s.t. $T(X) = \sum_{i=1}^m A_i X A_i^{\dagger}$. The matrices $A_1,\ldots, A_m$ 
are called Kraus operators of $T$ (and they are not unique). $T$ is called completely positive trace 
preserving (cptp) if $T^*(I) = \sum_{i=1}^m A_i^{\dagger} A_i = I$.
\end{definition}

\begin{remark}
The above is actually not the usual definition of completely positive operators.  $T$ is defined to be positive if $T(X) \succeq 0$ whenever $X \succeq 0$. $T$ is completely positive if $I_k \tensor T$ is positive for all $k \ge 1$. Choi \cite{Choi} proved that an operator is completely positive iff 
it is of the form stated above.
\end{remark}

\begin{definition}[Tensor products of operators]\label{def:tensor} Given operators $T_1: M_{d_1}(\mathbb{C}) \to M_{d_1}(\mathbb{C})$ and $T_2: M_{d_2}(\mathbb{C}) \to M_{d_2}(\mathbb{C})$, we define their tensor product 
$T_1 \tensor T_2 : M_{d_1d_2}(\mathbb{C}) \to M_{d_1d_2}(\mathbb{C})$ in the natural way
$$
(T_1 \tensor T_2) (X \tensor Y) = T_1(X) \tensor T_2(Y)
$$
and extend by linearity to the whole of $M_{d_1d_2}(\mathbb{C})$. 
\end{definition}

\begin{definition}\label{dual,ds}
If $T(X) = \sum_{i=1}^m A_i X A_i^{\dagger}$ is a completely positive operator, we define its dual $T^*$ by
$T^*(X) = \sum_{i=1}^m A_i^{\dagger} X A_i$. If both $T$ and $T^*$ are trace preserving, namely $T(I) = T^*(I) = I$
then we call $T$ (and $T^*$) {\em doubly stochastic}.
\end{definition}

\begin{definition}[Rank Decreasing Operators] A completely positive operator $T$ is said to be rank-decreasing 
if there exists an $X \succeq 0$ s.t. 
$\text{rank}(T(X)) < \text{rank}(X)$. $T$ is said to be $c$-rank-decreasing if there exists an $X \succeq 0$ s.t. $\text{rank}(T(X)) \le \text{rank}(X)-c$. We will sometimes refer to operators that are not rank decreasing as rank non-decreasing.
\end{definition}

Now that we defined completely positive operators, we define their {\em capacity}, which is a very important complexity measure
of such operators suggested in~\cite{gurvits2004}. Its evolution will be central to the complexity analysis of our algorithm.

\begin{definition}[Capacity]\cite{gurvits2004} The capacity of a completely positive operator $T$, denoted by $\capac(T)$, is defined as
$$\capac(T) = \text{inf} \{ \text{Det}(T(X)) : \text{$X \succ 0$, Det$(X) = 1$}\}$$
\end{definition}

This notion of capacity has a very interesting history. Some special cases of capacity were defined in \cite{gur-sam1} and \cite{gur-sam2}. 
It was then extended to hyperbolic polynomials and this also led to a resolution (and extremely elegant proofs) of the 
Van der Waerden Conjecture for Mixed Discriminants \cite{gur_Waerden}.

The next proposition shows how capacity changes when linear transformations are applied (as in the algorithm)
to the completely positive operator.

\begin{proposition}[\cite{gurvits2004}]{\label{cap_mul}} Let $T$ be the operator defined by $A_1,\ldots,A_m$ and let 
$T_{B,C}$ be the operator defined by $B A_1 C,\ldots, B A_m C$, where $B,C$ are invertible matrices. Then 
$$
\capac(T_{B,C}) = |\Det(B)|^2 |\Det(C)|^2 \capac(T)
$$
\end{proposition}

\begin{proof}
\begin{align*}
\capac(T_{B,C}) &= \text{inf} \left\{ \text{Det} \left( \sum_{i=1}^m B A_i C X C^{\dagger} A_i^{\dagger} B^{\dagger}\right) : \text{$X \succ 0$, Det$(X) = 1$} \right\} \\
&= |\Det(B)|^2 \cdot \text{inf} \left\{ \text{Det} \left( \sum_{i=1}^m A_i C X C^{\dagger} A_i^{\dagger} \right) : \text{$X \succ 0$, Det$(X) = 1$} \right\} \\
&= |\Det(B)|^2 \cdot \text{inf} \left\{ \text{Det} \left( \sum_{i=1}^m A_i X A_i^{\dagger} \right) : \text{$X \succ 0$, Det$(X) = |\Det(C)|^2$} \right\} \\
&= |\Det(B)|^2 \cdot |\Det(C)|^2 \cdot \capac(T)
\end{align*}
\end{proof}

The next proposition gives a useful upper bound on the capacity of trace preserving completely positive operators; this will be used for the convergence analysis of the algorithm.

\begin{proposition}[Capacity Upper Bound] {\label{cap_bound_1}} Let $T$ be a completely positive operator 
with Kraus operators $A_1,\ldots, A_m$ (which are $n \times n$ complex matrices). Also suppose that either 
$\sum_{i=1}^m A_i^{\dagger} A_i = I$ or $\sum_{i=1}^m A_i A_i^{\dagger} = I$. Then $\capac(T) \le 1$.
\end{proposition}

\begin{proof}
Note that $\tr(T(I)) = n$ in either case.
\begin{align*}
\capac(T) \le \Det(T(I)) &\le \left( \tr(T(I))/n \right)^n \\
&= 1
\end{align*}
The second inequality follows from the AM-GM inequality.
\end{proof}

\subsection{Algorithm $G$ and its convergence rate}\label{AlgG}

We now describe Algorithm $G$ to test if a completely positive operator $T$ (given in terms of Kraus operators $A_1,\ldots,A_m$) is rank non-decreasing (equivalently properties $(1)-(8)$ in Theorem \ref{Equivalences} one of which is that $A_1,\ldots,A_m$ admit no shrunk subspace). We will then analyze its convergence rate, namely the number of scaling iterations needed as a  function of the input size. In section \ref{sec:bit_complexity_continuity}, we will continue with the finer analysis of the bit complexity of this algorithm.

Since the property of $A_1,\ldots, A_m$ having a shrunk subspace or not remains invariant if we replace $A_1, \ldots, A_m$ by a basis spanning the subspace spanned by $A_1,\ldots, A_m$, we can always assume wlog that $m \le n^2$. Suppose the maximum bit size of the entries of $A_1,\ldots,A_m$ is $b$. Since scaling the matrices $A_1,\ldots,A_m$ doesn't change the problem, we can assume that $A_1,\ldots,A_m$ have integer entries of magnitude at most $M = 2^{O(b)}$. 

\noindent Algorithm $G$ below is essentially Gurvits' algorithm \cite{gurvits2004} for Edmonds' problem for subspaces of matrices having Edmonds-Rado property. It is a non-commutative generalization of the Sinkhorn scaling procedure to scale non-negative matrices to doubly stochastic matrices (see for example \cite{LSW} and references therein). The algorithm alternates applying a ``normalizing'' basis change from the left and right to the given matrices, so as to alternately make the operator or its dual trace preserving. The idea is that this process will converge to a doubly stochastic operator iff it is not rank-decreasing, and furthermore, we can bound the number $t$
of iterations by $\poly(n, \log(M))$. We will use the following to measure of our operator to being doubly stochastic.

\begin{definition}
$$
\ds(T) = \tr \left[ \left(T(I) - I\right)^2 \right] + \tr \left[ \left( T^*(I)   - I\right)^2 \right]
$$
\end{definition}

\begin{definition} Given a completely positive operator $T$, define its {\em right normalization} $T_R$ as follows:
\begin{align*}
T_R(X) = T \left(T^*(I)^{-1/2} X T^*(I)^{-1/2} \right)
\end{align*}

\end{definition}

\noindent Note that $T_R^*(I) = I$.

\begin{definition} Given a completely positive operator $T$, define its {\em left normalization} $T_L$ as follows:
\begin{align*}
T_L(X) =  T(I)^{-1/2} T(X) T(I)^{-1/2}
\end{align*}

\end{definition}

\noindent Note that $T_L(I) = I$. These operations are referred to as row and column operations in \cite{gurvits2004}. 

We next prove that if a completely positive operator $T$ is doubly stochastic, or close to being one in this measure, then it is rank non-decreasing. This will be important for the termination condition of Algorithm $G$.

\begin{theorem}[\cite{gurvits2004}]{\label{DS_full}} If $T$ is a completely positive operator which is right $($or left$)$ normalized satisfying $\ds(T) \le 1/(n+1)$, then $T$ is rank non-decreasing.
\end{theorem}

\begin{proof} Wlog assume that $T(I) = I$ and 
$$
\tr \left[ \left( T^*(I)-I\right)^2\right] \le 1/(n+1)
$$  
Suppose $X \succeq 0$ is a psd matrix s.t. $\text{Rank}(X) = r$. We would like to prove that $\text{Rank}(T(X)) \ge r$. Let
$$
X = \sum_{i=1}^r \lambda_i v_i v_i^{\dagger}
$$
be the eigenvalue decomposition of $X$ where $v_i$'s are orthonormal and $\lambda_i > 0$. Then
$$
T(X) = \sum_{i=1}^r \lambda_i T\left( v_i v_i^{\dagger} \right)
$$
Let us denote $T\left( v_i v_i^{\dagger} \right)$ by $R_i$. Since $T(I) = I$, we get that $\sum_{i=1}^r R_i \preceq T(I) = I$. Also note that 
$$
\text{Rank} \left( \sum_{i=1}^r \lambda_i R_i \right) = \text{Rank} \left( \sum_{i=1}^r R_i \right)
$$
This is because $R_i$'s are psd matrices and hence a vector is in the kernel of $\sum_{i=1}^r \lambda_i R_i$ iff it is in the kernel of all the $R_i$'s. Because of $\sum_{i=1}^r R_i \preceq I$, we get that
$$
\text{Rank} \left( T(X)\right) = \text{Rank} \left( \sum_{i=1}^r R_i \right) \ge \tr \left[  \sum_{i=1}^r R_i \right]
$$
Suppose $T^*(I) = I + \Delta$. We know that $\tr[\Delta^2] \le 1/(n+1)$. Then 
$$
\tr[R_i] = \tr \left[ T^*(I) v_i v_i^{\dagger} \right] = 1 + \tr\left[ \Delta v_i v_i^{\dagger} \right]
$$
Let $P$ denote the projection onto $\{v_1,\ldots,v_r\}$. Adding the above equations for all $i \in \{1,\ldots,r\}$, we get
\begin{align*}
\tr \left[  \sum_{i=1}^r R_i \right] &= r + \tr\left[ \Delta P \right] \\
&\ge r - \tr[\Delta^2]^{1/2} \tr[P^2]^{1/2} \\
&\ge r - \sqrt{\frac{r}{n+1}} 
\end{align*}
Since $r < n+1$, this gives us 
$$
\text{Rank} \left( T(X)\right) > r-1
$$
which completes the proof since $\text{Rank} \left( T(X)\right)$ is an integer. 
\end{proof}

We are now ready to present the algorithm.

\begin{Algorithm}
\textbf{Input}: Completely positive operator $T$ given in terms of Kraus operators $A_1, \ldots, A_m \in \mathbb{Z}^{n \times n}$. Each entry of $A_i$ has absolute value at most $M$. \\
\textbf{Output}: Is $T$ rank non-decreasing?  

\begin{enumerate}
\item Check if $T(I)$ and $T^*(I)$ are singular. If any one of them is singular, then output that the operator is rank decreasing, otherwise proceed to step 2. 
\item Perform right and left normalizations on $T$ alternatively for $t$  steps. Let $T_j$ be the operator after $j$ steps. Also let $\eps_j = \ds(T_j)$. 
\item Check if $\min \{\eps_j : 1 \le j \le t\} \le 1/(6n)$. If yes, then output that the operator is rank non-decreasing otherwise output rank decreasing.
\end{enumerate}
\caption{Algorithm $G$}
\label{Gurvits_alg}
\end{Algorithm}

\begin{remark}
In algorithm $G$, the operators are maintained in terms of the Kraus operators and it is easy to see 
the effect of right and left normalizations on the Kraus operators. In fact, they are named such 
since for a right (left) normalization, the Kraus operators multiplied on the right (left) by 
$T^*(I)^{-1/2} \left( T(I)^{-1/2} \right)$. 
\end{remark}

The main objective is to analyze the minimum number of steps $t$ for which this algorithm terminates with the correct answer.
The following theorem of~\cite{gurvits2004} gives the time analysis of the algorithm {\em assuming} 
an initial lower bound on the capacity of the input completely positive operator. In the next subsection we will 
prove our main result, an appropriate lower bound, which shows that the algorithm terminates with the correct answer in 
polynomial time.

\begin{theorem}[\cite{gurvits2004}]{\label{gurvitsalg}} Let $T$ be a completely positive operator 
that is rank non-decreasing, with Kraus operators $A_1, \ldots, A_m$, which are 
$n \times n$ integer matrices such that each entry of $A_i$ has absolute value at most $M$. 
Suppose $T_R$ is the right normalization of $T$. If $\capac(T_R) \ge f(n, M)$, then 
Algorithm $G$ when applied for at least $t = 2 + 36n \cdot \log(1/f(n, M))$ steps is correct. 
\end{theorem}

\noindent For completeness sake, we provide a full proof of this theorem. Again, this analysis 
follows similar ones for the classical Sinkhorn iterations, e.g. as in ~\cite{LSW,GurYianilos}. 
Basically, capacity increases by a factor roughly $1+1/36n$ per iteration as long as it is 
not too close to 1.

\begin{proof}
	If either $T(I)$ or $T^*(I)$ is singular, then $T$ is rank-decreasing. 
	When $T(I)$ is singular, $T$ decreases the rank of $I$. When $T^*(I)$ is singular, any vector in the kernel of $T^*(I)$ lies in the kernels of all the $A_i$'s. If $T(I)$ and $T^*(I)$ are both non-singular, it is easy to verify that $T_j(I)$ and $T_j^*(I)$ will remain non-singular for all $j$ and hence step 2 is well defined. Also using Theorem \ref{DS_full} and the fact that right and left normalizations don't change the property of being rank decreasing, Algorithm $G$ will always output rank decreasing if $T$ is rank-decreasing.

So what is left to prove is if $T$ is rank non-decreasing, then $\min \{\eps_j : 1 \le j \le t\} \le 1/6n$. Assume to the contrary that it is not. 
Denote by $\capac_j$ to be the capacity of the operator $T_j$. By Lemma~\ref{lem:quant-progress} below (which essentially is a robust version of the AM-GM inequality),
if $\eps_j > 1/6n$, then $\capac_{j+1} \ge \expon(1/36n) \cdot \capac_j$.

%

\noindent From the assumption of the theorem, we know that $\capac_1 \ge f(n, M)$. Also it is easy to 
see that $\capac_j \le 1$ for all $j$ (Proposition \ref{cap_bound_1}). However by the 
assumption that $\min \{\eps_j : 1 \le j \le t\} > 1/6n$ and the increase in capacity per iteration seen above, 
we get that 
$$
1 \ge \capac_t \ge \expon \left(\frac{t-1}{36n} \right) \cdot \capac_1 = \expon \left(\frac{t-1}{36n} \right) \cdot f(n,M)
$$
Plugging in $t = 2 + 36n \cdot \log(1/f(n, M))$ gives us the required contradiction.

\end{proof}

In the main Theorem \ref{capacity_lb} below we will prove that the quantity $f(n, M)$ used in the statement of Theorem \ref{gurvitsalg} is $\ge \expon(-4n \log(Mmn))$, which will prove that the number of {\em iterations} needed in Algorithm $G$ is $O(n^2 \log(Mmn))$. But this alone doesn't guarantee that the algorithm is actually polynomial time, since the bit sizes of numbers involved might get exponential.
As it happens, simple truncation suffices to overcome this problem, as shown in~\cite{gurvits2004}, and we reproduce this analysis in Subsection~\ref{bit_complexity} for completeness.

\subsection{Main Theorem: The Lower Bound on Capacity}\label{sec:capacity_lb}

\noindent In this subsection we prove our main theorem, a lower bound on capacity of an operator in terms of its description size, in Theorem \ref{capacity_lb}. Before diving into the proof we provide a high level plan. We will first prove that if a completely positive operator with integer entries is rank non-decreasing (that is has positive capacity), then the capacity is actually non-negligible.
Our starting point  (Theorem~\ref{Equivalences}) is the statement of the equivalence between the rank non-decreasing property and an algebraic condition (non-vanishing of a certain determinant). Using Alon's Combinatorial Nullstellensatz, we can ensure small coefficients in the algebraic condition above. 
We then state and prove  (Lemma~\ref{lemm:matrix_CS})  a Cauchy-Schwartz inequality for matrices needed next.
The main result (Theorem~\ref{lem:cap-square} ) proceeds by applying
 an appropriate scaling to the original integral operator, which  reduces it to one that preserves the identity matrix. The latter property (which provides a trace bound), the integrality of the original operator (which provides an integral determinantal lower bound) and the multiplicativity of capacity under scaling combine (via the above inequality) to give the desired bound. We now proceed with the details of this plan.

We will need the fact that (nonzero) 
polynomials of degree $d$ cannot vanish on all points with non-negative integer
coordinates with sum $\leq d$. This follows from Alon's combinatorial nullstellensatz~\cite[Theorem 1.2]{Alon_CN}.

\begin{lemma}[\cite{Alon_CN}]\label{lem:small-SZ}
	If $p(x_1, \ldots, x_n) \in \C[x_1, \ldots, x_n]$ is a (nonzero) polynomial of degree $d$,
      where the individual degree of the variable $x_i$ is at most $d_i$, 
	then there exists $(a_1, \ldots, a_n) \in \mathbb{Z}_{\ge 0}^n$ such that 
	$\dst\sum_{i=1}^n a_i \leq d$ and $a_i \le d_i$, for which $p(a_1, \ldots, a_n) \neq 0$.
\end{lemma}

As a corollary of Theorem~\ref{Equivalences} and of Lemma~\ref{lem:small-SZ}, we get the following:

\begin{corollary}\label{thm:dim-bounds} Let $A_1, \ldots, A_m$ be $n \times n$ complex matrices s.t. 
	the completely positive operator $T_A$ defined by $A_1,\ldots,A_m$ is rank non-decreasing. 
	Then there exist matrices $D_1, \ldots, D_m \in \cM_d(\C)$ for some $d$
	such that $$ \Det(D_1 \otimes A_1 + \cdots + D_m \otimes A_m) \neq 0.  $$
	Furthermore, $D_1,\ldots,D_m$ can be chosen to be matrices with non-negative integer entries s.t. 
$$
\sum_{j=1}^m \sum_{k,l=1}^d D_j(k,l)^2 \le n^2 d
$$
\end{corollary}

\begin{proof}
	By equivalences (3) and (7) of Theorem~\ref{Equivalences}, we know that there exist matrices $F_1, \ldots, F_m \in \cM_d(\C)$, for some $d$, 
	such that $\Det(F_1 \otimes A_1 + \cdots + F_m \otimes A_m) \neq 0$. This implies that the 
	polynomial $p(X_1, \ldots, X_m) = \Det(X_1 \otimes A_1 + \cdots + X_m \otimes A_m)$,
	where $X_i$ are symbolic $d \times d$ matrices, is nonzero. Hence, by Lemma~\ref{lem:small-SZ},
	we know that there exist matrices of positive integers $D_1, \ldots, D_m$ such that 
	$\Det(D_1 \otimes A_1 + \cdots + D_m \otimes A_m) \neq 0$ and that 
	$$
\sum_{j=1}^m \sum_{k,l=1}^d D_j(k,l) \leq \deg(p) = nd
$$
$$
D_j(k,l) \le n
$$ 
This is because the total degree of $p$ is $nd$ and the individual degree of each variable $X_j(k,l)$ in $p$ is at most $n$. This implies the desired upper bound.
\end{proof}

We will also need the following Cauchy-Schwarz like inequality:

\begin{lemma}\label{lemm:matrix_CS}
Let $C_1,\ldots,C_m$ and $D_1,\ldots,D_m$ be $d_1$ and $d_2$ dimensional complex matrices, respectively. Then
$$
\tr\left[ \left( \sum_{i=1}^m C_i \otimes D_i\right) 
\left( \sum_{j=1}^m C_j^{\dagger} \otimes D_j^{\dagger}\right)\right] \le 
\tr\left[ \sum_{i=1}^m C_i C_i^{\dagger}\right]  \cdot \tr\left[ \sum_{j=1}^m D_j D_j^{\dagger}\right]
$$
\end{lemma}

\begin{proof}
We start with the inequality:
\begin{align*}
0 &\le \tr\left[ \left( C_i \otimes D_j - C_j \otimes D_i \right) \left( C_i^{\dagger} \otimes D_j^{\dagger} - C_j^{\dagger} \otimes D_i^{\dagger} \right)\right] \\
&= \tr\left[ C_i C_i^{\dagger} \otimes D_j D_j^{\dagger} + C_j C_j^{\dagger} \otimes D_i D_i^{\dagger}\right] - \tr\left[ C_i C_j^{\dagger} \otimes D_j D_i^{\dagger} + C_j C_i^{\dagger} \otimes D_i D_j^{\dagger}\right] \\
&= \tr\left[ C_i C_i^{\dagger}\right] \cdot \tr\left[ D_j D_j^{\dagger}\right] +  \tr\left[ C_j C_j^{\dagger}\right] \cdot \tr\left[ D_i D_i^{\dagger}\right] - \tr\left[ C_i C_j^{\dagger} \otimes D_i D_j^{\dagger}\right] - \tr\left[ C_j C_i^{\dagger} \otimes D_j D_i^{\dagger}\right]
\end{align*}
The first inequality is just $\tr\left[ Y Y^{\dagger}\right] \ge 0$. The second equality follows from linearity of trace, $\tr\left[Y\otimes Z\right] = \tr[Y]\cdot \tr[Z]$ and $\tr[XY^{\dagger}] = \tr[YX^{\dagger}]$. Rearranging the above, we get the following inequality:
\begin{align} 
\tr\left[ C_i C_j^{\dagger} \otimes D_i D_j^{\dagger}\right] + \tr\left[ C_j C_i^{\dagger} \otimes D_j D_i^{\dagger}\right] \le  \tr\left[ C_i C_i^{\dagger}\right] \cdot \tr\left[ D_j D_j^{\dagger}\right] +  \tr\left[ C_j C_j^{\dagger}\right] \cdot \tr\left[ D_i D_i^{\dagger}\right] \label{eqn:ankit20}
\end{align}
Summing Equation (\ref{eqn:ankit20}) over pairs $(i,j)$ with $i<j$, we obtain the following:
\begin{align*}
\sum_{i \neq j} \tr\left[ C_i C_j^{\dagger} \otimes D_i D_j^{\dagger}\right] \le \sum_{i \neq j} \tr\left[ C_i C_i^{\dagger}\right] \cdot \tr\left[ D_j D_j^{\dagger}\right]
\end{align*}
Adding $\sum_i \tr\left[ C_i C_i^{\dagger}\right] \cdot \tr\left[ D_i D_i^{\dagger}\right]$ to both sides completes the proof (due to linearity of trace). 
\end{proof}





Now we are ready to prove our main theorem for operators with integral Kraus operators.

\begin{theorem}[\textbf{Capacity of Square Operators}]\label{lem:cap-square}
	Suppose $T_A$ is a completely positive operator which is rank non-decreasing and has Kraus operators 
	$A_1, \ldots, A_m \in \cM_n(\Z)$. In this case:
	$$ \capac(T_A) \ge  \frac{1}{n^{2n}} $$
\end{theorem}

\begin{proof}
Since $T_A$ is rank non-decreasing, by Corollary~\ref{thm:dim-bounds}, there exist non-negative 
integer matrices $D_1,\ldots,D_m$ of some dimension $d$ s.t. $\Det(D_1 \otimes A_1 + \cdots + D_m \otimes A_m) \neq 0$ and also 
\begin{align}
\sum_{j=1}^m \sum_{k,l=1}^d D_j(k,l)^2 \le n^2 d \label{eqn:ankit25}
\end{align}
Since $A_1,\ldots,A_m$ are also integer matrices, $\Det(D_1 \otimes A_1 + \cdots + D_m \otimes A_m)$ is a 
non-zero integer and hence
\begin{align}
|\Det(D_1 \otimes A_1 + \cdots + D_m \otimes A_m)| \ge 1 \label{eqn:ankit1}
\end{align}
Let $X \succ 0$ be such that $\Det(X)=1$. Consider the matrices $C_i = T_A(X)^{-1/2}A_i X^{1/2}$. This is a scaling intended so that the operator $T_C$ defined by the matrices $C_i$ preserves identity i.e. satisfies
\begin{align*}
T_C(I) = \sum_{i=1}^m C_i C_i^{\dagger} 
&= T_A(X)^{-1/2} \left( \sum_{i=1}^m A_i X^{1/2} X^{1/2}A_i^{\dagger} \right) T_A(X)^{-1/2}\\
&= I_n
\end{align*}
Hence 
\begin{align}
\tr\left[\sum_{i=1}^m C_i C_i^{\dagger} \right] = n \label{eqn:ankit37}
\end{align}
Let $Y$ denote the matrix $D_1 \otimes C_1 + \cdots + D_m \otimes C_m$. Then 
\begin{align}
|\Det(Y)| &= |\Det(D_1 \otimes C_1 + \cdots + D_m \otimes C_m)|  \\
&= |\Det(D_1 \otimes A_1 + \cdots + D_m \otimes A_m)| \cdot \Det(X)^{d/2} \cdot 
\Det\left( T_A(X)\right)^{-d/2}\nonumber \\
&\ge  \Det\left( T_A(X)\right)^{-d/2} \label{eqn:ankit2}
\end{align}
The equality follows from multiplicativity of determinant and the fact that 
$\Det\left( I_d \otimes Y\right) = \Det(Y)^d$. The inequality follows from $\Det(X)=1$ and 
Equation (\ref{eqn:ankit1}). 
Hence we obtain
\begin{align}
\Det \left( Y Y^{\dagger}\right) = \left| \Det(Y)\right|^2 \ge \Det\left( T_A(X)\right)^{-d} \label{eqn:ankit21}
\end{align}
On the other hand, by the AM-GM inequality
\begin{align}
\Det \left( Y Y^{\dagger}\right) \le \left(\frac{\tr\left[ Y Y^{\dagger}\right]}{n d}\right)^{n d} \label{eqn:ankit22}
\end{align}
Now consider
\begin{align}
\tr\left[ Y Y^{\dagger}\right] &=\tr\left[ \left( \sum_{i=1}^m D_i \otimes C_i\right) \left( \sum_{j=1}^m D_j^{\dagger} \otimes C_j^{\dagger}\right)\right] \nonumber \\
&\le \tr\left[ \sum_{i=1}^m D_i D_i^{\dagger}\right]  \cdot \tr\left[ \sum_{j=1}^m C_j C_j^{\dagger}\right] \nonumber \\
&= \left( \sum_{i=1}^m \sum_{k,l=1}^d D_i(k,l)^2 \right) \cdot n \nonumber \\
&\le n^3 d \label{eqn:ankit24}
\end{align}
The first inequality follows from Lemma \ref{lemm:matrix_CS}. The second equality follows from Equation (\ref{eqn:ankit37}). The second inequality follows from Equation (\ref{eqn:ankit25}). Combining equations (\ref{eqn:ankit21}), (\ref{eqn:ankit22})  and (\ref{eqn:ankit24}), we obtain:
\begin{align*}
\Det\left( T_A(X)\right) \ge \frac{1}{n^{2n}}
\end{align*}
Since $X$ was an arbitrary matrix s.t. $X \succ 0$ and $\Det(X)=1$, we obtain the desired lower bound on $\capac(T_A)$. 
\end{proof}

\begin{remark} A slightly stronger bound can be obtained by using quantum permanents \cite{gurvits2004} in the case when the $D_i$'s are of constant dimension.
\end{remark}

Theorem \ref{lem:cap-square} immediately implies the following capacity lower bound that we need.

\begin{theorem}{\label{capacity_lb}}
Suppose $T$ is a cptp map which is rank non-decreasing and is obtained by right normalization of an operator with Kraus operators 
$A_1, \ldots, A_m$, which are $n \times n$ integer matrices such that each entry of $A_i$ has absolute value at most $M$. 
Then $\capac(T) \ge \expon(-4 n \log(Mmn))$. 
\end{theorem}

\begin{proof}
Let $T_A$ be the operator defined by $A_1,\ldots,A_m$. Since $T$ is the right normalization of $T_A$, by Proposition \ref{cap_mul}, we get that
\begin{align}
\capac(T) = \frac{\capac(T_A)}{\Det \left( T^*(I)\right)} \label{eqnankit:100}
\end{align}
Note that $T^*(I)$ is a psd matrix. Also 
\begin{align*}
\tr \left[ T^*(I)\right] &= \sum_{i=1}^m \tr \left[ A_i^{\dagger} A_i\right] \\
&= \sum_{i=1}^m \sum_{k,l=1}^n |A_i(k,l)|^2 \\
&\le M^2 m n^2 
\end{align*}
This implies (via the AM-GM inequality) that
\begin{align}
\Det \left( T^*(I)\right) &\le \left( \frac{\tr \left[ T^*(I)\right]}{n}\right)^n \nonumber \\
&\le (M^2mn)^n \label{eqnankit:101}
\end{align}
Combining Equations (\ref{eqnankit:100}) and (\ref{eqnankit:101}), and Theorem \ref{lem:cap-square} gives the desired bound. 
\end{proof}

%% file: cap_properties.tex
In this section, we prove some interesting properties of capacity. In subsection \ref{subsec:cap_ds}, we prove that the capacity of operators which are close to doubly stochastic is close to $1$. This is used in the proof of Theorem \ref{mult_cap_hard}. In subsection \ref{subsec:charac_opt}, we characterize the almost optimal points of capacity in terms of approximate fixed points of an operator. This will be used later in section \ref{sec:bit_complexity_continuity} in the proof of continuity of capacity (theorem \ref{capacity_continuity}). In \ref{subsec:mult_cap}, we prove a multiplicativity property of capacity under tensor products of operators. This property was used in the lower bound on capacity in the previous version of this paper. While it is no longer needed for this purpose now, it is an interesting and intriguing fact by itself, so we include a proof. 

\subsection{Capacity of almost doubly stochastic operators}\label{subsec:cap_ds}

The following definition of capacity for non-negative matrices is due to \cite{GurYianilos}. They used it to analyze Sinkhorn's algorithm for matrix scaling. 

\begin{definition}[Capacity of non-negative matrices] Given a non-negative matrix $A$, its capacity is defined as follows:
$$
\capac(A) = \inf \left\{ \prod_{i=1}^n (Ax)_i : \prod_{i=1}^n x_i = 1, x > 0\right\} 
$$
\end{definition}

The next lemma states that the capacity of almost doubly stochastic matrices is close to $1$.

\begin{lemma}{\label{cap_almostDS:classical}}
Let $A$ be a non-negative $n \times n$ matrix s.t. the row sums of $A$ are $1$ and 
$$
\sum_{i=1}^n (c_i-1)^2 \le \eps
$$
where $c_1,\ldots,c_n$ are the column sums of $A$. Then 
$$
\capac(A) \ge (1-\sqrt{n\eps})^n
$$
\end{lemma}

\begin{proof}
We start by proving that for doubly stochastic matrices $B$, $\capac(B) \ge 1$. This will follow from concavity of log. Indeed, let $x \in \mathbb{R}^n$ s.t. $\prod_{i=1}^n x_i = 1$ and $x > 0$. Then 
\begin{align*}
\sum_{i=1}^n \log \left((Bx)_i \right) &= \sum_{i=1}^n \log \left( \sum_{j=1}^n B_{i,j} x_j\right) \\
&\ge \sum_{i=1}^n \sum_{j=1}^n B_{i,j} \log(x_j) \\
&= \sum_{j=1}^n \log(x_j) \sum_{i=1}^n B_{i,j} \\
&= \sum_{j=1}^n \log(x_j) \\
&= 0
\end{align*}
For the inequality, we used concavity of log and the fact that row sums of $B$ are $1$. For the third equality, we used that the column sums of $B$ are $1$. 

Now we need to move on to $A$ which is almost doubly stochastic. A direct argument like for the doubly stochastic case does not work. We need to first prove the following claim (\cite{LSW}):

\begin{claim}{\label{almostDS:decomp}} There exists a decomposition of $A = \lambda B + Z$, where $B$ is doubly stochastic, $Z$ is non-negative and $\lambda \ge 1-\sqrt{n \eps}$. 
\end{claim}

Given the claim, it is easy to complete the proof of the lemma. 
$$
\capac(A) \ge \lambda^n \capac(B) \ge (1-\sqrt{n\eps})^n
$$
So we end with a proof of the claim from \cite{LSW}. We will first prove that there is a decomposition $A = D + Z$ where $D$ is multiple of a doubly stochastic matrix and $\text{per}(Z)=0$. Initially, we start with $D=0$ and $Z=A$. As long as $\text{per}(Z)>0$, there is an $\alpha > 0$ and a permutation matrix $P$ s.t. $Z' = Z-\alpha P$ is nonnegative and the number of non-zero entries in $Z'$ is strictly less than that of $Z$. Replace $D$ by $D' = D + \alpha P$ and $Z$ by $Z'$. Note that $D'$ is also a multiple of a doubly stochastic matrix. After at most $n^2$ iterations, we will find a decomposition $A = \lambda B + Z$, where $\text{per}(Z)=0$ (and $B$ is doubly stochastic). We will now prove that $\lambda \ge 1-\sqrt{n\eps}$. If $\lambda = 1$, then we are already done, so assume $\lambda < 1$. In this case, consider the matrix $C = \frac{Z}{1-\lambda}$. Row sums of $C$ are $1$ and the $i^{\text{th}}$ column sum $c_i' = \frac{c_i - \lambda}{1-\lambda}$. Then
$$
\sum_{i=1}^n (c_i'-1)^2 = \frac{1}{(1-\lambda)^2} \sum_{i=1}^n (c_i - 1)^2 \le \frac{\eps}{(1-\lambda)^2}
$$
Since $\text{per}(C)=0$ and row sums are $1$, it follows that (see Lemma 5.2 in \cite{LSW})
$$
\sum_{i=1}^n (c_i'-1)^2 \ge 1/n
$$
and hence 
$$
\frac{\eps}{(1-\lambda)^2} \ge 1/n
$$
which implies $\lambda \ge 1-\sqrt{n \eps}$
\end{proof}

The next lemma says that capacity of almost doubly stochastic operators is close to $1$. The proof will proceed by reducing the operator case to the non-negative matrix case established above.
\begin{lemma}\label{cap_almostDS:quantum}
Let $T$ be a positive operator acting on $n \times n$ matrices such that $\tr[(T(I)-I)^2] \le \eps$ and $T^*(I)=I$ (equivalently $T$ is trace-preserving). Then $\capac(T) \ge (1 - \sqrt{n \eps})^n$. 
\end{lemma}

\begin{remark} The proof can be adapted to obtain a similar statement when both $T(I), T^*(I) \approx I$.
\end{remark}

\begin{proof}
Let $X$ be a positive definite matrix s.t. $\Det(X)=1$. Let 
$$
X = \sum_{j=1}^n \lambda_j v_j v_j^{\dagger}
$$
be an eigenvalue decomposition of $X$ with $v_1,\ldots,v_n$ forming a orthonormal basis. Then
$$
T(X) = \sum_{j=1}^n \lambda_j T(v_j v_j^{\dagger})
$$ 
Let us denote $T(v_j v_j^{\dagger})$ by $R_j$. Then since $T$ is trace preserving, we have that $\tr(R_j) = 1$ for all $j$.
Also let
$$
T(X) = \sum_{i=1}^n \sigma_i u_i u_i^{\dagger}
$$
be an eigenvalue decomposition for $T(X)$ with $u_1,\ldots,u_n$ forming a orthonormal basis. It follows that for all $i$, 
$$
\sigma_i = \sum_{j=1}^n \lambda_j u_i^{\dagger} R_j u_i
$$
Let $A$ denote the non-negative matrix s.t. $A_{i,j} = u_i^{\dagger} R_j u_i$. Then $\sigma = A \lambda$. Also the column sums of $A$ are all $1$ since $\tr(R_j) = 1$ for all $j$. We will also prove that 
$$
\sum_{i=1}^n (r_i-1)^2 \le \eps
$$
where $r_1,\ldots,r_n$ denote the column sums of $A$. Then the lemma follows from Lemma \ref{cap_almostDS:classical} (applied with row sums switched with column sums) and the facts that $\Det(T(X)) = \prod_{i=1}^n \sigma_i$ and $\Det(X) = \prod_{i=1}^n \lambda_i$. Now note that
\begin{align*}
r_i = \sum_{j=1}^n A_{i,j} 
= \sum_{j=1}^n u_i^{\dagger} R_j u_i
= u_i^{\dagger} \left( \sum_{j=1}^n T(v_j v_j^{\dagger})\right) u_i
= u_i^{\dagger} T\left( \sum_{j=1}^n v_j v_j^{\dagger}\right) u_i
= u_i^{\dagger} T(I) u_i
\end{align*}
Hence
\begin{align*}
\sum_{i=1}^n (r_i-1)^2 = \sum_{i=1}^n \left(u_i^{\dagger} T(I) u_i-1 \right)^2 &= \sum_{i=1}^n \left(u_i^{\dagger} (T(I)-I)u_i \right)^2 \\
&\le \sum_{i=1}^n u_i^{\dagger} (T(I)-I)^2 u_i\\
&= \tr[(T(I)-I)^2] \le \eps
\end{align*}
The first inequality can be proved via convexity of square.
\end{proof}

We also prove an easy lemma for the other direction: operators with capacity almost $1$ are almost doubly stochastic.

\begin{lemma}\label{cap_almost1}
Let $T$ be a positive operator acting on $n \times n$ matrices s.t. $\capac(T) \ge \expon(-\delta)$, $\delta \le 1/6$ and also $T^*(I)=I$. Then $\ds(T) = \tr[(T(I)-I)^2] \le 6 \delta$.
\end{lemma}

\begin{proof}
$\capac(T) \ge \expon(-\delta)$ implies that $\Det(T(I)) \ge \expon(-\delta)$. Also $\tr[T(I)] = \tr\left[ T^*(I)\right] = n$. Suppose $ \tr[(T(I)-I)^2] = \alpha$. Then Lemma \ref{lem:prod-upper-bd} below implies that 
$$
\Det(T(I)) \le \expon(-\alpha/6)
$$
This implies that $\alpha \le 6 \delta$.
\end{proof}

\begin{remark}
Note that the parameters in Lemmas \ref{cap_almostDS:quantum} and \ref{cap_almost1} don't match and that is because capacity and the distance to doubly stochasticity don't characterize each other exactly. 
\end{remark}

\subsection{Characterization of approximate optimizers for capacity}\label{subsec:charac_opt}

The following definition will be useful.

\begin{definition}[Fixed points] Let $T$ be a completely positive operator. Define $\text{Fixed}(T,\eps)$ to be the set of hermitian positive-definite matrices which are $\eps$-approximate fixed point of the operator $X \rightarrow T^*(T(X)^{-1})^{-1}$ i.e. $C \in \text{Fixed}(T,\eps)$ if the following holds:
$$
\tr \left[ \Bigg( C \cdot T^* \left(T(C)^{-1} \right)  - I \Bigg)^2\right] \le \eps
$$
\end{definition}

The following lemma essentially says that the elements of $\text{Fixed}(T,\eps)$ are approximate minimizers for $\capac(T)$. 

\begin{lemma}\label{approximate_fixedpt} Suppose $C \in \text{Fixed}(T,\eps)$.
If $\eps \le 1/(n+1)$, then $T$ is rank non-decreasing. Furthermore, 
$$
(1-\sqrt{n \eps})^n \cdot \frac{\Det(T(C))}{\Det(C)} \le \capac(T) \le \frac{\Det(T(C))}{\Det(C)}
$$
Here $n$ is the size of matrices on which $T$ acts. Similar statement also holds for the operator $X \rightarrow T(T^*(X)^{-1})^{-1}$. 
\end{lemma}

We will use a $C \in \text{Fixed}(T,\eps)$ to find a scaling of $T$ which is almost doubly stochastic and then apply Theorem \ref{DS_full} (saying that almost doubly stochastic operators are rank non-decreasing) and Lemma \ref{cap_almostDS:quantum} (giving a lower bound on capacity of almost doubly stochastic operators).

\begin{proof} Consider the operator 
$$
Z(X) = T(C)^{-1/2} \cdot T \left( C^{1/2} \cdot X \cdot C^{1/2} \right) \cdot T(C)^{-1/2}
$$ 
Then $Z(I) = I$ and also 
$$
\tr \left[ \left( Z^*(I) - I\right)^2 \right] = \tr \left[ \left( C^{1/2} \cdot T^* \left( T(C)^{-1}\right) \cdot C^{1/2} - I\right)^2 \right] = \tr \left[ \Bigg( C \cdot T^* \left(T(C)^{-1} \right)  - I \Bigg)^2\right] \le \eps
$$
When $\eps \le 1/(n+1)$, by Theorem \ref{DS_full}, $Z$ is rank non-decreasing and hence $T$ is rank non-decreasing. Also by Lemma \ref{cap_almostDS:quantum}, 
$$
\capac(Z) = \capac(Z^*) \ge (1-\sqrt{n\eps})^n
$$
Then by Proposition \ref{cap_mul}
$$
\capac(T) = \capac(Z) \cdot \frac{\Det(T(C))}{\Det(C)} \ge (1-\sqrt{n \eps})^n \cdot \frac{\Det(T(C))}{\Det(C)} 
$$
The upper bound on $\capac(T)$ is clear from the definition of capacity. 
\end{proof}

Next we prove the other direction but again, not with matching parameters. 

\begin{lemma} Suppose $C \succ 0$ is s.t.
$$
 \capac(T) \ge \expon(-\delta) \cdot \frac{\Det(T(C))}{\Det(C)}
$$
where $\delta \le 1/6$. Then $C \in \text{Fixed}(T,6 \delta)$.
\end{lemma}

\begin{proof} As in the proof of Lemma \ref{approximate_fixedpt}, define the operator 
$$
Z(X) = T(C)^{-1/2} \cdot T \left( C^{1/2} \cdot X \cdot C^{1/2} \right) \cdot T(C)^{-1/2}
$$ 
Then $Z(I) = I$ and by Proposition \ref{cap_mul}
$$
\capac(Z^*) = \capac(Z) \ge 1-\delta
$$
Now by Lemma \ref{cap_almost1} (applied to $Z^*$), 
$$
\ds(Z) = \ds(Z^*) = \tr[(Z^*(I)-I)^2] \le 6 \delta
$$
This implies that  $C \in \text{Fixed}(T,6 \delta)$ because of the way operator $Z$ is set up.
\end{proof}

\subsection{Multiplicativity of capacity}\label{subsec:mult_cap}

\noindent In this subsection, we prove a curious multiplicativity property of capacity of completely positive operators. This is not required for the lower bound on capacity presented in this version but the easy direction of this theorem was required in the lower bound on capacity in the previous version. 

\begin{theorem}{\label{mult_cap_hard}}
Let $T_1$ and $T_2$ be completely positive operators where $T_1$ acts on matrices of dimension $d_1$ and $T_2$ acts on matrices of dimension $d_2$. Then 
$$
\capac(T_1 \tensor T_2)^{1/(d_1 d_2)} = \capac(T_1)^{1/d_1} \cdot \capac(T_2)^{1/d_2}
$$
\end{theorem}

\noindent We first prove the easy $\le$ direction of the above theorem. If $X$ and $Y$ are the minimizers for $\capac(T_1)$ and $\capac(T_2)$, then looking at how much $T_1 \tensor T_2$ shrinks the determinant of $X \tensor Y$ will give us the required statement. The infimum for $\capac(T_1)$ and $\capac(T_2)$ might not be achieved but that is fine.

\begin{lemma}{\label{mult_cap_easy}}
Let $T_1$ and $T_2$ be completely positive operators where $T_1$ acts on matrices of dimension $d_1$ and $T_2$ acts on matrices of dimension $d_2$. Then 
$$
\capac(T_1 \tensor T_2)^{1/(d_1 d_2)} \le \capac(T_1)^{1/d_1} \cdot \capac(T_2)^{1/d_2}
$$
\end{lemma}

\begin{proof}

\noindent Suppose $X$, $Y$ be $d_1$ and $d_2$ dimensional matrices respectively s.t. $X,Y \succ 0$ and \linebreak $\Det(X), \Det(Y) = 1$. Then $X \tensor Y \succ 0$ and $\Det(X \tensor Y) = 1$ as well. Also
$$
\Det \left( (T_1 \tensor T_2) (X \tensor Y)\right) = \Det \left( T_1(X) \tensor T_2(Y)\right) = \Det(T_1(X))^{d_2} \cdot \Det(T_2(Y))^{d_1}
$$
This proves that $\capac(T_1)^{d_2} \cdot \capac(T_2)^{d_1} \ge \capac(T_1 \tensor T_2)$ (by taking $X$ and $Y$ to be sequences of matrices which approach $\capac(T_1)$ and $\capac(T_2)$ respectively). Taking $1/(d_1d_2)$ powers on both sides completes the proof of the easy direction.

\end{proof}

\noindent We are now ready to prove Theorem \ref{mult_cap_hard}. The main ingredient is the duality between capacity being $0$ and the existence of a scaling to almost doubly stochastic position. 

\begin{proof}{(Of Theorem \ref{mult_cap_hard})} The easier $\le$ direction was proved above in Lemma \ref{mult_cap_easy}. Now we prove the harder $\ge$ direction. We can assume wlog that $\capac(T_1) > 0$ and $\capac(T_2) > 0$ (otherwise $\ge$ direction is trivial). Now by the analysis of Algorithm G (Theorem \ref{gurvitsalg}), we know that for any arbitrary $\eps > 0$, there exist invertible matrices $B, C \in M_{d_1}(\mathbb{C})$ and $D, E \in M_{d_2}(\mathbb{C})$ s.t. the scaled operators $S_1$ and $S_2$, defined below, are almost doubly stochastic. 
\begin{align*}
S_1(X) = B \cdot T_1(C \cdot X \cdot C^{\dagger}) \cdot B^{\dagger} \\
S_2(Y) = D \cdot T_2(E \cdot Y \cdot E^{\dagger}) \cdot D^{\dagger} 
\end{align*}
More formally, $S_1^*(I)=I, S_2^*(I)=I$ and $\tr \left[ (S_1(I) - I)^2 \right] \le \eps, \tr \left[ (S_2(I) - I)^2 \right] \le \eps$. Now consider the operator $S = S_1 \otimes S_2$. 
Note that $S(I) = S_1(I) \otimes S_2(I)$ and $S^*(I) = S_1^*(I) \otimes S_2^*(I)$. Hence $S^*(I) = I$ and $\tr \left[ (S(I) - I)^2 \right] \le \eps' = 2 (d_1 + d_2 + \eps) \eps$. This can be seen from the following chain of inequalities (and equalities):
\begin{align*}
\tr \left[ (S(I) - I)^2 \right] &= \tr \left[ (S_1(I) \otimes S_2(I) - I)^2 \right] \\
&= \tr \left[ (S_1(I) \otimes S_2(I) - S_1(I) \otimes I + S_1(I) \otimes I - I)^2 \right] \\
&\le 2 \tr \left[ \left(S_1(I) \otimes (S_2(I)-I) \right)^2 \right] + 2 \tr \left[ \left((S_1(I)-I) \otimes I\right)^2 \right] \\
&\le 2 (d_1+\eps) \eps + 2 \eps d_2 \\
&= 2 (d_1 + d_2 + \eps) \cdot \eps
\end{align*}
The first inequality follows from the following Cauchy-Schwarz inequality for psd matrices: 
$$\tr\left[ (M_1 + M_2)^2\right] \le 2 \tr[M_1^2] + 2 \tr[M_2^2]
$$
The second inequality follows from $\tr[S_1(I)^2] \le d_1 + \eps$, which follows from the fact that $\tr[S_1(I)] = \tr[S_1^*(I)] = \tr\left[I_{d_1} \right] = d_1$ and that $\tr \left[ (S_1(I) - I)^2 \right] \le \eps$, as well as $\tr \left[ (S_2(I) - I)^2 \right] \le \eps$. 
\\
Note that the operator $S$ is explicitly given as follows:
$$
S(Z) = \bigg(B \tensor D \bigg) \bigg(T_1 \tensor T_2 \bigg) \left( \bigg(C \tensor E \bigg) Z \left(C^{\dagger} \tensor E^{\dagger} \right) \right) \left(B^{\dagger} \tensor D^{\dagger} \right)
$$
Hence by Proposition \ref{cap_mul}, 
\begin{align*}
\capac(S) &= |\Det(B \tensor D)|^2 \cdot |\Det(C \tensor E)|^2 \cdot \capac(T_1 \tensor T_2) \\
&= |\Det(B)|^{2 d_2} \cdot |\Det(D)|^{2 d_1} \cdot |\Det(C)|^{2d_2} \cdot |\Det(E)|^{2 d_1} \cdot \capac(T_1 \tensor T_2)
\end{align*}
By Lemma \ref{cap_almostDS:quantum}, we have that $\capac(S) \ge (1 - \sqrt{d_1 d_2 \eps'})^{d_1 d_2}$ and hence
\begin{align}
\capac(T_1 \tensor T_2) \ge |\Det(B)|^{-2 d_2} \cdot |\Det(D)|^{-2 d_1} \cdot |\Det(C)|^{-2d_2} \cdot |\Det(E)|^{-2 d_1} \cdot (1 - \sqrt{d_1 d_2 \eps'})^{d_1 d_2} \label{eqn1:multiplicative}
\end{align}
Combining the fact that $S_1^*(I)=I$ and $S_2^*(I)=I$ along with Proposition \ref{cap_bound_1}, we get that $\capac(S_1), \capac(S_2) \le 1$. Also by Proposition \ref{cap_mul},we have that
\begin{align}
&\capac(T_1) = |\Det(B)|^{-2} \cdot |\Det(C)|^{-2} \cdot \capac(S_1) \le |\Det(B)|^{-2} \cdot |\Det(C)|^{-2} \label{eqn2:multiplicative} \\
&\capac(T_2) = |\Det(D)|^{-2} \cdot |\Det(E)|^{-2} \cdot \capac(S_2) \le |\Det(D)|^{-2} \cdot |\Det(E)|^{-2} \label{eqn3:multiplicative}
\end{align}
Combining equations (\ref{eqn1:multiplicative}), (\ref{eqn2:multiplicative}) and (\ref{eqn3:multiplicative}), we get that 
$$
\capac(T_1 \tensor T_2)^{1/d_1d_2} \ge \capac(T_1)^{1/d_1} \cdot \capac(T_2)^{1/d_2} \cdot (1 - \sqrt{d_1 d_2 \eps'})
$$
Since $\eps' = 2 (d_1 + d_2 + \eps) \eps$ can be taken to be arbitrarily close to $0$, this completes the proof.
\end{proof}

%% file: bit_and_continuity.tex
In this section, we will provide the bit complexity analysis of Algorithm $G$ and also provide explicit bounds on the continuity of capacity by a sensitivity analysis
of Algorithm $G$. We start by analyzing a natural iterative sequence associated with an operator $T$ that will be very useful, both, in the bit complexity and continuity analysis. Section \ref{sec:evolution} will describe how the operators $T_j$ in Algorithm $G$ evolve with respect this iterative sequence. Section \ref{bit_complexity} provides the bit complexity analysis of Algorithm $G$ and Section \ref{sec:capacity_continuity} provides explicit bounds for continuity of capacity. 

\noindent Consider the sequence of matrices given by $S_0 = T^*(I)$ ($I$ is of size $n \times n$), and
\begin{equation}\label{eq:scaling-matrices} 
   S_j = 
   \begin{cases} 
        T(S_{j-1}^{-1}) & \text{$j$ odd}, j \ge 1 \\
        T^*(S_{j-1}^{-1}) & \text{$j$ even}, j \ge 1.
   \end{cases}
\end{equation}

\noindent The next proposition studies the stability properties of this sequence of matrices. 

\begin{proposition}{\label{long_prop}} For a real symmetric positive definite matrix $X$, let $l(X)$ denote its smallest eigenvalue and $u(X)$ its largest.
\begin{enumerate}
\item Suppose $A$ is an $n \times n$ real symmetric positive definite matrix with integer entries s.t. the magnitude of any entry of $A$ is at most $M$. Then $l(A) \ge \frac{1}{(Mn)^{n-1}}$ and $u(A) \le Mn$. 
\item Define $\alpha = (M^2 n^2 m)^{n-1}$. Let $T$ be a completely positive operator whose Kraus operators $A_1,\ldots,A_m$ are $n \times n$ integer matrices and each entry of $A_i$ is of magnitude at most $M$. Also assume that $T(I), T^*(I)$ are both non-singular. Then
$$
(M^2 n^2 m) I \succeq T(I), T^*(I) \succeq \frac{1}{(M^2 n^2 m)^{n-1}} I
$$
\item Let $T, M, n, m$ be as before. Let $X$ be a real symmetric matrix. Then $||T(X)|| \le \alpha ||X||$. Also if X is real symmetric positive definite, then $l(T(X)) \ge \alpha^{-1} \cdot l(X)$ and $u(T(X)) \le \alpha \cdot u(X)$. Similarly $||T^*(X)|| \le \alpha ||X||$, $l(T^*(X)) \ge \alpha^{-1} \cdot l(X)$ and $u(T^*(X)) \le \alpha \cdot u(X)$.
\item $S_j$ are real symmetric positive definite matrices for all $0 \le j \le t$. Also $l(S_j) \ge \alpha^{-(j+1)}$ and $u(S_j) \le \alpha^{j+1}$ for all $0 \le j \le t$. 
\item For all $1 \le k, l \le n$, $|S_j(k,l)| \le \alpha^{j+1}$. 
\end{enumerate}
\end{proposition}

\begin{proof}
\begin{enumerate}
\item Let us first prove that $u(A) \le Mn$. 
\begin{align*}
u(A) &= \max_{\text{$v$ st. $||v||_2 = 1$}} ||Av||_2 \\
&= \max_{\text{$v$ st. $||v||_2 = 1$}} \sqrt{\sum_{k=1}^n \left( \sum_{l=1}^n A_{k,l} v_l \right)^2} \\
&\le \max_{\text{$v$ st. $||v||_2 = 1$}} \sqrt{\sum_{k=1}^n \left( \sum_{l=1}^n M |v_l| \right)^2} \\
&= \max_{\text{$v$ st. $||v||_2 = 1$}} \sqrt{n} \cdot M \cdot \left( \sum_{l=1}^n |v_l| \right) \\
&\le n M
\end{align*}
The last inequality follows from Cauchy-Schwarz. Now since $A$ is positive definite and the determinant is an integer, $\Det(A) \ge 1$. Then
\begin{align*}
1 \le \Det(A) \le u(A)^{n-1} l(A) \le (Mn)^{n-1} l(A)
\end{align*} 
from which it follows that $l(A) \ge \frac{1}{(Mn)^{n-1}}$.
\item Follows from previous point by noting that each entry of $T(I), T^*(I)$ is an integer of magnitude at most $M^2 n m$ and $T(I), T^*(I)$ are both symmetric matrices. 
\item We know that 
$$
\alpha I \succeq T(I), T^*(I) \succeq \alpha^{-1} I
$$
Suppose $X$ is a real symmetric matrix satisfying $||X|| = \beta$. Then 
$$
\beta I \succeq X \succeq -\beta I
$$
Then 
\begin{align*}
T(X + \beta I) \succeq 0
\end{align*}
since $T$ is completely positive and hence it maps psd matrices to psd matrices. Thus we get
$$
T(X) + \beta T(I) = T(X + \beta I) \succeq 0
$$ 
This follows from linearity of $T$. Then
$$
T(X) \succeq -\beta T(I) \succeq - \beta \alpha I
$$
Similarly one can prove that $\beta \alpha I \succeq T(X)$ which would imply $||T(X)|| \le \alpha \beta = \alpha ||X||$. Other statements can be proven in a similar manner.
\item We will prove this via induction on $j$. It is true for $S_0$ by point 2 in the proposition. Suppose the statement holds for $S_{j-1}$. Then $S_j = T(S_{j-1}^{-1})$ or  $S_j = T^*(S_{j-1}^{-1})$. It does not really matter which case it is for the purpose of this proof. Lets assume the former.
Then
\begin{align*}
l(S_j) &= l(T(S_{j-1}^{-1})) \\
&\ge \alpha^{-1} \cdot l(S_{j-1}^{-1}) \\
&= \alpha^{-1} \cdot u(S_{j-1})^{-1} \\
&\ge \alpha^{-1} \cdot \alpha^{-j} \\
&= \alpha^{-(j+1)}
\end{align*}
The first inequality is by point 3 in the proposition. The second inequality is by induction hypothesis. Similarly we can also prove that $u(S_j) \le \alpha^{j+1}$, which would complete the induction step. 
\item 
\begin{align*}
|S_j(k,l)| &= |e_k^T S_j e_l| \\
&\le ||e_k||_2 ||S_j e_l||_2 \\
&\le \alpha^{j+1}
\end{align*}
Here $e_k$ and $e_l$ are the standard basis vectors. 
\end{enumerate}
\end{proof}

\noindent Define another sequence of matrices defined by $U_0 = S_0$ and 
\[ U_j = \begin{cases} 
      T(U_{j-1}^{-1}) + \Delta_j& \text{$j$ odd}, j \ge 1 \\
      T^*(U_{j-1}^{-1}) + \Delta_j& \text{$j$ even}, j \ge 1 \\
   \end{cases}
\]
where $\Delta_j$'s are small perturbations. We now study the closeness of the sequences $\{U_j\}$ and $\{S_j\}$.

\begin{lemma}{\label{induct_u}}
If $||\Delta_j|| \le \frac{1}{2^j \cdot \alpha^{j+1}}$ for all 
$j \ge 1$, then $l(U_j) \ge \frac{1}{2^j \cdot \alpha^{j+1}}$ and $u(U_j) \le 2^j \cdot \alpha^{j+1}$, 
for all $j \ge 0$. Here $\alpha = (M^2 n^2 m)^{n-1}$. 
\end{lemma}

\begin{proof}
We will prove this using induction on $j$. It is true for $j=0$ since $U_0 = S_0$. Assume that $l(U_j) \ge \frac{1}{2^j \cdot \alpha^{j+1}}$ and $u(U_j) \le 2^j \cdot \alpha^{j+1}$. Then
$U_{j+1} = T(U_j^{-1}) + \Delta_{j+1}$ or $U_{j+1} = T^*(U_j^{-1}) + \Delta_{j+1}$. Suppose it is $T(U_j^{-1}) + \Delta_{j+1}$. The other case is similar.
\begin{align*}
l(U_{j+1}) &\ge l(T(U_j^{-1})) - ||\Delta_{j+1}|| \\
&\ge \alpha^{-1} \cdot l(U^{-1}_j) -  ||\Delta_{j+1}|| \\
&=  \alpha^{-1} \cdot u(U_j)^{-1} - ||\Delta_{j+1}|| \\
&\ge \alpha^{-1} \cdot \frac{1}{2^j \alpha^{j+1}} - \frac{1}{2^{j+1} \alpha^{j+2}} \\
&= \frac{1}{2^{j+1} \alpha^{j+2}}
\end{align*}
Also
\begin{align*}
u(U_{j+1}) &\le u(T(U_j^{-1})) + ||\Delta_{j+1}|| \\
&\le \alpha \cdot u(U_j^{-1}) + ||\Delta_{j+1}|| \\
&= \alpha \cdot l(U_j)^{-1} + ||\Delta_{j+1}|| \\
&\le \alpha \cdot 2^j \alpha^{j+1} + \frac{1}{2^{j+1} \alpha^{j+2}} \\
&\le 2^{j+1} \cdot \alpha^{j+2}
\end{align*}
Note that there is a lot of slack in this part, but it does not matter for our purposes. 
\end{proof}

\begin{lemma}{\label{spectral_diff}} The following holds for every integer $t$: suppose $||\Delta_j|| \le \delta$ for all $1 \le j \le t$, where $\delta \le \frac{1}{(2\alpha)^{t+1}}$. Then 
$$
||S_j - U_j|| \le (2\alpha)^{(2t+1)\cdot(t+1)} \delta
$$ 
for all $0 \le j \le t$. Here $\alpha = (M^2 n^2 m)^{n-1}$.
\end{lemma}

\begin{proof}
\begin{align*}
||S_j - U_j|| &= ||T(S_{j-1}^{-1}) - T(U_{j-1}^{-1}) - \Delta_j|| \\
&\le ||T(S_{j-1}^{-1} - U_{j-1}^{-1})|| + \delta \\
&\le \alpha \cdot ||S_{j-1}^{-1} - U_{j-1}^{-1}|| + \delta \\
&= \alpha \cdot ||S_{j-1}^{-1} (U_{j-1} - S_{j-1}) U_{j-1}^{-1}|| + \delta \\
&\le \alpha \cdot ||S_{j-1}^{-1}|| \cdot  ||S_{j-1} - U_{j-1}|| \cdot ||U_{j-1}^{-1}|| + \delta \\
&\le 2^{j-1} \alpha^{2j+1} \cdot ||S_{j-1} - U_{j-1}|| + \delta \\
&\le (2\alpha)^{2t+1} \cdot ||S_{j-1} - U_{j-1}|| + \delta
\end{align*}
For the fourth inequality we used Proposition \ref{long_prop} and Lemma \ref{induct_u}. Note that value of $\delta$ is small enough so that conditions of Lemma \ref{induct_u} are satisfied. 
\end{proof}

\subsection{Evolution of the Operator through Scaling}\label{sec:evolution}

In this subsection we compute succinct expressions for the operators $T_j$ which appear in Algorithm $G$,
together with expressions for their capacity and distance to doubly stochastic. These expressions will 
involve the matrices $S_j$ defined in the previous subsection. These operators $T_j$ are 
the intermediary operators in the scaling procedure, and the succinct expressions will be important for us
when approximating capacity, or proving the stability of the capacity of a quantum operator. 

Starting from a completely positive operator $T$ which satisfies $T(I)=I$, denote the sequence of 
operators obtained after right and left normalizations by $\{T_j\}$ i.e. $T_j$ is the operator obtained after 
$j$ iterations of right and left normalizations (as in Algorithm $G$). Note that
\[ T_j(X) = \begin{cases} 
      T_{j-1} \left(T_{j-1}^*(I)^{-1/2} \cdot X \cdot T_{j-1}^*(I)^{-1/2} \right) &\text{$j$ odd}, j \ge 1 \\
      T_{j-1}(I)^{-1/2} \cdot T_{j-1}(X) \cdot T_{j-1}(I)^{-1/2} & \text{$j$ even}, j \ge 1 \\
   \end{cases}
\]

\noindent For each $j$, $T_j$ is an operator scaling of $T$. So 
\begin{equation}\label{eq:scaled-operator}
T_j(X) = C_j \cdot T \left( D_j^{\dagger} \cdot X \cdot D_j\right) \cdot C_j^{\dagger}
\end{equation}
for some non-singular matrices $C_j$ and $D_j$. Let us denote $C_j^{\dagger} C_j$ by $P_j$ and $D_j^{\dagger} D_j$ by $Q_j$. It turns out $P_j$ and $Q_j$ are the matrices that really matter for the purpose of analyzing Algorithm $G$ and we will see next how these evolve. 

When $j$ is odd, $T_j^*(I)=I$. Note that 
\begin{equation}\label{eq:dual-scaled-operator}
T_j^*(X) = D_j \cdot T^*(C_j^{\dagger} \cdot X \cdot C_j) \cdot D_j^{\dagger}
\end{equation}
So the condition $T_j^*(I)=I$ gives us that $Q_j = T^*(P_j)^{-1}$. Also note that when $j$ is odd, $P_j = P_{j-1}$. When $j$ is even, we have, $T_j(I)=I$, which implies $P_j = T(Q_j)^{-1}$ and also $Q_j = Q_{j-1}$ holds. This can be summarized as follows:

\[ (P_j, \ Q_j) = 
\begin{cases} 
      (P_{j-1}, \  T^*(P_{j-1})^{-1}) &\text{$j$ odd}, j \ge 1 \\
      (T(Q_{j-1})^{-1}, \ Q_{j-1}) & \text{$j$ even}, j \ge 1 \\
   \end{cases}
\]
along with $T_0 = T$ and $P_0, Q_0 = I$. Thus we can see that
\begin{equation}\label{eq:scaling-factors} 
(P_j, \ Q_j) = 
\begin{cases} 
      (S_{j-2}^{-1}, \ S_{j-1}^{-1}) &\text{$j$ odd}, j \ge 1 \\
      (S_{j-1}^{-1}, \ S_{j-2}^{-1}) & \text{$j$ even}, j \ge 1 \\
   \end{cases}
\end{equation}
with the convention $S_{-1}=I$. For Algorithm $G$, what matters is 
\[ \eps_j = 
\begin{cases} 
      \tr \left[ \left(T_j(I)-I \right)^2 \right] &\text{$j$ odd}, j \ge 1 \\ \\
      \tr \left[ \left(T^*_j(I)-I \right)^2 \right] & \text{$j$ even}, j \ge 1 
   \end{cases}
\]
Since 
$$
T_j(I) = C_j \cdot T \left( D_j^{\dagger} D_j\right) \cdot C_j^{\dagger} = \left( C_j^{\dagger} \right)^{-1} \cdot P_j \cdot T(Q_j) \cdot C_j^{\dagger}
$$
and 
$$
T_j^*(I) = D_j \cdot T^*(C_j^{\dagger}C_j) \cdot D_j^{\dagger} = \left( D_j^{\dagger} \right)^{-1} \cdot Q_j \cdot T^*(P_j) \cdot D_j^{\dagger}
$$
we get (because similar matrices have the same trace) that 
\[ \eps_j = 
\begin{cases} 
      \tr \left[ \left(P_j \cdot T(Q_j)-I \right)^2 \right] &\text{$j$ odd}, j \ge 1 \\ \\
      \tr \left[ \left(Q_j \cdot T^*(P_j)-I \right)^2 \right] & \text{$j$ even}, j \ge 1 
   \end{cases}
\]
and hence 
\begin{subnumcases}{\eps_j = }
   \tr \left[ \left(S_{j-2}^{-1} \cdot T(S_{j-1}^{-1})-I \right)^2 \right] = \tr \left[ \left(S_{j-2}^{-1} \cdot S_j-I \right)^2 \right] & $\text{$j$ odd}, j \ge 1$ \label{eqn:epsj_1}
   \\
    \tr \left[ \left(S_{j-2}^{-1} \cdot T^*(S_{j-1}^{-1})-I \right)^2 \right] = \tr \left[ \left(S_{j-2}^{-1} \cdot S_j-I \right)^2 \right]  & $ \text{$j$ even}, j \ge 1 $ \label{eqn:epsj_2}
\end{subnumcases}

\noindent For the computation of capacity, the determinants that matter are
\[ d_j = 
\begin{cases} 
     \Det(T_{j-1}^*(I))^{-1} &\text{$j$ odd}, j \ge 1 \\ 
      \Det(T_{j-1}(I))^{-1} & \text{$j$ even}, j \ge 1 
   \end{cases}
\]
Then by Proposition \ref{cap_mul}, we have that
$$
\capac(T_r) = \capac(T_0) \cdot \prod_{j=1}^r d_j = \capac(T) \cdot \prod_{j=1}^r d_j 
$$
By the discussion above and the fact that similar matrices have the same determinants, we get that
\[ d_j = 
\begin{cases} 
      \Det \left(Q_{j-1}^{-1} \right) \cdot \Det \left(T^*(P_{j-1}) \right)^{-1}  = \Det(S_{j-3}) \cdot \Det(S_{j-1})^{-1} &\text{$j$ odd}, j \ge 1 \\ \\
     \Det \left(P_{j-1}^{-1} \right) \cdot \Det \left(T(Q_{j-1}) \right)^{-1} = \Det(S_{j-3}) \cdot \Det(S_{j-1})^{-1}  & \text{$j$ even}, j \ge 1 
   \end{cases}
\]
with the convention $S_{-2}, S_{-1} = I$. Thus we get
\begin{align}\label{eq:capacity-product}
\nonumber \capac(T_r) &= \capac(T) \cdot \prod_{j=1}^r d_j \\
\nonumber &=  \capac(T) \cdot \prod_{j=1}^r  \Det(S_{j-3}) \cdot \Det(S_{j-1})^{-1} \\
\nonumber &= \capac(T) \cdot \Det(S_{-2}) \cdot \Det(S_{-1}) \cdot \Det(S_{r-2})^{-1} \cdot \Det(S_{r-1})^{-1} \\
&= \capac(T) \cdot \Det(S_{r-2})^{-1} \cdot \Det(S_{r-1})^{-1} 
\end{align}

\subsection{Bit Complexity Analysis of Algorithm $G$}\label{bit_complexity}

We are now ready to analyze the bit complexity of Algorithm $G$. We will prove that if one runs Algorithm $G$ while truncating the numbers to an appropriate polynomial number of bits there is essentially no change in the convergence and required number of iterations. We will call the algorithm working with truncated inputs Algorithm $G'$.

Given a matrix $A$, let $\text{Trn}(A)$ be the matrix obtained by truncating the entries of A up to $P(n,\log(M))$ bits after the decimal point. $P(n,\log(M))$ is a polynomial which we will specify later. Note that $||A - \text{Trn}(A)||_{\infty} \le 2^{-P(n, \log(M))}$. 
We now describe Algorithm $G'$, which is the variant of Algorithm $G$ with truncation.

\begin{Algorithm}
\textbf{Input}: Completely positive operator $T$ given in terms of Kraus operators $A_1, \ldots, A_m \in \mathbb{Z}^{n \times n}$. Each entry of $A_i$ has absolute value at most $M$. \\
\textbf{Output}: Is $T$ rank non-decreasing?  

\begin{enumerate}
\item Check if $T(I)$ and $T^*(I)$ are singular. If any one of them are singular, then output that the operator is rank decreasing, otherwise proceed to step 2. 
\item Let $U_{-1} = S_{-1} = I$ and $U_0 = S_0 = T^*(I)$. \\
For($j = 1$ to $t$): 
$$ U_j = 
\begin{cases} 
	\text{Trn}(T(U_{j-1}^{-1})), \text{ if $j$ odd and } \\ 
	\text{Trn}(T^*(U_{j-1}^{-1})), \text{ if $j$ even.}
	\end{cases}
$$ 
Let $\tilde{\eps}_j = \tr \left[ \left(U_{j-2}^{-1} \cdot U_j - I \right)^2\right]$.
\item Check if $\min \{\tilde{\eps}_j : 1 \le j \le t\} \le 1/6n$. If yes, then output that the operator is rank non-decreasing otherwise output rank decreasing.
\end{enumerate}
\caption{Algorithm $G'$ (Algorithm $G$ with truncation)}
\label{Gurvits_alg_trn}
\end{Algorithm}

The parameter $t$ will be chosen as before: $t = 2 + 36n \cdot \log(1/f(n, M))$. We now show that throughout the iterations, distances to double stochasticity is essentially the same for the original and truncated algorithms $G$ and $G'$.

\begin{lemma}{\label{trunc}}
For an appropriate choice of $P(n,\log(M))$, $|\eps_j - \tilde{\eps}_j| \le 1/12n$, for $1 \le j \le t = 2 + 36n \cdot \log(1/f(n, M))$. 
\end{lemma}

\noindent Here $\eps_j = \ds(T_j)$ as defined in Algorithm $G$. Also note that by equations (\ref{eqn:epsj_1}) and (\ref{eqn:epsj_2}), 
$$
\eps_j = \tr \left[ \left(S_{j-2}^{-1} \cdot S_j - I \right)^2\right] 
$$

Let us first prove the correctness of Algorithm $G'$ assuming Lemma \ref{trunc}. As before, we can 
assume that $T(I)$ and $T^*(I)$ are both non-singular. If $\min \{\tilde{\eps}_j : 1 \le j \le t\} \le 1/6n$, 
then by Lemma \ref{trunc}, 
$$
\min \{\eps_j : 1 \le j \le t\} \le 1/6n + 1/12n \le 1/3n
$$
Hence by Theorem \ref{DS_full}, $T$ is rank non-decreasing.
Now assume, for the reverse direction, that $T$ is rank non-decreasing. Then by the analysis in 
Section~\ref{AlgG}, $\min \{\eps_j : 1 \le j \le t\} \le 1/12n$. Hence by Lemma~\ref{trunc}
$$
\min \{\tilde{\eps}_j : 1 \le j \le t\} \le 1/12n + 1/12n = 1/6n
$$
This proves the correctness of Algorithm $G'$. 

\begin{proof}(Of Lemma \ref{trunc}) $\eps_j = \tr \left[ \left(S_{j-2}^{-1} \cdot S_j - I \right)^2\right]$ 
and $\tilde{\eps}_j = \tr \left[ \left(U_{j-2}^{-1} \cdot U_j - I \right)^2\right]$. 
Let $M_j = S_{j-2}^{-1} \cdot S_j $ and $N_j = U_{j-2}^{-1} \cdot U_j$. Then
$$
||M_j||_{\infty} = ||S_{j-2}^{-1} \cdot S_j||_{\infty} \le ||S_{j-2}^{-1} \cdot S_j|| \le ||S_{j-2}^{-1}|| \cdot ||S_j|| \le \alpha^{2j}
$$
where the last inequality follows from Proposition \ref{long_prop}. Also
$$
||N_j||_{\infty} = ||U_{j-2}^{-1} \cdot U_j||_{\infty} \le ||U_{j-2}^{-1} \cdot U_j|| \le ||U_{j-2}^{-1}|| \cdot ||U_j|| \le 2^{2j-2} \cdot \alpha^{2j}
$$
where the last inequality follows from Lemma \ref{induct_u} and the fact that $P(n, \log(M))$ will be chosen to be a large enough polynomial so that the perturbations $\Delta_j$ in the lemma satisfy the condition $||\Delta_j|| \le 2^{-P(n,\log(M))} \le \frac{1}{2^j \cdot \alpha^{j+1}}$. 
Now
\begin{align*}
|\eps_j - \tilde{\eps_j}| &\le |\tr (M_j^2 - N_j^2)| + 2 |\tr(M_j - N_j)| \\
&= |\tr[(M_j+N_j)(M_j-N_j)]| + 2 |\tr(M_j - N_j)| \\
&\le (2 \alpha)^{2j} \cdot \sum_{k,l = 1}^n |(M_j - N_j)(k,l)| + 2 n ||M_j - N_j|| \\
&\le (2 \alpha)^{2j} \cdot n \cdot \sqrt{\sum_{k,l=1}^n |(M_j - N_j)(k,l)|^2} + 2 n ||M_j - N_j|| \\
&\le (2 \alpha)^{2j} \cdot n \cdot \sqrt{n} \cdot ||M_j - N_j|| + 2 n ||M_j-N_j|| \\
\end{align*}
In the first equality, we used the fact that $\tr(M_j N_j) = \tr(N_j M_j)$. In the second inequality, we used the fact that the maximum magnitude of any entry in $M_j+N_j$ is bounded by $(2 \alpha)^{2j}$ and that $|\tr(M_j - N_j)|$ is upper bounded by $n ||M_j-N_j||$. The third inequality is just Cauchy-Schwarz and the fourth is the fact that Frobenius norm of a matrix is upper bounded by $n ||M_j-N_j||^2$. Let us try to upper bound $||M_j-N_j||$ now.
\begin{align*}
||M_j-N_j|| &= ||S_{j-2}^{-1} \cdot S_j - U_{j-2}^{-1} \cdot U_j|| \\ &= ||S_{j-2}^{-1} \cdot S_j - S_{j-2}^{-1} \cdot U_j + S_{j-2}^{-1} \cdot U_j + U_{j-2}^{-1} \cdot U_j|| \\
&\le ||S_{j-2}^{-1} \cdot S_j - S_{j-2}^{-1} \cdot U_j|| + ||S_{j-2}^{-1} \cdot U_j + U_{j-2}^{-1} \cdot U_j|| \\
&\le ||S_{j-2}^{-1}|| \cdot ||S_{j} - U_j|| + ||S_{j-2}^{-1} - U_{j-2}^{-1}|| \cdot ||U_j|| \\
&=  ||S_{j-2}^{-1}|| \cdot ||S_{j} - U_j|| + ||S_{j-2}^{-1}\left(U_{j-2} - S_{j-2}\right)U_{j-2}^{-1}|| \cdot ||U_j|| \\
&\le ||S_{j-2}^{-1}|| \cdot ||S_{j} - U_j|| + ||S_{j-2}^{-1}|| \cdot ||U_{j-2} - S_{j-2}|| \cdot |U_{j-2}^{-1}|| \cdot ||U_j|| \\
&\le \alpha^{j-1} \cdot (2 \alpha)^{(2t+1) \cdot (t+1)} \delta + \alpha^{j-1} \cdot (2 \alpha)^{(2t+1) \cdot (t+1)} \delta \cdot 2^{j-2} \alpha^{j-1} \cdot 2^j \alpha^{j+1} \\
&\le (2 \alpha)^{10 t^2} \cdot \delta
\end{align*} 
The second last inequality follows by application of Proposition \ref{long_prop} and Lemmas \ref{induct_u} and \ref{spectral_diff}. Also $\delta$ here is at most $n \cdot 2^{-P(n,\log(M))}$ since $||\Delta_j||_{\infty} \le 2^{-P(n,\log(M))}$. 

Now $\alpha = (M^2 n^2 m)^{n-1} = \expon(\Theta(n \log(n) \log(M)))$ (since $m \le n^2$). Hence
\begin{align*}
(2 \alpha)^{10 t^2}&= \expon(\Theta(n \log(n) \log(M) \cdot t^2)) \\
&= \expon\left(\Theta\left(n^3 \log(n) \log(M) \log^2\left(1/f(n,M)\right)\right)\right)
\end{align*}
Hence choosing 
$$
P(n, \log(M)) = n^4 \log(M) \log^2\left(1/f(n,M)\right)
$$ 
suffices to get 
$$
|\eps_j - \tilde{\eps_j}| \le 1/12n
$$
We proved in Theorem \ref{capacity_lb} that $\log(1/f(n,M))$ is $\poly(n,\log(M))$, so $P(n,\log(M))$ is also a polynomial in $n$ and $\log(M)$. 
\end{proof}

\subsection{Continuity of Capacity}\label{sec:capacity_continuity}
In this section, we prove the continuity of capacity. We prove the following theorem:

\begin{theorem}\label{capacity_continuity}
Suppose $A_1,\ldots,A_m$ and $B_1,\ldots,B_m$ are two tuples of $n \times n$ matrices such that the bit-complexity of elements of $A_i$'s is $b$ and $||A_i - B_i|| \le \delta$ for all $i$. Let $T_A$ be the operator defined by $A_1,\ldots,A_m$ and $T_B$ be the operator defined by  $B_1,\ldots,B_m$. Then there exists a polynomial $P(n, b, \log(m))$ s.t. if $\delta \le \expon(-P(n,b, \log(m)))$, then $\capac(T_A) > 0$ implies $\capac(T_B) > 0$. Furthermore, 
$$
\left(1 - \frac{P(n,b, \log(m))}{\log(1/\delta)^{1/3}} \right) \le \frac{\capac(T_A)}{\capac(T_B)} \le \left( 1 + \frac{P(n,b, \log(m))}{\log(1/\delta)^{1/3}} \right)
$$
\end{theorem}

The fact that capacity is continuous is already mentioned in \cite{gurvits2004} and can be proved by other methods. But here we provide explicit bounds on the continuity parameters. Recall that $\text{Fixed}(T,\eps)$ (defined in subsection \ref{subsec:charac_opt}) is the set of hermitian positive-definite matrices $C$ which are $\eps$-approximate fixed points of the operator $X \rightarrow T(T^*(X)^{-1})^{-1}$ i.e. satisfy the following condition:
$$
\tr \left[ \Bigg( C \cdot T^* \left(T(C)^{-1} \right)  - I \Bigg)^2\right] \le \eps
$$
The main insight is that by the analysis of Algorithm $G$, since $\capac(T_A) > 0$, there exists a $C \in \text{Fixed}(T_A,\eps)$ with low spectral norm. Since $T_B$ is close to $T_A$, $C \in \text{Fixed}(T_B,\eps')$ for $\eps'$ close to $\eps$ and then the proof is finished by applying Lemma \ref{approximate_fixedpt}. 

\begin{proof}{(Of Theorem \ref{capacity_continuity})}
The proof of Theorem \ref{gurvitsalg} can be modified to prove the following: there exists a polynomial $Q(n,b)$ s.t. for all $\eps > 0$, if we run $t = Q(n,b)/\eps$ iterations of Algorithm $G$ starting from $T_0 = T_A$ satisfying $\capac(T_A) > 0$, then for some $1 \le j \le t$, $\eps_j = \ds(T_j) \le \eps$. Essentially, at each step capacity increases by roughly $\expon(\Omega(\eps))$ if $\eps_j > \eps$, capacity is lower bounded by $\expon(-Q(n,b))$ initially and upper bounded by $1$ always. 
\\
\\
\noindent By equations (\ref{eqn:epsj_1}) and (\ref{eqn:epsj_2}), we know that 
$$
\eps_j = \tr \left[ \left( S_{j-2}^{-1} S_j - I\right)^2\right]
$$
where $\{S_i\}$ is the sequence of matrices given by $S_0 = T_A^*(I)$, and
\[ S_i = \begin{cases} 
      T_A(S_{i-1}^{-1}) & \text{$i$ odd}, i \ge 1 \\
      T_A^*(S_{i-1}^{-1}) & \text{$i$ even}, i \ge 1 \\
   \end{cases}
\]
Suppose $1 \le r \le t$ be such that $\eps_r \le \eps$. Wlog assume that $r$ is odd. Then $S_r = T_A \left( T_A^* \left(S_{r-2}^{-1} \right)^{-1}\right)$. $\eps_r \le \eps$ implies that $S_{r-2}^{-1}$ is an $\eps$-approximate fixed point of the operator $X \rightarrow T_A \left( T_A^*(X)^{-1} \right)^{-1}$. Hence by Lemma \ref{approximate_fixedpt}, 
\begin{align}
(1-\sqrt{n\eps})^n \cdot \frac{\Det \left(T_A \left( S_{r-2}^{-1}\right)\right)}{\Det \left( S_{r-2}^{-1} \right)} \le \capac(T_A) \le \frac{\Det \left(T_A \left( S_{r-2}^{-1}\right)\right)}{\Det \left( S_{r-2}^{-1} \right)} \label{eqn:cap_TA}
\end{align}
We will prove now that $S_{r-2}^{-1}$ is also an $\eps'$-approximate fixed point of the operator $X \rightarrow T_B \left( T_B^*(X)^{-1} \right)^{-1}$ for an appropriate choice of $\eps'$. Let us denote $S_{r-2}^{-1}$ by $C$. By an application of Proposition \ref{long_prop}, it follows that the lowest and highest eigenvalues of $C$, $l(C)$ and $u(C)$ satisfy
\begin{align}
\frac{1}{l(C)}, u(C) \le \expon(O(n \cdot b \cdot r \cdot \log(nm))) \le \expon(Q_1(n,b, \log(m))/\eps) \label{eqn:newankit_1}
\end{align}
where $Q_1(n,b, \log(m))$ is another polynomial s.t. $Q_1(n,b, \log(m)) = O(Q(n,b) \cdot n \cdot b \cdot \log(nm))$. Let $D$ be an arbitrary matrix. Then
\begin{align}
||T_A(D) - T_B(D)|| &= ||\sum_{i=1}^m A_i D A_i^{\dagger} - \sum_{i=1}^m B_i D B_i^{\dagger}|| \nonumber \\
&\le \sum_{i=1}^m ||A_i D A_i^{\dagger} - B_i D B_i^{\dagger}|| \nonumber\\
&\le \sum_{i=1}^m \left( ||A_i D A_i^{\dagger} - A_i D B_i^{\dagger}|| + ||A_i D B_i^{\dagger} - B_i D B_i^{\dagger}||\right) \nonumber\\
&\le \sum_{i=1}^m \left( ||A_i|| \cdot ||D|| \cdot ||A_i^{\dagger} - B_i^{\dagger}|| + ||A_i - B_i|| \cdot ||D|| \cdot ||B_i^{\dagger}||\right) \nonumber\\
&\le 2m \cdot \left( n \cdot \expon(b)  + \delta \right)\cdot ||D|| \cdot \delta  \nonumber \\
&= \expon(Q_2(n,b, \log(m)) \cdot ||D|| \cdot \delta \label{eqn:closeTA_TB_1}
\end{align}
The first two inequalities are just the triangle inequality. The third inequality follows from submultiplicativity of the spectral norm. The fourth inequality follows from the fact that 
$$
||A_i|| \le n \cdot ||A_i||_{\infty} \le n \cdot \expon(b)
$$ and that 
$$
||B_i^{\dagger}|| = ||B_i|| \le ||A_i|| + ||A_i - B_i|| \le ||A_i|| +\delta
$$
Similarly, we have that
\begin{align}
||T_A^*(D) - T_B^*(D)|| \le  \expon(Q_2(n,b, \log(m)) \cdot ||D|| \cdot \delta \label{eqn:closeTA_TB_2}
\end{align}
We will now upper bound $||T_A\left( T_A^*(C)^{-1}\right) - T_B\left( T_B^*(C)^{-1}\right)||$:
\begin{align}
&||T_A\left( T_A^*(C)^{-1}\right) - T_B\left( T_B^*(C)^{-1}\right)|| \nonumber \\ &\le ||T_A\left( T_A^*(C)^{-1}\right) - T_B\left( T_A^*(C)^{-1}\right)|| + ||T_B\left( T_A^*(C)^{-1}\right) - T_B\left( T_B^*(C)^{-1}\right)|| \nonumber \\
&\le ||T_A\left( T_A^*(C)^{-1}\right) - T_B\left( T_A^*(C)^{-1}\right)|| + ||T_A\left( T_A^*(C)^{-1} - T_B^*(C)^{-1}\right) - T_B\left( T_A^*(C)^{-1} - T_B^*(C)^{-1}\right)|| \nonumber \\ 
&+ ||T_A\left( T_A^*(C)^{-1} - T_B^*(C)^{-1}\right)|| \label{eqn:newankit_7}
\end{align}
Equation (\ref{eqn:newankit_1}) along with Proposition \ref{long_prop} implies that the lowest and highest eigenvalues of $T_A^*(C)$ satisfy the following:
\begin{align}
\frac{1}{l\left(T_A^*(C)\right)}, u\left(T_A^*(C) \right) \le \expon(Q_4(n,b, \log(m))/\eps) \label{eqn:newankit_2}
\end{align}
Now applying equation (\ref{eqn:closeTA_TB_1}) with $D=T_A^*(C)^{-1}$ along with equation (\ref{eqn:newankit_2}) gives us the following:
\begin{align}
||T_A\left( T_A^*(C)^{-1}\right) - T_B\left( T_A^*(C)^{-1}\right)|| \le \expon(Q_5(n,b, \log(m))/\eps) \cdot \delta \label{eqn:newankit_3}
\end{align}
Now we will upper bound $||T_A^*(C)^{-1} - T_B^*(C)^{-1}||$. 
\begin{align}
||T_A^*(C)^{-1} - T_B^*(C)^{-1}|| &= ||T_A^*(C)^{-1} T_B^*(C)^{-1} \left( T_A^*(C) - T_B^*(C)\right)|| \nonumber \\
&\le ||T_A^*(C)^{-1}|| \cdot ||T_B^*(C)^{-1}|| \cdot ||T_A^*(C) - T_B^*(C)|| \nonumber \\
&\le \expon(Q_6(n,b, \log(m))/\eps) \cdot \delta \label{eqn:newankit_4}
\end{align}
The last inequality follows from equations (\ref{eqn:newankit_2}) and (\ref{eqn:closeTA_TB_2}). Note that we need $\delta \le \expon(-P(n,b, \log(m)))$ for a sufficiently large polynomial $P$ here to upper bound $||T_B^*(C)^{-1}||$ via equations  (\ref{eqn:newankit_2}) and (\ref{eqn:closeTA_TB_2}). Now applying equation (\ref{eqn:closeTA_TB_1}) with $D =  T_A^*(C)^{-1} - T_B^*(C)^{-1}$ along with equation (\ref{eqn:newankit_4}) gives us that
\begin{align}
||T_A\left( T_A^*(C)^{-1} - T_B^*(C)^{-1}\right) - T_B\left( T_A^*(C)^{-1} - T_B^*(C)^{-1}\right)|| \le  \expon(Q_7(n,b, \log(m))/\eps) \cdot \delta^2 \label{eqn:newankit_5}
\end{align}
We are left to upper bound $||T_A\left( T_A^*(C)^{-1} - T_B^*(C)^{-1}\right)||$. This follows from Proposition \ref{long_prop} and equation (\ref{eqn:newankit_4}).
\begin{align}
||T_A\left( T_A^*(C)^{-1} - T_B^*(C)^{-1}\right)|| \le  \expon(Q_8(n,b, \log(m))/\eps) \cdot \delta \label{eqn:newankit_6}
\end{align}
Combining equations (\ref{eqn:newankit_7}), (\ref{eqn:newankit_3}), (\ref{eqn:newankit_5}) and (\ref{eqn:newankit_6}), we get the following
\begin{align}
||T_A\left( T_A^*(C)^{-1}\right) - T_B\left( T_B^*(C)^{-1}\right)|| \le \expon(Q_9(n,b, \log(m))/\eps) \cdot \delta \label{eqn:newankit_8}
\end{align}
Let us denote the matrix $D_1 = C \cdot T_A\left( T_A^*(C)^{-1}\right)$ and $D_2 = C \cdot  T_B\left( T_B^*(C)^{-1}\right)$. $D_1$ determines whether $C$ is an approximate-fixed point of $T_A$ and $D_2$ determines whether $C$ is an approximate-fixed point of $T_B$. 
\begin{align}
||D_1 - D_2|| &\le ||C|| \cdot ||T_A\left( T_A^*(C)^{-1}\right) - T_B\left( T_B^*(C)^{-1}\right)|| \nonumber \\
&\le \expon(Q_{10}(n,b, \log(m))/\eps) \cdot \delta \label{eqn:newankit_9}
\end{align}
The second inequality follows from equations (\ref{eqn:newankit_1}) and (\ref{eqn:newankit_8}). We also have the following elementary inequality:
\begin{align}
\tr\left[ (D_1-I)^2\right] - \tr\left[ (D_2-I)^2\right] \le (||D_1|| + ||D_2|| + 2n) \cdot ||D_1 - D_2||
\end{align}
The above inequality implies that $C$ is an $\eps'$-approximate fixed point of $T_B$ for 
$$
\eps' = \eps + \expon(Q_{11}(n,b,\log(m))/\eps) \cdot \delta
$$
We can now choose $\eps = \frac{1}{2(n+1)}$. Then as long as $\delta \le \frac{\expon(-Q_{11}(n,b,\log(m))/\eps)}{2(n+1)}$, then $C$ is a $1/(n+1)$-approximate fixed point of the operator $X \rightarrow T_B \left( T_B^*(X)^{-1} \right)^{-1}$ and by Lemma \ref{approximate_fixedpt}, $T_B$ is rank non-decreasing and hence $\capac(T_B) > 0$. This proves first part of the theorem. 
\\
\\
\noindent For the second part of the theorem, since $C$ is an $\eps'$-approximate fixed point of $T_B$ for 
$$
\eps' = \eps + \expon(Q_{11}(n,b,\log(m))/\eps) \cdot \delta
$$
by Lemma \ref{approximate_fixedpt}, we have that 
\begin{align}
(1-\sqrt{n\eps'})^n \cdot \frac{\Det \left(T_B \left( C\right)\right)}{\Det \left( C\right)} \le \capac(T_B) \le \frac{\Det \left(T_B \left( C\right)\right)}{\Det \left( C \right)} \label{eqn:cap_TB} 
\end{align}
Also note that
\begin{align*}
||T_A(C) \cdot T_B(C)^{-1} - I|| &= ||\left(T_A(C) - T_B(C) \right) \cdot T_B(C)^{-1}|| \\
&\le ||\left(T_A(C) - T_B(C) \right)|| \cdot ||T_B(C)^{-1}|| \\
&\le \expon(Q_{12}(n,b,\log(m))/\eps) \cdot \delta
\end{align*}
by equation (\ref{eqn:closeTA_TB_1}) and Proposition \ref{long_prop}. Hence 
\begin{align}
\left( 1- \expon(Q_{12}(n,b,\log(m))/\eps) \cdot \delta\right)^n \le \frac{\Det(T_A(C))}{\Det(T_B(C))} \le \left( 1+ \expon(Q_{12}(n,b,\log(m))/\eps) \cdot \delta\right)^n \label{eqn:dets_close}
\end{align}
Now, combining equations (\ref{eqn:cap_TA}) and (\ref{eqn:cap_TA}), we get that 
$$
(1-\sqrt{n\eps})^n \cdot  \frac{\Det(T_A(C))}{\Det(T_B(C))} \le \frac{\capac(T_A)}{\capac(T_B)} \le \frac{1}{(1-\sqrt{n\eps'})^n} \cdot \frac{\Det(T_A(C))}{\Det(T_B(C))}
$$
Combining this with equation (\ref{eqn:dets_close}) gives us
$$
(1-\sqrt{n\eps})^n \cdot \left( 1- \expon(Q_{12}(n,b,\log(m))/\eps) \cdot \delta\right)^n \le \frac{\capac(T_A)}{\capac(T_B)} \le \frac{1}{(1-\sqrt{n\eps'})^n} \cdot \left( 1+ \expon(Q_{12}(n,b,\log(m))/\eps) \cdot \delta\right)^n
$$
Now choose $\eps = 2 \cdot \text{max}\left\{ Q_{11}(n,b,\log(m)),  Q_{12}(n,b,\log(m))\right\}/\log(1/\delta)$. This ensures that $\eps' \le 2 \eps$ and elementary calculations can then finish the proof of the theorem. 
\end{proof}

%% file: comp-capacity.tex
In this section, we show how algorithm $G$ can be used to compute an approximation to the 
capacity of any quantum operator. For simplicity of exposition, we will present in this section
an analysis of convergence of algorithm $G$ without truncation. Afterwards, in subsection \ref{sec:comp_cap_truncation}, we show how to
adapt the analysis of algorithm $G$ to handle the truncation. This corresponds to the analysis
of algorithm $G'$ in the previous section. 

We begin with the following lemma, which is an adaptation of Lemma 3.10 from~\cite{LSW}.

\begin{lemma}\label{lem:prod-upper-bd}
	Let $x_1, \ldots, x_n$ be positive real numbers such that $\dst\sum_{i=1}^n x_i = n$
	and $\dst\sum_{i=1}^n (x_i-1)^2 = \alpha$. Then
	$$ \prod_{i=1}^n x_i \le 
	\begin{cases} 
	\exp\left(- \alpha/6 \right), \ \text{ if } \alpha \leq 1,  \\ 
	\exp\left(\frac{-1}{6} \right), \ \text{ otherwise.}
	\end{cases}$$
\end{lemma}

\begin{proof}
	We have two cases to analyze:
	
	\textbf{Case 1:} $\alpha \leq 1$. 
	
	In this case, by using the inequality 
	$1 + t \leq \dst\exp\left(t - \frac{t^2}{2} + \frac{t^3}{3}\right)$, 
	which holds for all $t \in \R$, we have:
	\begin{align*}
		\prod_{i=1}^n x_i = \prod_{i=1}^n [1 + (x_i -1)] &\leq 
		\prod_{i=1}^n \exp\left((x_i-1) - \frac{(x_i-1)^2}{2} + \frac{(x_i-1)^3}{3}\right) \\
		&= \exp\left(\sum_{i=1}^n(x_i-1) - \frac{1}{2} \cdot \sum_{i=1}^n (x_i-1)^2 + 
		\frac{1}{3} \cdot \sum_{i=1}^n (x_i-1)^3 \right) \\
		&\leq \exp\left(- \frac{1}{2} \cdot \alpha + 
		\frac{1}{3} \cdot \alpha^{3/2} \right) \leq \exp\left(- \frac{\alpha}{6} \right)
	\end{align*}	 
	Where in the last inequalities we used the fact that 
	$\dst\sum_{i=1}^n (x_i-1)^3 \leq \left( \dst\sum_{i=1}^n (x_i-1)^2 \right)^{3/2}$
	and $\alpha^{3/2} \le \alpha$.
	
	\textbf{Case 2:} $\alpha > 1$. 		

	Consider the function $f(\lambda) = \prod_{i=1}^n (1+\lambda (x_i-1))$. We will prove that 
	$f$ is a decreasing function of $\lambda$ when $\lambda \in [0,1]$. In that case 
	\begin{align*}
	\prod_{i=1}^n x_i = \prod_{i=1}^n [1 + (x_i -1)] = f(1) \le f\left( \frac{1}{\sqrt{\alpha}} \right) 
	= \prod_{i=1}^n \left[1 + \frac{x_i -1}{\sqrt{\alpha}} \right] \le \exp\left(- \frac{1}{6} \right)
	\end{align*}
	where the last inequality follows from Case 1. Now let us prove that $f$ is decreasing.
	$$
	f'(\lambda) = f(\lambda) \cdot \left( \sum_{j=1}^n \frac{x_i-1}{1+\lambda (x_i-1)} \right) \le f(\lambda) \cdot \left( \sum_{j=1}^n (x_i-1)\right) = 0
	$$
	This completes the proof.

\end{proof}

As a corollary of Lemma~\ref{lem:prod-upper-bd}, we obtain the following 
quantitative progress measure towards computing capacity:

\begin{lemma}[Quantitative Progress]\label{lem:quant-progress}
	Let $T$ be a right (left) normalized quantum operator such that $\ds(T) = \alpha$.
	Additionally, let $\widetilde{T}$ be the left (right) normalization of operator $T$. Then,
	$$ \capac(\widetilde{T}) \geq
	\begin{cases} 
	\capac(T) \cdot \exp\left(\alpha/6 \right), 
	\ \text{ if } \alpha \leq 1,  \\ 
	\capac(T) \cdot \exp\left(\frac{1}{6} \right), \ \text{ otherwise.}
	\end{cases} $$
\end{lemma}

\begin{proof}
	Suppose $T$ is right normalized and $\widetilde{T}$ is the left normalization of $T$. 
	Proposition~\ref{cap_mul} tells us that
	$$ \capac(\widetilde{T}) = \det(T(I))^{-1} \cdot \capac(T). $$
	Let $\lambda_1, \ldots, \lambda_n$ be the eigenvalues of $T(I)$. 
	As $n = \tr(T(I)) = \dst\sum_{i=1}^n \lambda_i$ and 
	$\alpha = \ds(T) = \tr[(T(I) - I)^2] = \dst\sum_{i=1}^n (\lambda_i -1)^2$, the conditions
	of Lemma~\ref{lem:prod-upper-bd} apply and we have 
	$$
	\det(T(I)) = \prod_{i=1}^n \lambda_i \leq
	\begin{cases} 
	\exp\left(- \alpha/6 \right), \ \text{ if } \alpha \leq 1,  \\ 
	\exp\left(\frac{-1}{6} \right), \ \text{ otherwise.}
	\end{cases}
	$$
	This implies the desired lower bounds on $\capac(\widetilde{T})$. Since the case
	where $T$ is left normalized is analogous, we omit the argument.
\end{proof}

We now state a slight modification of Algorithm G, with a view towards computing 
the capacity of a quantum operator.

\begin{Algorithm}
\textbf{Input}: Quantum operator $T$ given in terms of Kraus operators 
$A_1, \ldots, A_m \in \mathbb{Z}^{n \times n}$ and approximation parameter $\epsilon > 0$. 
Each entry of $A_i$ has absolute value at most $M$. \\
\textbf{Output}: $\capac(T)$ with multiplicative error of $(1\pm\epsilon)$. 

\begin{enumerate}
\item Check if $T(I)$ and $T^*(I)$ are singular. If any one of them is singular, then output 
$\capac(T) = 0$, otherwise proceed to step 2. 
\item Alternately perform right and left normalizations on $T = T_0$ for $t$ steps. 
Let $T_j$ be the operator after $j$ steps. Also let $\eps_j = \ds(T_j)$. 
Go to step 3.
\item If $\eps_j  \le \dfrac{\eps^2}{n^3}$ for any $0 \leq j \leq t$, go to step 4. 
Otherwise, output $\capac(T) = 0$. 
\item For the smallest $0 \leq j \leq t$ such that $\eps_j  \le \dfrac{\eps^2}{n^3}$,
output 
$$ \capac(T) = \dst\prod_{i=0}^{j-1} \det(R_i) ,
\text{ where } R_i = \begin{cases} T_i(I), \text{ if $i$ is odd}, \\ T_i^*(I), \text{ otherwise.} \end{cases}$$
\end{enumerate}
\caption{Algorithm G, computing capacity}
\label{Gurvits_alg_capacity}
\end{Algorithm}

\begin{theorem}\label{thm:computing-capacity} 
Let $T$ be a completely positive operator, whose Kraus 
operators are given by $n \times n$ integer matrices $A_1, \ldots, A_m$, 
such that each entry of $A_i$ has absolute value at most $M$.  
Algorithm $G$ when applied for 
$t = \dfrac{n^3}{\eps^2} \cdot \left(1 + 10n^2\log(Mn) \right)$ 
steps approximates $\capac(T)$ within a multiplicative factor of $1 \pm \epsilon$. 
\end{theorem}

\begin{proof}
	If either $T(I)$ or $T^*(I)$ is singular, then $T$ decreases the rank of $I$. 
	When $T(I)$ is singular, rank decreasing follows by definition. When $T^*(I)$ is singular, one way 
	to see it is that $\text{Im}(A_i) \subseteq \text{Im} \left( T^*(I)\right)$ for all $i$.
 	Since $\text{Im}\left(A_i A_i^{\dagger}\right) = \text{Im}(A_i)$, we get that 
	$ \text{Im} \left( T(I)\right) \subseteq  \text{Im} \left( T^*(I)\right)$ 
	and hence $T(I)$ is singular. Therefore, the algorithm is 
	correct on step 1, by outputting $\capac(T) = 0$.
	
	If $T(I)$ and $T^*(I)$ are both non-singular, it is easy to verify that $T_j(I)$ and $T_j^*(I)$ 
	will remain non-singular for all $j$ and hence step 3 is well defined. 
	
	We now divide the proof into two cases:
	
	\textbf{Case 1:} $T$ is rank decreasing.
	
	In this case, since right and left normalizations don't change the 
	property of being rank decreasing, we have $\capac(T_j) = 0$ for all $0 \leq j \leq t$. 
	Hence, Lemma~\ref{cap_almostDS:quantum} implies that 
	$\ds(T_j) > \dfrac{\eps^2}{n^3}$ for all $0 \leq j \leq t$. In this case, step 3 of Algorithm G
	will correctly output $\capac(T) = 0$.
	
	\textbf{Case 2:} $T$ is rank non-decreasing.	
	
	In this case, we will show that there must exist $0 \leq j \leq t$ such that 
	$\eps_j \leq \dfrac{\eps^2}{n^3}$. Assume the contrary, for the sake of contradiction. 
	By Theorem~\ref{capacity_lb}, we know that $\capac(T_1) \geq \exp(-10n^2\log(Mn))$.
	Also Proposition~\ref{cap_bound_1} implies that $\capac(T_j) \le 1$ for all $j$. 
	However by the assumption that $\eps_j > \dfrac{\eps^2}{n^3}$, 
	Lemma~\ref{lem:quant-progress} implies that 
	$\capac(T_{j+1}) \geq \exp(\eps^2/n^3) \cdot \capac(T_j)$ for all $0 \leq j \leq t$.
	Hence, we obtain:
	$$
	1 \ge \capac(T_{t+1}) \ge \exp\left(\dfrac{t\eps^2}{n^3} \right) \cdot \capac(T_1) \geq 
	\exp\left(\dfrac{t\eps^2}{n^3} \right) \cdot \exp(-10n^2\log(Mn))
	$$
	Plugging in $t = \dfrac{n^3}{\eps^2} \cdot \left(1 + 10n^2\log(Mn) \right)$ gives us 
	the required contradiction.	
	
	Now that we showed the existence of $0 \leq j \leq t$ such that $\eps_j \leq \dfrac{\eps^2}{n^3}$, 
	we will show that step 4 indeed computes a good approximation to capacity. 
	For the first $\eps_j$ such that $\eps_j \leq \dfrac{\eps^2}{n^3}$, 
	Lemma~\ref{cap_almostDS:quantum} implies that 
	$\capac(T_j) \geq (1 - \sqrt{n\eps_j})^n \geq (1 - \eps/n)^n \geq 1 - \eps$.
	Since $\capac(T_j) = \capac(T) \cdot \left( \dst\prod_{i=0}^{j-1} \det(R_i) \right)^{-1}$,
	we have 
	$$ \capac(T) = \capac(T_j) \cdot \dst\prod_{i=0}^{j-1} \det(R_i).$$
	As $1- \eps \leq \capac(T_j) \leq 1$, we obtain the correct approximation.
\end{proof}

\subsection{Computing Capacity with Truncation}\label{sec:comp_cap_truncation}

In this subsection, we analyze the computation of capacity when we truncate the input matrices.
This analysis will be similar to the one in Section~\ref{bit_complexity}. We begin with some intuition
on why truncation works. 

Notice that to approximate the capacity, all we need is to compute the determinants of the matrices 
$U_i$ in Algorithm~\ref{Gurvits_alg_trn}. The $U_i$'s are the truncations of the matrices 
$S_i$ from equation~\eqref{eq:scaling-matrices}, the latter matrices being important as 
they describe the scaled operator $T_j$ in terms of the original operator $T$, see 
equations~(\ref{eq:scaled-operator},~\ref{eq:dual-scaled-operator},~\ref{eq:scaling-factors}). 
The reason why truncating 
the input works is because the eigenvalues of $U_i$ are very similar to the eigenvalues of $S_i$.
Therefore, we can show that $\det(S_i) \approx \det(U_i)$. This will imply that the truncated capacity
is a good approximation to the actual capacity. The analysis will rely mainly on Lemma~\ref{spectral_diff},
which gives a good bound on the spectral norm of $S_i - U_i$. Now we state the truncated algorithm.

Given a matrix $A$, let $\trn(A)$ be the matrix obtained by truncating the entries of $A$ up to 
$P(n, 1/\epsilon, \log(M)) = \frac{1}{\epsilon} \cdot (n^{12}\log^4(Mn)) \cdot \log(n^4/\epsilon^2)$ 
bits after the decimal point. 

\begin{Algorithm}
\textbf{Input}: Quantum operator $T$ given in terms of Kraus operators 
$A_1, \ldots, A_m \in \mathbb{Z}^{n \times n}$ and approximation parameter $\epsilon > 0$. 
Each entry of $A_i$ has absolute value at most $M$. \\
\textbf{Output}: $\capac(T)$ with multiplicative error of $(1\pm\epsilon)$. 

\begin{enumerate}
\item Check if $T(I)$ and $T^*(I)$ are singular. If any one of them is singular, then output 
$\capac(T) = 0$, otherwise proceed to step 2. 
\item Let $U_{-1} = S_{-1} = I_n$ and $U_0 = S_0 = T^*(I)$. Additionally, for $1 \le j \le t$, let
$$ U_j = 
	\begin{cases} 
		\trn(T(U_{j-1}^{-1})), \text{ if $i$ is odd}, \\ 
		\trn(T^*(U_{j-1}^{-1})), \text{ otherwise.} 
	\end{cases}
$$ 
Also let $\wt{\eps}_j = \tr\left[ (U_{j-2}^{-1} \cdot U_j - I_n)^2 \right]$. Go to step 3.
\item If $\wt{\eps}_j  \le \dfrac{\eps^2}{4n^3}$ for any $0 \leq j \leq t$, go to step 4. 
Otherwise, output $\capac(T) = 0$. 
\item For the first $0 \leq j \leq t$ such that $\wt{\eps}_j  \le \dfrac{\eps^2}{4n^3}$,
output 
$$ \capac(T) = \det(U_{j-1}) \cdot \det(U_{j-2}).$$
\end{enumerate}
\caption{Algorithm G with truncation, computing capacity}
\label{Gurvits_alg_capacity_trn}
\end{Algorithm}

We now proceed to the analysis of Algorithm~\ref{Gurvits_alg_capacity_trn}. 
In Theorem~\ref{thm:computing-capacity}, we proved the correctness
of Algorithm G without truncation. Thus, to prove correctness of 
Algorithm~\ref{Gurvits_alg_capacity_trn}, it is enough to prove two statements: 
\begin{enumerate}
	\item if $\wt{\epsilon}_j \leq \eps^2/4n^3$, then the 
	operators $T_j$ will satisfy the $\epsilon_j \leq \eps^2/n^3$ bounds
	\item $\|U_i - S_i \| \leq  2^{- P(n, 1/\eps, \log(M))/2}$.
\end{enumerate} 
The first item implies that steps 1 to 3 of the algorithm above are correct, and the second item will tell us
that step 4 indeed computes an $1 \pm \eps$ approximation to capacity. More formally, we have the
following theorem.

\begin{theorem}\label{thm:computing-capacity-truncated} 
Let $T$ be a completely positive operator, whose Kraus 
operators are given by $n \times n$ matrices $A_1, \ldots, A_m \in \mathbb{Z}^{n \times n}$, 
such that each entry of $A_i$ has absolute value at most $M$.  
Algorithm~\ref{Gurvits_alg_capacity_trn}, with truncation parameter 
$P(n, 1/\epsilon, \log(M)) = \frac{1}{\epsilon} \cdot (n^{12}\log^4(Mn)) \cdot \log(n^4/\epsilon^2)$ 
when applied for $t = \dfrac{4n^3}{\eps^2} \cdot \left(1 + 10n^2\log(Mn) \right)$ 
steps approximates $\capac(T)$ within a multiplicative factor of $1 \pm \epsilon$. 
\end{theorem}

\begin{proof}
	By applying Lemma~\ref{trunc} with parameters $t$ and $P(n, 1/\epsilon, \log(M))$ as above, we get 
	that 
	$$|\wt{\epsilon}_i - \epsilon_i| < \dfrac{\eps^2}{n^4} \text{ for all } 0 \leq i \leq t. $$
	Therefore, steps 1 to 3 of Algorithm~\ref{Gurvits_alg_capacity_trn} work just as if we had not done any
	truncation (as in Algorithm~\ref{Gurvits_alg_capacity}). This implies that we will always output 
	$\capac(T) = 0$ whenever the operator $T$ is rank decreasing.
	
	We are now left with the computation of capacity when $T$ is rank non-decreasing, which is done in
	step 4. By applying Lemma~\ref{spectral_diff} with parameter $\delta = 2^{- P(n, 1/\epsilon, \log(M))}$,
	we get
	$$ \|U_i - S_i \| \leq  2^{- P(n, 1/\eps, \log(M))/2}, \text{ for all } 0 \leq i \leq t. $$
	Let $0 \leq \mu_{i1} \leq \mu_{i2} \leq \dots \leq \mu_{in}$ be the eigenvalues of $U_i$ and  
	$0 \leq \lambda_{i1} \leq \lambda_{i2} \leq \dots \leq \lambda_{in}$ be the eigenvalues of $S_i$.
	From $\|U_i - S_i \| \leq  2^{- P(n, 1/\eps, \log(M))/2}$ and Lemma~\ref{induct_u}, we have
	$$ | \mu_{i \ell} - \lambda_{i \ell} | \leq 2^{- P(n, 1/\eps, \log(M))/2} \leq \frac{\eps \mu_{i \ell}}{4n}  \then
	\mu_{i \ell}\left(1 - \frac{\eps}{4n} \right) \leq \lambda_{i \ell} \leq \mu_{i \ell} \left(1 + \frac{\eps}{4n} \right).$$
	Hence, we have that
	\begin{align*}
		\det(S_{j-1}) \cdot \det(S_{j-2}) = \prod_{i=j-2}^{j-1} \prod_{\ell=1}^n \lambda_{i, \ell} &\geq 
		\prod_{i=j-2}^{j-1} \prod_{\ell=1}^n \mu_{i, \ell} \left(1- \frac{\epsilon}{4nt}\right) \\
		&\geq \left(1- \frac{\epsilon}{4n}\right)^{2n} \cdot \prod_{i=j-2}^{j-1} \prod_{\ell=1}^n \mu_{i, \ell} \\
		&\geq (1 - \eps/2) \cdot \prod_{i=j-2}^{j-1} \prod_{\ell=1}^n \mu_{i, \ell} 
		= (1- \eps/2) \cdot \det(U_{j-1}) \cdot \det(U_{j-2})
	\end{align*} 
	Similarly, we have that $\det(S_{j-1}) \cdot \det(S_{j-2}) \leq (1+\eps/2)\cdot \det(U_{j-1}) \cdot \det(U_{j-2})$.
	
	As $\wt{\eps}_j \leq \dfrac{\eps^2}{4n^3}$ implies that 
	$\eps_j \leq \dfrac{\eps^2}{4n^3} + \dfrac{\eps^2}{n^4}$, by Lemma~\ref{cap_almostDS:quantum} we
	have $\capac(T_j) \in [(1-\eps/2), 1]$. Thus, equation~\eqref{eq:capacity-product} yields
	\begin{align*}
		\capac(T) &= \capac(T_j) \det(S_{j-1}) \det(S_{j-2}) \\
		&\then (1-\eps/2) \det(S_{j-1}) \det(S_{j-2}) \leq \capac(T) \leq \det(S_{j-1}) \det(S_{j-2}) \\
		&\then (1-\eps/2)^2\det(U_{j-1}) \det(U_{j-2}) \leq 
		\capac(T) \leq (1+\eps/2) \det(U_{j-1}) \det(U_{j-2}).
	\end{align*}
	The inequalities above imply that $\det(U_{j-1}) \cdot \det(U_{j-2})$, that is, the output of 
	Algorithm~\ref{Gurvits_alg_capacity_trn}, lies in the interval $[(1-\eps) \capac(T), (1+\eps)\capac(T)]$.
\end{proof}

%% file: conclusion.tex
In this paper we gave a polynomial time algorithm for computing the non-commutative rank of a symbolic matrix over any subfield of 
the complex numbers. We stated its different incarnations and implications to the many different areas in which this problem arises 
(indeed we feel that expositing these many connections, some essential to the present result, may yield better future interaction 
between them with possible more benefits). We note that our algorithm and the analysis bypasses the need to use any degree bounds at all. We further note again that despite the purely algebraic nature of the problem our algorithm is purely analytic, generating a  
sequence of complex matrices and testing its convergence.

We collect now the most obvious directions for future research, some of them already mentioned in the paper.
\begin{itemize}
\item Find more applications of this algorithm to optimization problems.
\item Can we use these techniques to design an efficient deterministic algorithm for the orbit-closure intersection problem for the Left-Right action? In terms of invariants, this is equivalent to asking if two tuples of matrices can be separated by invariants of the Left-Right action (over algebraically closed fields of characteristic $0$). More formally given two tuples of matrices, $(A_1, \ldots, A_m)$ and $(B_1,\ldots,B_m)$, check whether for all $(T_1,\ldots, T_m)$ of arbitrary dimension, 
$$
\Det \left( T_1 \tensor A_1 + \cdots T_m \tensor A_m\right) = \Det \left( T_1 \tensor B_1 + \cdots T_m \tensor B_m\right)
$$
The results of \cite{derksen2015} give a randomized polynomial time algorithm for this problem (over algebraically closed fields of characteristic $0$): just plug in random $(T_1,\ldots,T_m)$ of polynomial dimension.
\item Find a black-box algorithm for SINGULAR. That is, efficiently produce (deterministically) a polynomial size set $\mathcal{S}$ of tuples of polynomial dimension matrices s.t. for all $L = \sum_{i=1}^m x_i A_i$ s.t. $L$ is non-singular, it holds that
for some $(T_1,\ldots,T_m) \in \mathcal{S}$, 
$$
\Det \left( T_1 \tensor A_1 + \cdots T_m \tensor A_m \right) \neq 0
$$
Due to the recent polynomial dimension bounds of $\cite{derksen2015}$, it can be proved that a random set $\mathcal{S}$ works. The challenge is to produce it deterministically. As a special case, this captures deterministic parallel algorithms for the decision version of bipartite perfect matching (when $A_1,\ldots,A_m$ are elementary matrices $E_{i,j}$ representing the edges of a bipartite graph). So perhaps, techniques from the recent breakthrough work \cite{FGT2016} can be useful. 
\item Explore further the connection between commutative and non-commutative PIT problems. We feel that beyond the many 
connections between commutative and non-commutative settings that arise here, this different angle of looking at the  commutative 
PIT problem, relating it to its non-commutative cousin, may help in the major quest of finding an efficient  deterministic algorithm for it. As mentioned above, this viewpoint has already resulted in a deterministic PTAS for the commutative rank \cite{BlaserJP16}.
\item We design an efficient algorithm for checking if a completely positive operator is rank-decreasing. Can we do the same for positive operators? Algorithm G in fact works for positive operators as well and all that is needed is to prove an effective lower bound on the capacity $\capac(T)$ of a positive operator $T$ which is rank non-decreasing (similar to Theorem \ref{capacity_lb}). It was already proven in $\cite{gurvits2004}$ that $\capac(T) > 0$ for a positive operator $T$ which is rank non-decreasing. 
\item Design a strongly polynomial time algorithm for operator scaling. Strongly polynomial time algorithms for matrix scaling were given by \cite{LSW}. Can they be extended to the operator case?
\item Can we compute $(1+\eps)$ approximation to $\capac(T)$ in time $\poly(n,b,\log(1/\eps))$? For computing capacity of non-negative matrices, such algorithms exist. One of the algorithms in \cite{LSW} has this stronger convergence rate. Also for matrices, capacity can be formulated as a convex program and hence the Ellipsoid algorithm also gives this stronger convergence rate \cite{GurYianilos}.
\item Can we design efficient algorithms for testing the null-cone of general quivers? There is reduction from general quivers to Kronecker-quiver or the left-right action (e.g. see \cite{derksen2015}) but the reduction is not always efficient. What about the general problem of testing the null-cone of actions of reductive groups?
\end{itemize}

%% file: comp-ncrank.tex
In this section, we show how to compute the non-commutative rank of any (not
necessarily square) matrix with linear entries over the free skew field $\Qxx$. 
This will be achieved in two ways: the first, in Subsection~\ref{classical-red}, by reducing this problem to testing 
singularity of a certain square matrix with linear entries, and the second, in Subsection~\ref{quantum}, by a purely quantum
approach which in a sense mimics the reduction from maximum matching to perfect matching.

In fact, we solve a more general problem. Subsection~\ref{Higman} starts with a reduction of computing $\ncrank$ of a matrix with 
{\em polynomial} entries (given by formulae), to the problem of computing the $\ncrank$ of a matrix with linear entries, 
via the so-called  ``Higman's trick'' (Proposition~\ref{prop:higman-trick}). We give the simple  quantitative analysis of this reduction, which as far as 
we know does not appear in the literature and may be useful elsewhere. This reduction, with the two above, allow computing the 
non-commutative rank of any matrix in time polynomial in the description of its entries.

\subsection{Higman's Trick}~\label{Higman}

Before stating the full version of the effective Higman trick, we need to define the bit complexity of a
formula computing a non-commutative polynomial.

\begin{definition}[Bit Complexity of a Formula]
	Let $\Phi$ be a non-commutative formula without divisions such that each of its gates
	computes a polynomial in $\QX$ (i.e., the inputs to the formula are either rational numbers 
	or non-commutative variables). The \emph{bit complexity} of $\Phi$ 
	is the maximum bit complexity of any rational input appearing in the formula $\Phi$.
\end{definition}

With this definition in hand, we can state and prove Higman's trick, which first appeared in~\cite{higman40}.
In the proposition below, it will be useful to have the following notation to denote the 
direct sum of two matrices $A$ and $B$:

$$ A \oplus B = \begin{pmatrix} A & 0 \\ 0 & B \end{pmatrix}, $$

where the zero matrices in the top right and bottom left corners are of appropriate dimensions. Before stating and proving Higman's trick, 
let us work through a small example which showcases the essence of the trick.

Suppose we want to know the $\ncrk$ of matrix $\begin{pmatrix} 1 & x \\ y & z + xy \end{pmatrix}$. The problem here is that this 
matrix is not linear, and we need to have a linear matrix. How can we convert this matrix into a linear matrix while preserving the 
rank, or the complement of the rank? To do this, we need to remove the multiplication happening in $z + xy$.

Notice that the complement of its rank does not change after the following transformation:
$$ \begin{pmatrix} 1 & x \\ y & z + xy \end{pmatrix} \mapsto \begin{pmatrix} 1 & x & 0 \\ y & z + xy & 0 \\ 0 & 0 & 1 \end{pmatrix}. $$
Since the complement of the rank does not change after we perform elementary row or column operations, we can first
add $x \cdot \text{(third row)}$ to the second row, and then subtract $\text{(third column)} \cdot y$ to the second column, to obtain:
$$\begin{pmatrix} 1 & x & 0 \\ y & z + xy & 0 \\ 0 & 0 & 1 \end{pmatrix} \mapsto 
\begin{pmatrix} 1 & x & 0 \\ y & z + xy & x \\ 0 & 0 & 1 \end{pmatrix} \mapsto 
\begin{pmatrix} 1 & x & 0 \\ y & z & x \\ 0 & -y & 1 \end{pmatrix} $$

The complement of the rank of this last matrix is the same as the complement of the rank of our original matrix! In particular, if this
last matrix is full rank, it implies that our original matrix is also full rank. This is the essence of Higman's trick. We now proceed to
its full version.

\begin{proposition}[Effective Higman's Trick]\label{prop:higman-trick}
	Let $A \in \QX^{m \times n}$ be a matrix where each entry $a_{ij}$ is computed by
	a non-commutative formula of size $\le s$ and bit complexity $\le b$ without divisions.
	Let $k$ be the total number of multiplication gates used 
	in the computation of the entries of $A$. There exist matrices $P \in \GL_{m+k}(\QX)$,
	$Q \in \GL_{n+k}(\QX)$ such that  $P (A \oplus I_k) Q$ is a matrix with linear entries
	and coefficients with bit complexity bounded by $b$. Moreover, given access to the formulas computing the 
	entries, one can construct $P$ and $Q$ efficiently in time $\poly(m, n, s, b)$. Since $P$ and $Q$ are non-singular
	matrices, the {\em co-rank} and the {\em co-nc-rank} of $P (A \oplus I_k) Q$ are the same as the {\em co-rank} and
	the {\em co-nc-rank} of $A$.
\end{proposition}

\begin{proof}
	Let $\mult(a_{ij})$ be the number of multiplication gates in the formula computing entry $a_{ij}$ and 
	$$ T = \dst\sum_{\substack{1 \le i \le m \\ 1 \le j \le n}} \mult(a_{ij}). $$ 
	That is, $T$ is the total number of multiplication gates used to compute all entries of the matrix $A$.
	
	We prove this proposition by induction on $T$, for matrices of all dimensions. 
	The base case, when $T = 0$, is trivial, as in this case $A$ itself has linear entries. 
	Suppose now that the proposition is true for all matrices (regardless of their dimensions) 
	which can be computed by formulas using $< T$ multiplication gates. 
	
	Let $A$ be our matrix, which can be computed using $T$ multiplications.
	W.l.o.g., we can assume that $\mult(a_{mn}) \ge 1$.
	Then, by finding a multiplication gate in the formula for $a_{mn}$ that has no other multiplication gate as 
	an ancestor, we can write $a_{mn}$ in the form $a_{mn} = a + b \cdot c$, where 
	$$ \mult(a_{mn}) = \mult(a) + \mult(b) + \mult(c) + 1. $$
	
	Hence, the matrix
	$$ B = \left(I_{m-1} \oplus \begin{pmatrix} 1 & b \\ 0 & 1 \end{pmatrix} \right) (A \oplus 1) 
	\left(I_{n-1} \oplus \begin{pmatrix} 1 & 0 \\ -c & 1 \end{pmatrix} \right) $$  
	is such that 
	$$ b_{ij} = 
		\begin{cases}
			a_{ij}, \text{ if } i \le m, j \le n \text{ and } (i,j) \neq (m,n) \\
			a, \text{ if } (i,j) = (m,n) \\
			b, \text{ if } (i,j) = (m, n+1) \\
			-c, \text{ if } (i,j) = (m+1, n) \\
			1, \text{ if } (i,j) = (m+1, n+1) \\
			0 \text{ otherwise} 
		\end{cases}
	$$
	
	Therefore, the number of multiplications needed to compute $B$ is given by 
	\begin{align*} 
		\dst\sum_{\substack{1 \le i \le m+1 \\ 1 \le j \le n+1}} \mult(b_{ij}) &= 
	\left(\dst\sum_{\substack{1 \le i \le m \\ 1 \le j \le n}} \mult(a_{ij}) \right) - \mult(a_{mn}) + \mult(a) + \mult(b) + \mult(c) \\
	&= T - \mult(a_{mn}) + \mult(a) + \mult(b) + \mult(c) \\
	&= T -1
	 \end{align*}
	 Since $B$ is an $(m+1) \times (n+1)$ matrix which can be computed by using a total of $T-1$ multiplication gates,
	 by the induction hypothesis, there exist $P' \in \GL_{m+1 + (T-1)}(\QX) = \GL_{m+T}(\QX)$ and 
	 $Q' \in \GL_{n+1+(T-1)}(\QX) = \GL_{n+T}(\QX)$ such that $P'(B \oplus I_{T-1})Q'$ is a linear matrix.
	 Since 
	 \begin{align*} 
	 	B \oplus I_{T-1} &=  \left(I_{m-1} \oplus \begin{pmatrix} 1 & b \\ 0 & 1 \end{pmatrix} \oplus I_{T-1} \right) 
						(A \oplus I_T) 
						\left(I_{n-1} \oplus \begin{pmatrix} 1 & 0 \\ -c & 1 \end{pmatrix} \oplus I_{T-1} \right) \\
					&= R(A \oplus I_T)S,
	\end{align*}
	where $R = \left(I_{m-1} \oplus \begin{pmatrix} 1 & b \\ 0 & 1 \end{pmatrix} \oplus I_{T-1} \right) \in \GL_{m+T}(\QX)$ and
	$S = \left(I_{n-1} \oplus \begin{pmatrix} 1 & 0 \\ -c & 1 \end{pmatrix} \oplus I_{T-1} \right) \in \GL_{n+T}(\QX)$,
	we have that
	$$ P'(B \oplus I_{T-1})Q' = (P'R) (A \oplus I_T) (SQ'). $$
	Setting $P = P'R$ and $Q = SQ'$ proves the inductive step and completes the proof. Since we only use subformulas
	of the formulas computing the entries of $A$, the bound on the bit complexity does not change.
\end{proof}

\subsection{Classical Reduction}\label{classical-red}

Having shown the effective version of Higman's trick, we can now compute the $\ncrk$ of a matrix over $\QX$.
We begin with a lemma which will tell us that we can reduce the problem of computing the $\ncrk$ of a matrix
by testing fullness of a smaller matrix with polynomial entries.

\begin{lemma}[Reduction to Fullness Testing]\label{lem:size-reduction}
	Let $M \in \FX^{m \times n}$ be any matrix. In addition, let $U = (u_{ij})$ and $V = (v_{ij})$
	be generic matrices in new, non-commuting variables $u_{ij}$ and $v_{ij}$, of 
	dimensions $r \times m$ and $n \times r$, respectively. Then, $\ncrank(M) \ge r$ iff the matrix $UMV$ is full.
\end{lemma}

\begin{proof} 
	Since $\ncrank(M) \ge r$, there exists an $r \times r$ minor of $M$ of full rank. Let $Q$ be such
	a minor of $M$. W.l.o.g.,\footnote{Notice that we can make the following assumption just to simplify
	notation. In actuality, we do not know where the full rank minor is located in $M$.} 
	we can assume that $Q$ is the $[r] \times [r]$ principal minor of $M$. Hence, we have that
	$$ UMV = \begin{pmatrix} U_1 & U_2 \end{pmatrix} 
	\begin{pmatrix} Q & M_2 \\ M_3 & M_4 \end{pmatrix} 
	\begin{pmatrix} V_1 \\ V_2  \end{pmatrix}, $$
	where $U_1$ and $V_1$ are $r \times r$ matrices and the others are matrices with the proper 
	dimensions. 
	
	Letting $U' = \begin{pmatrix} I_r & 0 \end{pmatrix}$ and 
	$V' = \begin{pmatrix} I_r \\ 0  \end{pmatrix}$, the equality above becomes:
	
	$$ U'MV' = Q. $$
	
	As 
	$$ r \ge \ncrank(UMV) \ge \ncrank(U'MV') = \ncrank(Q) = r, $$ 
	we obtain that $UMV$ is full, as we wanted. Notice that the second inequality 
	comes from the fact that rank does not increase after restrictions of the new variables.
\end{proof}

Notice that we do not know the rank $\ncrank(M)$ a priori. Therefore, our algorithm will
try all possible values of $r \in [n]$ and output the maximum value of $r$
for which we find a full matrix.

For each $r \times r$ matrix $UMV$, we can use the effective
Higman's trick to convert $UMV$ into a $s \times s$ matrix with linear 
entries. With this matrix, we can use the truncated 
Gurvits' algorithm to check whether the matrix we just obtained is full.
Since we have this test, we will be able to output the correct rank. Algorithm~\ref{alg:comp-ncrank}
is the precise formulation of the procedure just described.

\begin{Algorithm}
Input: $M \in \QX^{m \times n}$ s.t. each entry of $M$ is a polynomial 
computed by a formula of size bounded by $s$ and bit complexity bounded by $b$. \\
Output: $\ncrank(M)$ 

\begin{enumerate}
	\item For $1 \le r \le \min(m,n)$, let $U_r$ and $V_r$ be $r \times m$ and $n \times r$ generic matrices in new, non-commuting 
		variables $u_{ij}^{(r)}$, $v_{ij}^{(r)}$.
	\item Let $M_r = U_r M V_r$.
	\item Apply the effective Higman's trick to $M_r$ to obtain a matrix $N_r$ with linear entries on the variables $x_1, \ldots, x_m$
	and $u_{ij}^{(r)}$, $v_{ij}^{(r)}$.
	\item Use Algorithm $G'$ to test whether $N_r$ is full rank.
	\item Output the maximum value of $r$ for which $N_r$ is full rank.
\end{enumerate}
\caption{Noncommutative Rank Algorithm}
\label{alg:comp-ncrank}
\end{Algorithm}

\begin{theorem}\label{thm:comp-ncrank}
	Let $M \in \QX^{m \times n}$ be s.t. each entry of $M$ is a polynomial 
	computed by a formula of size bounded by $s$ and bit complexity bounded by $b$. 
	There exists a deterministic algorithm that finds the non-commutative rank of $M$ in time $\poly(m, n, s, b)$.
\end{theorem}

\begin{proof}
	To prove this theorem, it is enough to show that Algorithm~\ref{alg:comp-ncrank} is correct and it runs
	with the desired runtime.
	
	Without loss of generality, we can assume that $n \le m$. Therefore we have that $\ncrank(M) \le n$.
	By Lemma~\ref{lem:size-reduction}, if $r \le \ncrank(M)$, then matrix $M_r$ will be of full rank (and therefore will not
	have a shrunk subspace, by Theorem~\ref{Equivalences}). Since
	$M_r = U_r M V_r$, from the formulas computing the entries of $M$ we obtain formulas of size at most
	$2smn$ computing the entries of $M_r$. Moreover, the bit complexities of these formulas will still 
	be bounded by $b$, as multiplication by generic matrices do not mix any of the polynomials of $M$.
	
	By Proposition~\ref{prop:higman-trick} and the fact that the size of the formulas computing the entries of $M_r$
	are bounded by $2smn$, we have that $N_r$ is a linear matrix of dimensions $(k+r) \times (k+r)$,
	where $k \le 2s(mn)^2$ and the bit complexity of the coefficients bounded by $b$. 
	Moreover, $N_r = P (M_r \oplus I_{k}) Q$ implies that $N_r$ is full if, and only if, $M_r$ is full, which
	is true if, and only if, $\ncrank(M) \ge r$.
	
	Now, by Theorem~\ref{main}, we have a deterministic polynomial time algorithm  
	to determine whether $N_r$ is full rank.
	If $r \le \ncrank(M)$, $N_r$ will be full, and the maximum 
	such $r$ will be exactly when $r = \ncrank(M)$. Therefore, by outputting the maximum $r$ for which $N_r$ we compute $\ncrank(M)$. This proves that our algorithm is correct. Notice that the runtime is polynomial in the input size, as we perform at most $n$ applications of the Higman trick and of 
	Algorithm $G'$. This completes the proof.
\end{proof}

\subsection{The Quantum Reduction}\label{quantum}

Here we present a different reduction from computing non-commutative rank to fullness testing from a quantum viewpoint. 
We will only work with square matrices though. As we saw, by Higman's trick, we can assume the matrices to be linear. So we are 
given a matrix $L = \sum_{i=1}^m x_i A_i \in M_n(\Fx)$. A combination of Theorems \ref{Equivalences} and \ref{cohn-nc=inner} shows that
$\ncrank(L) \le r$ iff the operator defined by $A_1,\ldots,A_m$ is $n-r$-rank-decreasing. So we just want to check whether a 
completely positive operator is $c$-rank-decreasing and we will do this by using an algorithm for checking if an operator is 
rank-decreasing as a black box using the following lemma:

\begin{lemma}
Let $T : M_n(\mathbb{C}) \rightarrow M_n(\mathbb{C})$ be a completely positive operator. Define an operator $\overline{T} : M_{n+c-1}(\mathbb{C}) \rightarrow M_{n+c-1}(\mathbb{C})$ as follows:
\[ \overline{T} \left(
\begin{bmatrix}
    X_{1,1} & X_{1,2} \\
    X_{2,1} & X_{2,2}
\end{bmatrix} \right)
= 
\begin{bmatrix}
    T(X_{1,1}) + \text{tr}(X_{2,2})I_n & 0 \\
    0 & \text{tr}(X_{1,1})I_{c-1}
\end{bmatrix}
\]
Here $X_{1,1}$, $X_{1,2}$, $X_{2,1}$, $X_{2,2}$ are $n \times n$, $n \times c-1$, $c-1 \times n$, $c-1 \times c-1$ matrices respectively. Then $\overline{T}$ is completely positive and $T$ is $c$-rank-decreasing iff $\overline{T}$ is rank-decreasing. Note that we are considering $c \le n$.
\end{lemma}

\begin{proof} A well known characterization due to Choi \cite{Choi} states that $\overline{T}$ is completely positive iff $\sum_{i,j = 1}^{n+c-1} E_{i,j} \tensor \overline{T}(E_{i,j})$ is psd. Here $E_{i,j}$ is the matrix with $1$ at $i,j$ position and $0$ everywhere else. Now
\[ \overline{T}(E_{i,j}) = \begin{cases} 
      \begin{bmatrix}
    T(E_{i,j})  & 0 \\
    0 & I_{c-1}
\end{bmatrix} & 1 \le i=j \le n \\ \\
      \begin{bmatrix}
    T(E_{i,j})  & 0 \\
    0 & 0
\end{bmatrix} & 1 \le i, j \le n, i \neq j \\ \\
      \begin{bmatrix}
    0  & 0 \\
    0 & 0
\end{bmatrix} & 1 \le i \le n, n+1 \le j \le n+c-1 \: \text{or} \: n+1 \le i \le n+c-1, 1 \le j \le n \\ \\
 \begin{bmatrix}
    I_n  & 0 \\
    0 & 0
    \end{bmatrix} & n+1 \le i=j \le n+c-1 \\ \\
     \begin{bmatrix}
    0  & 0 \\
    0 & 0
    \end{bmatrix} & n+1 \le i, j \le n+c-1, i \neq j \\
   \end{cases}
\]
From here it is easy to verify that $\sum_{i,j = 1}^{n+c-1} E_{i,j} \tensor \overline{T}(E_{i,j})$ is psd given that $\sum_{i,j=1}^n E_{i,j} \tensor T(E_{i,j})$ is psd. Now suppose that $\overline{T}$ is rank-decreasing. This can only happen if $X_{1,1} = 0$ or $X_{2,2} = 0$, otherwise 
\[ \overline{T} \left(
\begin{bmatrix}
    X_{1,1} & X_{1,2} \\
    X_{2,1} & X_{2,2}
\end{bmatrix} \right)
= 
\begin{bmatrix}
    T(X_{1,1}) + \text{tr}(X_{2,2})I_n & 0 \\
    0 & \text{tr}(X_{1,1})I_{c-1}
\end{bmatrix}
\]
is full rank. If $X_{1,1} = 0$, then 
\[ \begin{bmatrix}
    0 & X_{1,2} \\
    X_{2,1} & X_{2,2}
\end{bmatrix}
\]
can be psd (and hermitian) only if $X_{1,2} = X_{2,1} = 0$. In this case a $c-1$ ranked matrix is mapped to rank $n$ matrix. So $X_{2,2}$ has to be zero. Then again by the psd condition $X_{1,2} = X_{2,1} = 0$. So 
\[ \overline{T} \left(
\begin{bmatrix}
    X_{1,1} & 0 \\
    0 & 0
\end{bmatrix} \right)
= 
\begin{bmatrix}
    T(X_{1,1})  & 0 \\
    0 & \text{tr}(X_{1,1})I_{c-1}
\end{bmatrix}
\]
and $X_{1,1} \neq 0$ and 
\[ \text{Rank} \left(
\begin{bmatrix}
    X_{1,1} & 0 \\
    0 & 0
\end{bmatrix} \right)
> 
\text{Rank} \left( \begin{bmatrix}
    T(X_{1,1})  & 0 \\
    0 & \text{tr}(X_{1,1})I_{c-1}
\end{bmatrix} \right)
\]
Hence $\text{Rank}(T(X_{1,1})) \le \text{Rank}(X_{1,1}) - c$. This proves one direction. Now suppose that $T$ is $c$-rank-decreasing and $\text{Rank}(T(X)) \le \text{Rank}(X)-c$, then 
\[ \text{Rank} \left( \overline{T} \left(
\begin{bmatrix}
    X & 0 \\
    0 & 0
\end{bmatrix} \right) \right)
= 
\text{Rank} \left( \begin{bmatrix}
    T(X)  & 0 \\
    0 & \text{tr}(X) I_{c-1} 
\end{bmatrix}  \right)
< 
\text{Rank} \left( \begin{bmatrix}
    X  & 0 \\
    0 & 0 
\end{bmatrix}  \right)
\]
This proves the lemma.
\end{proof}

\begin{remark} This seems to be the ``quantum" analogue of obtaining a maximum matching oracle based on a perfect matching oracle: add c-1 dummy vertices to both sides of the bipartite graph and connect them to everything. Then the new graph has a perfect matching iff the original graph had a matching of size $\ge n - c + 1$.
\end{remark}

\begin{remark}
Here we didn't specify a set of Kraus operators for the operator $\overline{T}$ which seem to be needed to run Algorithms \ref{Gurvits_alg} and \ref{Gurvits_alg_trn} but Kraus operators can be obtained by looking at the eigenvectors of $\sum_{i,j=1}^{n+c-1} E_{i,j} \tensor \overline{T}(E_{i,j})$. Alternatively Algorithms \ref{Gurvits_alg} and \ref{Gurvits_alg_trn} can also be interpreted as acting directly on the Choi-Jamiolkowski state of $\overline{T}$ i.e. $\sum_{i,j=1}^{n+c-1} E_{i,j} \tensor \overline{T}(E_{i,j})$. 
\end{remark}